%% file: main.tex
\documentclass[acmsmall,screen]{acmart}
\setcopyright{None}
\acmPrice{}
\acmDOI{10.1145/3371080}
\acmYear{2020}
\copyrightyear{2020}
\acmJournal{PACMPL}
\acmVolume{4}
\acmNumber{POPL}
\acmArticle{12}
\acmMonth{1}

\usepackage{booktabs}
\usepackage{amsthm,mathtools,amssymb,stmaryrd}
\usepackage{amsfonts}
\usepackage{graphicx}
\usepackage{listings}
\usepackage{url}
\usepackage{hyperref}
\usepackage{xspace}
\usepackage{multirow}
\usepackage{textcomp}
\usepackage[T1]{fontenc}
\usepackage[scaled=0.75]{beramono}
\usepackage{paralist}
\usepackage[inference]{semantic}
\usepackage{algorithmicx}
\usepackage[noend]{algpseudocode}
\usepackage{wrapfig}
\usepackage{pdfpages}

\hypersetup{colorlinks,
  linkcolor=ACMDarkBlue,
  citecolor=ACMPurple,
  urlcolor=ACMDarkBlue,
  filecolor=ACMDarkBlue}

\usepackage[version=0.96]{pgf}
\usepackage{tikz}
\usetikzlibrary{arrows,shapes,automata,backgrounds,petri,positioning}
\usepackage{todonotes}

\usepackage{ltablex,array}
\input{defs}

\begin{document}

\title{Program Synthesis by Type-Guided Abstraction Refinement}

\author{Zheng Guo}
\affiliation{
  \institution{UC San Diego}
  \country{USA}
}
\email{zhg069@ucsd.edu}

\author{Michael James}
\affiliation{
  \institution{UC San Diego}
  \country{USA}
}
\email{m3james@ucsd.edu}

\author{David Justo}
\affiliation{
  \institution{UC San Diego}
  \country{USA}
}
\email{djusto@ucsd.edu}

\author{Jiaxiao Zhou}
\affiliation{
  \institution{UC San Diego}
  \country{USA}
}
\email{jiz417@ucsd.edu}

\author{Ziteng Wang}
\affiliation{
  \institution{UC San Diego}
  \country{USA}
}
\email{ziw329@ucsd.edu}

\author{Ranjit Jhala}
\affiliation{
  \institution{UC San Diego}
  \country{USA}
}
\email{jhala@cs.ucsd.edu}

\author{Nadia Polikarpova}
\affiliation{
  \institution{UC San Diego}
  \country{USA}
}
\email{npolikarpova@ucsd.edu}

\renewcommand{\shortauthors}{Z. Guo et al.}

\begin{abstract}
\input{abstract}
\end{abstract}



\maketitle

\input{intro}

\input{examples}

\input{checking}
\input{synthesis}

\input{implementation}

\input{evaluation}

\input{related}
\input{conclusions}

\begin{acks}                            
  The authors would like to thank Neil Mitchell for providing the \hoogle data
  and helpful feedback on the \tool web interface.
  We are also grateful to the anonymous reviewers of this and older versions of the paper 
  for their careful reading and many constructive suggestions.
\end{acks}

\bibliography{main}

\newpage
\appendix
\input{appendix}

\end{document}

%% file: defs.tex
\newcommand{\etc}{\emph{etc}\xspace}

\newcommand{\ie}{\emph{i.e.}\xspace}

\newcommand{\eg}{\emph{e.g.}\xspace}

\newcommand{\etal}{\emph{et~al.}\xspace}
\newcommand{\tygar}{\tname{TYGAR}}
\newcommand{\emphbf}[1]{\emph{\textbf{#1}\xspace}}
\newcommand{\mypara}[1]{\smallskip\noindent\emphbf{#1.}\xspace}
\newcommand{\setof}[1]{\ensuremath{\{ #1 \}}}
\newcommand{\sem}[1]{\ensuremath{\llbracket #1 \rrbracket}}

\newcommand{\mFixType}{\emph{Baseline}\xspace}
\newcommand{\mNoGar}{\emph{NOGAR}\xspace}
\newcommand{\mTopType}{\emph{\tygar-0}\xspace}
\newcommand{\mQryType}{\emph{\tygar-Q}\xspace}
\newcommand{\mQryTypeBounded}{\emph{\tygar-QB}\xspace}

\newif\ifdraft
\drafttrue

\ifdraft
\newcommand\authorrnote[3]{\textcolor{#1}{({#2}: {#3})}\xspace}
\newcommand{\TODO}[1]{{\color{orange!80!black}[\textsl{#1}]}}
\else
\newcommand\authorrnote[2]{}
\newcommand{\TODO}[1]{}
\fi

\newcommand{\nex}{\ensuremath{E}}
\newcommand{\tname}[1]{\textsc{#1}\xspace}
\newcommand{\tool}{\tname{H+}}
\newcommand{\sys}{\tool}
\newcommand{\hoogle}{\tname{Hoogle}}
\newcommand{\sover}{\tname{StackOverflow}}
\newcommand{\agda}{\tname{Agda}}
\newcommand{\sypet}{\tname{SyPet}}
\newcommand{\insynth}{\tname{InSynth}}
\newcommand{\prospector}{\tname{Prospector}}
\newcommand{\blaze}{\tname{Blaze}}
\newcommand{\djinn}{\tname{Djinn}}
\newcommand{\synquid}{\tname{Synquid}}
\newcommand{\myth}{\tname{Myth2}}
\newcommand{\corelang}{$\lambda_H$\xspace}

\definecolor{commentgreen}{rgb}{0.25,0.5,0.35}

\lstdefinestyle{numbers}
{
  numbers=left,
  numberstyle=\sf,
  xleftmargin=15pt
}

\newcommand{\T}[1]{\mbox{\lstinline[language=Haskell,basicstyle=\ttfamily,columns=fixed]^#1^}}

\tikzstyle{place}=[circle,thick,draw=blue!75,fill=blue!20,minimum size=6mm]
\tikzstyle{final}=[place,double]
\tikzstyle{other}=[place,draw=red!75,fill=red!20]
\tikzstyle{blank}=[minimum size=6mm]
\tikzstyle{transition}=[rectangle,thick,draw=black!75,fill=black!20,minimum size=4mm]
\tikzstyle{sol}=[very thick]
\tikzstyle{mismatch}=[thick,red!75]

\newcommand{\labels}[2]{{\begin{tabular}{c} \textcolor{blue}{#1} \\ \textcolor{red}{#2} \end{tabular}}}
\newcommand{\labelc}[1]{{\textcolor{red}{#1}}}

\newcommand*\Let[2]{\State #1 $\gets$ #2}
\algrenewcommand\algorithmicrequire{\textbf{Input:}}
\algrenewcommand\algorithmicensure{\textbf{Output:}}
\algnewcommand{\LineFor}[2]{
  \State\algorithmicfor\ {#1}\ \algorithmicdo\ {#2}}

\newcommand{\abssynth}{SynAbstract}
\newcommand{\shortestpath}{ShortestValidPath}

\newcommand{\rng}[1]{\mathsf{range}(#1)}
\newcommand{\subst}[2]{{#1} \mapsto {#2}}
\newcommand{\pre}{\mathsf{pre}}
\newcommand{\post}{\mathsf{post}}
\newcommand{\parents}{\mathsf{parents}}
\newcommand{\many}[1]{\overline{#1}}

\newcommand{\ty}[1]{\Lambda(#1)}
\newcommand{\abset}{\mathcal{A}}
\newcommand{\wrap}[2]{\mathsf{wrap}(#1,#2)}
\newcommand{\nat}{\mathbb{N}}
\newcommand{\net}{\mathcal{N}}
\newcommand{\steps}[1]{\xrightarrow{#1}}
\newcommand{\toBase}[1]{\mathsf{base}(#1)}
\newcommand{\fo}{\mathsf{desugar}}
\newcommand{\hask}[1]{\tilde{#1}}
\newcommand{\fromPath}[1]{\mathsf{terms}(#1)}
\newcommand{\fresh}[1]{\mathsf{fresh}(#1)}
\newcommand{\subterms}[1]{\mathsf{subterms}(#1)}
\newcommand{\close}[1]{\mathsf{close}(#1)}

\newcommand{\refines}{\preceq}
\newcommand{\inv}{\mathcal{I}}
\newcommand{\sub}{\sqsubseteq}
\newcommand{\meet}{\sqcap}
\newcommand{\join}{\sqcup}
\newcommand{\mgu}[1]{\mathsf{mgu}(#1)}
\newcommand{\bases}{\ensuremath{\mathbf{B}}}
\newcommand{\basebots}{\ensuremath{\mathbf{B_{\bot}}}}
\newcommand{\anew}{A_{\mathit{new}}}

\newcommand{\eapp}[2]{{#1}\ {#2}}
\newcommand{\elam}[2]{\lambda {#1}.{#2}}
\newcommand{\ebody}{e_{\mathit{body}}}
\newcommand{\sapp}[2]{{#1}{#2}}

\newcommand{\memph}[1]{\textcolor{red}{#1}}
\newcommand{\mnoemph}[1]{\textcolor{black}{#1}}

\newcommand{\chkarrow}{\Longleftarrow}
\newcommand{\synarrow}{\Longrightarrow}

\newcommand{\jtyping}[3]{\ensuremath{\Lambda;#1 \vdash #2 :: #3}}
\newcommand{\jtcheck}[3]{\ensuremath{\Lambda;#1 \vdash #2 \chkarrow #3}}
\newcommand{\jtinfer}[3]{\ensuremath{\Lambda;#1 \vdash #2 \synarrow #3}}

\newcommand{\jntyping}[3]{\ensuremath{\Lambda;#1 \not\vdash #2 \chkarrow #3}}

\newcommand{\jacheck}[4]{\ensuremath{\Lambda;#1 \vdash_{#4} #2 \chkarrow #3}}
\newcommand{\jainfer}[4]{\ensuremath{\Lambda;#1 \vdash_{#4} #2 \synarrow #3}}
\newcommand{\jancheck}[4]{\ensuremath{\Lambda;#1 \not\vdash_{#4} #2 \chkarrow #3}}

\newcommand{\jatyping}[4]{\Lambda;#1 \vdash_{#4} #2 :: #3}

\newcommand{\componentCount}{291\xspace}
\newcommand{\benchmarkCount}{44\xspace}
\newcommand{\hoogleOnlyBms}{24\xspace}

\newcommand{\ourBms}{17\xspace}

\newcommand{\timeout}{60\xspace}
\newcommand{\qualitytimeout}{100\xspace}

\newcommand{\qualityH}{\T{H+}\xspace}
\newcommand{\qualityHD}{\T{H-D}\xspace}
\newcommand{\qualityHR}{\T{H-R}\xspace}

\newcommand{\firstSolutionTime}{1.4\xspace}

\newcommand{\tygarQBSolnCount}{43\xspace}
\newcommand{\tygarQSolnCount}{34\xspace}
\newcommand{\tygarZSolnCount}{35\xspace}
\newcommand{\nogarSolnCount}{37\xspace}

\newcommand{\applyNTimesRank}{10\xspace}
\newcommand{\applyNTimesPosition}{11\xspace}
\newcommand{\firstJustRank}{18\xspace}
\newcommand{\firstJustPosition}{19\xspace}
\newcommand{\lookupRank}{1\xspace}
\newcommand{\lookupPosition}{18\xspace}

\newcommand{\interestingCount}{32\xspace}

%% file: abstract.tex
We consider the problem of type-directed component based synthesis
where, given a set of (typed) components and a query \emph{type},
the goal is to synthesize a \emph{term} that inhabits the query.
Classical approaches based on proof search in intuitionistic logics
do not scale up to the standard libraries of modern languages, which
span hundreds or thousands of components.
Recent graph reachability based methods proposed for languages like Java
do scale, but only apply to monomorphic data and components:
polymorphic data and components infinitely explode the size of
the graph that must be searched, rendering synthesis intractable.
We introduce \emph{type-guided abstraction refinement} (\tygar),
a new approach for scalable type-directed synthesis over
polymorphic datatypes and components.
Our key insight is that we can overcome the explosion
by building a graph over \emph{abstract types}
which represent a potentially unbounded set of concrete types.
We show how to use graph reachability to search
for candidate terms over abstract types, and introduce
a new algorithm that uses \emph{proofs of untypeability}
of ill-typed candidates to iteratively \emph{refine}
the abstraction until a well-typed result is found.
We have implemented \tygar in \tool, a tool that
takes as input a set of Haskell libraries and a
query type, and returns a Haskell term that uses
functions from the provided libraries to implement
the query type.
We have evaluated \tool on a set of \benchmarkCount queries using a set
of popular Haskell libraries with a total of
\componentCount components.
Our results demonstrate that \tool returns an interesting solution within the first five results
for \interestingCount out of \benchmarkCount queries.
Moreover, \tygar allows \tool
to rapidly return well-typed terms,
with the median time to first solution of just \firstSolutionTime seconds.
%

%% file: intro.tex
\section{Introduction}\label{sec:intro}

Consider the task of implementing a function \T{firstJust def mbs},
which extracts the first non-empty value from a list of options \T{mbs},
and if none exists, returns a default value \T{def}.
Rather than writing a recursive function,
you suspect you can implement it more concisely 
and idiomatically using components from a standard library.
If you are a Haskell programmer, at this point you 
will likely fire up Hoogle~\cite{Hoogle},
the Haskell's API search engine, and query it with 
the intended type of \T{firstJust}, \ie 
\T{a -> [Maybe a] -> a}.
The search results will be disappointing, however, 
since no single API function matches this type%
\footnote{We tested this query at the time of writing with the default Hoogle configuration (Hoogle 4).}.
In fact, to implement \T{firstJust} you need 
a snippet that composes three library functions 
from the standard \T{Data.Maybe} library, like so:
\T{\\def mbs -> fromMaybe def (listToMaybe (catMaybes mbs))}.
Wouldn't you like a tool that could automatically 
synthesize such snippets from type queries?

\mypara{Scalable Synthesis via Graph Reachability}
In general, our problem of type-directed 
\emph{component-based synthesis}, reduces 
to that of finding inhabitants for a given 
query type \cite{urzyczyn97}.
Consequently, one approach is to develop synthesizers 
based on proof search in intuitionistic logics \cite{djinn}.
However, search becomes intractable in the presence of 
libraries with hundreds or thousands of components.
Several papers address the issue of scalability 
by rephrasing the problem as one of reachability 
in a \emph{type transition network} (TTN), \ie a graph 
that encodes the library of components.
Each type is represented as a \emph{state}, and each 
component is represented as a directed \emph{transition} 
from the component's input type to its output type.
The synthesis problem then reduces to finding 
a \emph{path} in the network that begins at the 
query's input type and ends at the output 
type~\cite{Mandelin05}.
To model components (functions) that take 
\emph{multiple} inputs, we need only generalize 
the network to a \emph{Petri net} which has 
\emph{hyper}-transitions that link multiple 
input states with a single output. 
With this generalization, the synthesis problem 
can, once again, be solved by finding a path from 
the query's input types to the desired output 
yielding a scalable synthesis method for 
Java~\cite{FengM0DR17}.

\mypara{Challenge: Polymorphic Data and Components}
Graph-based approaches crucially rely on the assumption
that the size of the TTN is \emph{finite} (and manageable).
This assumption breaks down in the presence of \emph{polymorphic} 
components that are ubiquitous 
in libraries for modern functional languages.
\begin{inparaenum}[(a)]
\item With polymorphic \emph{datatypes} the set of 
      types that might appear in a program is unbounded: 
      for example, two type constructors \T{[]} and \T{Int} 
      give rise to an \emph{infinite} set of types 
      (\T{Int}, \T{[Int]}, \T{[[Int]]}, \etc).
\item Even if we bound the set of types, polymorphic \emph{components} 
      lead to a combinatorial explosion in the number of transitions:
      for example, the pair constructor with the type 
      \T{a -> b -> (a,b)} creates a transition from 
      \emph{every pair of types} in the system.
\end{inparaenum}
In other words, polymorphic data and components 
explode the size of the graph that must be searched, 
rendering synthesis intractable.

%

\mypara{Type-Guided Abstraction Refinement}
In this work we introduce \emph{type-guided abstraction refinement} (\tygar),
a new approach to scalable type-directed synthesis over polymorphic datatypes and 
components.
A high-level view of \tygar is depicted in \autoref{fig:workflow}.
The algorithm maintains an \emph{abstract transition network} (ATN) 
that finitely \emph{overapproximates} the infinite network 
comprising all monomorphic instances of the polymorphic data and components.
We use existing SMT-based techniques to find a suitable path in the compact ATN,
which corresponds to a candidate term.
If the term is well-typed, it is returned as the solution.
Due to the overapproximation, however, the ATN can contain \emph{spurious} paths,
which correspond to ill-typed terms. 
In this case, the ATN is \emph{refined} in order to exclude this spurious path,
along with similar ones.
We then repeat the search with the refined ATN until a well-typed solution is found.
As such, \tygar extends \emph{synthesis using abstraction refinement} (SYNGAR)~\cite{WangDS18},
from the domain of values to the domain of types.
\tygar's support for polymorphism also allows us to handle 
\emph{higher-order} components, which take functions as input, 
by representing functions (arrows) as a binary type 
constructor.
Similarly, \tygar can handle 
Haskell's ubiquitous \emph{type classes}, 
by following the dictionary-passing translation~\cite{WadlerB89}, 
which again, relies crucially on support for parametric polymorphism.
%
%

\begin{figure}[t]
\includegraphics[width=0.8\textwidth]{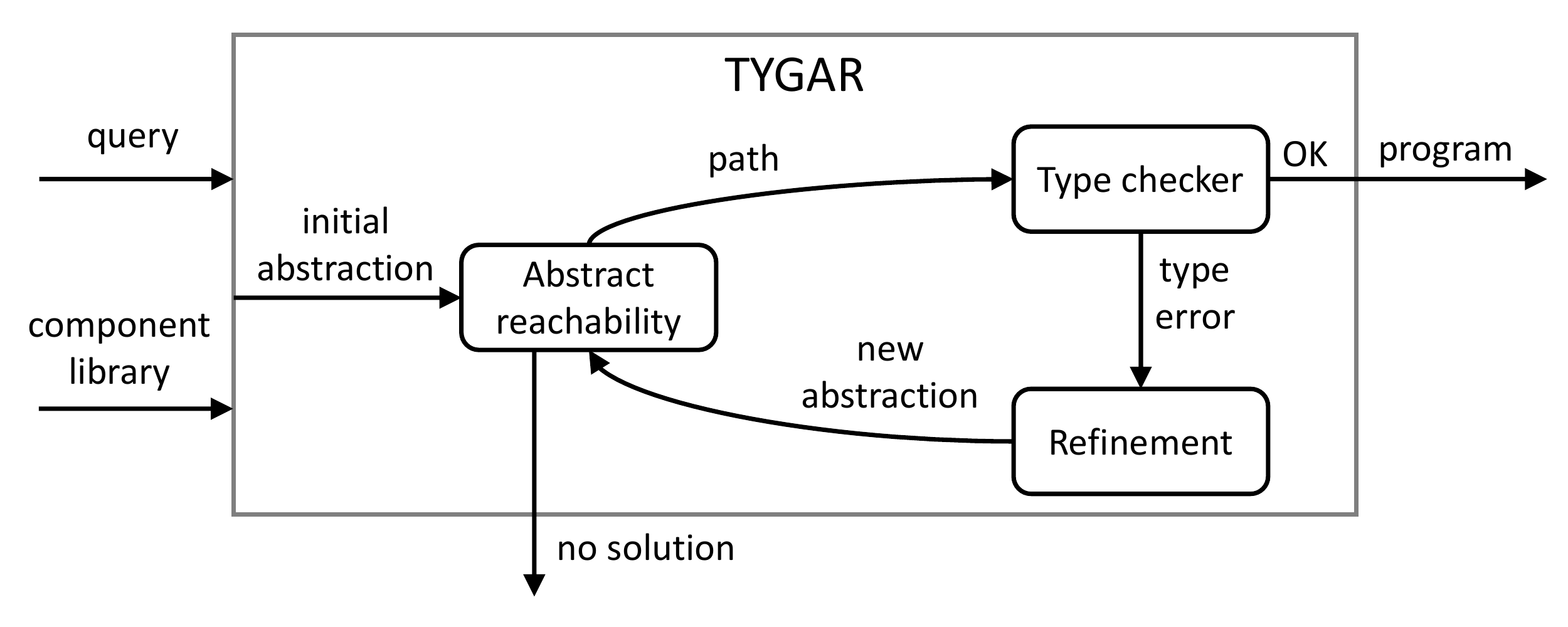}
\caption{Overview of the TYGAR synthesis algorithm.}\label{fig:workflow}
\end{figure}

\mypara{Contributions} In summary, this paper makes the following contributions:

\mypara{1. Abstract Typing} 
Our first contribution is a novel notion of 
\emph{abstract typing} grounded in the framework 
of abstract interpretation~\cite{Cousot1977}.
Our abstract domain is parameterized by a finite 
collection of polymorphic types, each of which 
abstracts a potentially infinite set of ground 
instances.
Given an abstract domain, we automatically derive 
an over-approximate type system, which we use to 
build the ATN.
This is inspired by predicate abstraction 
\cite{GrafSaidi97}, where the abstract domain 
is parameterized by a set of predicates,
and abstract program semantics at different 
levels of detail can be derived automatically 
from the domain.

\mypara{2. Type Refinement}
Our second contribution is a new algorithm that, given a spurious program, 
refines the abstract domain so that the program no longer type-checks abstractly.
To this end, the algorithm constructs a compact \emph{proof of untypeability} of the program:
it annotates each subterm with a type that is just precise enough to refute the program.

\mypara{3. \tool}
Our third contribution is an implementation of \tygar in \tool, 
a tool that takes as input a set of Haskell libraries and a type, 
and returns a ranked list of straight-line programs that have the desired type 
and can use any function from the provided libraries.
%
To keep in line with \hoogle's user interaction model 
familiar to Haskell programmers, \tool does not require 
any user input beyond the query type; this is in contrast 
to prior work on component-based synthesis~\cite{FengM0DR17,Shi2019}, 
where the programmer provides input-output examples 
to disambiguate their intent.
This setting poses an interesting challenge:
given that there might be hundreds of programs 
of a given type (including nonsensical ones 
like \T{head []}), how do we select just the 
\emph{relevant} programs, likely to be useful 
to the programmer?
We propose a novel mechanism for 
filtering out irrelevant programs
using GHC's \emph{demand analysis}~\cite{SergeyVJB17} 
to eliminate terms where some of the inputs 
are unused.
We have evaluated \tool on a set of \benchmarkCount queries collected 
from different sources (including \hoogle and \sover),
using a set of popular Haskell libraries with a total of 
\componentCount components.
Our evaluation shows that \tool is able to find a well-typed program 
for \tygarQBSolnCount out of \benchmarkCount queries
within the timeout of \timeout seconds.
It finds the first well-typed program 
within \firstSolutionTime seconds on average.
In \interestingCount out of \benchmarkCount queries, 
the top five results contains a useful solution%
\footnote{Unfortunately, ground truth solutions are not available 
for \hoogle benchmarks; we judge usefulness by manual inspection.}.
%
Further, our evaluation demonstrates that 
both abstraction and refinement are important for efficient synthesis.
A naive approach that \emph{does not use abstraction} and instead instantiates all polymorphic 
datatypes up to even a small depth of 1 yields a 
massive transition network, 
and is unable to solve any benchmarks within the timeout.
On the other hand, an approach that uses a fixed small ATN but \emph{no refinement}
works well on simple queries, 
but fails to scale as the solutions get larger.
Instead, the best performing search algorithm uses \tygar to start with a small initial ATN
and gradually extend it, up to a given size bound, with instances that are relevant for a given synthesis query.


%

%% file: examples.tex
\section{Background and Overview}\label{sec:examples}

We start with some examples that illustrate the
prior work on component-based synthesis that \sys
builds on (\autoref{sec:examples:sypet}),
the challenges posed by polymorphic components,
and our novel techniques for addressing those challenges.

\subsection{Synthesis via Type Transition Nets}\label{sec:examples:sypet}

The starting point of our work is \sypet~\cite{FengM0DR17},
a component-based synthesizer for Java.
Let us see how \sypet works by using the example query from
the introduction: \T{a -> [Maybe a] -> a}.
For the sake of exposition, we assume that our library
only contains three components listed in \autoref{fig:sypet} (left).
Hereafter, we will use Greek letters $\alpha, \beta, \ldots$
to denote \emph{existential type variables}---\ie the type
variables of components, which have to be instantiated by
the synthesizer---as opposed to $a, b, \ldots$ for
\emph{universal type variables} found in the \emph{query},
which, as far as the synthesizer is concerned, are just nullary type constructors.
Since \sypet does not support polymorphic components,
let us assume for now that an oracle provided
us with a small set of monomorphic types that suffice to answer this query,
namely, \T{a}, \T{Maybe a}, \T{Maybe (Maybe a)}, \T{[a]}, and \T{[Maybe a]}.
For the rest of this section, we abbreviate the names
of components and type constructors to their first letter
(for example, we will write \T{L (M a)} for \T{[Maybe a]})
and refer to the query arguments as \T{x1}, \T{x2}.

\begin{figure}
\begin{minipage}{.45\textwidth}
\begin{lstlisting}
-- | Value stored in the option
-- or default if the option is empty
fromMaybe :: _a -> Maybe _a -> _a
-- | All values from a list of options
catMaybes :: List (Maybe _a) -> List _a
-- | Head of the list
-- or empty option if the list is empty
listToMaybe :: List _a -> Maybe _a
\end{lstlisting}
\end{minipage}
\begin{minipage}{.45\textwidth}
\centering
\input{petri_concrete}
\end{minipage}
\caption{(left) A tiny component library.
(right) A Type Transition Net for this library and query \T{a -> List (Maybe a) -> a}.
The transitions \T{l<a>}, \T{f<a>} (resp. \T{l<M a>}, \T{f<M a>}) correspond to 
the polymorphic instances of the components \T{listToMaybe}, \T{fromMaybe} at 
type \T{a} (resp. \T{M a}).}
\label{fig:sypet}
\end{figure}
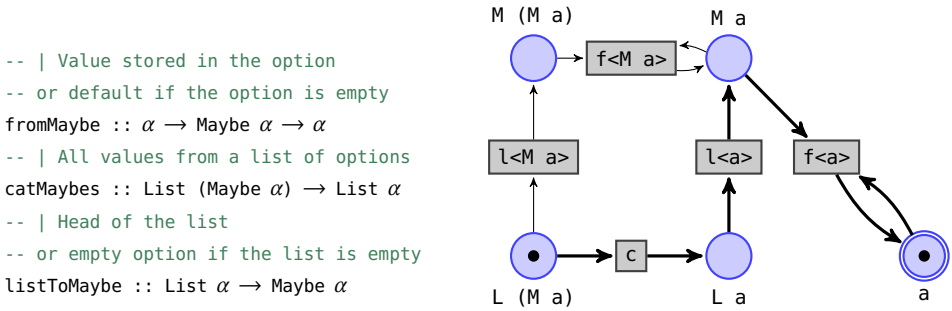

\mypara{Components as Petri Nets}
\sypet uses a Petri-net representation
of the search space, which we
refer to as the \emph{type transition net} (TTN).
The TTN for our running example is shown
in~\autoref{fig:sypet} (right).
Here \emph{places} (circles) correspond to types,
\emph{transitions} (rectangles) correspond to components,
and \emph{edges} connect components with their input and output types.
Since a component might require multiple inputs
of the same type, edges can be annotated with
\emph{multiplicities} (the default multiplicity is 1).
A \emph{marking} of a TTN assigns a non-negative
number of \emph{tokens} to every place.
The TTN can step from one marking to the next
by \emph{firing} a transition: if the input
places of a transition have sufficiently many tokens, 
the transition can fire, consuming those input
tokens and producing a token in the output place.
For example, given the marking in \autoref{fig:sypet},
transition \T{c} can fire, consuming the token
in \T{L (M a)} and producing one in \T{L a}; however
transition \T{f<a>} cannot fire as there is no token in \T{M a}.

\mypara{Synthesis via Petri-Net Reachability}
Given a synthesis query $T_1 \to \ldots \to T_n \to T$,
we set the \emph{initial marking} of the TTN to contain
one token for each input type $T_i$,
and the \emph{final marking} to contain a single token
in the type $T$.
The synthesis problem then reduces to finding
a \emph{valid path}, \ie a sequence of fired
transitions that gets the net from the initial
marking to the final marking.
\autoref{fig:sypet} shows the initial marking
for our query, and also indicates the final
marking with a double border around the return
type \T{a} (recall that the final marking
of a TTN always contains a single token
in a given place).
The final marking is reachable via the path $[\T{c}, \T{l}, \T{f}]$,
marked with thick arrows,
which corresponds to a well-typed program \T{f x1 (l (c x2))}.
In general, a path might correspond to multiple programs---%
if several tokens end up in the same place at any point along the path---%
of which at least one is guaranteed to be well-typed;
the synthesizer can then find the well-typed program using explicit or symbolic enumeration.

\subsection{Polymorphic Synthesis via Abstract Type Transition Nets}\label{sec:examples:abstract}

Libraries for modern languages like Haskell
provide highly polymorphic components that
can be used at various different instances.
For example, our universe contains three type
constructors---\T{a}, \T{L}, and \T{M}---%
which can give rise to infinitely many types,
so creating a place for each type is out of
question.
Even if we limit ourselves to those constructors
that are reachable from the query types by
following the components, we might still end
up with an infinite set of types: for example,
following \T{head :: List _a -> _a} backwards from \T{a} yields
\T{L a}, \T{L (L a)}, and so on.
This poses a challenge for Petri-net based
synthesis: \emph{which finite set of (monomorphic)
instances do we include in the TTN?}

On the one hand, we have to be careful not to include \emph{too many} instances.
In the presence of polymorphic components,
these instances can explode the number of transitions.
%
\autoref{fig:sypet} illustrates this for the \T{f} and \T{l} components,
each giving rise to two transitions,
by instantiating their type variable $\alpha$ with two different TTN places, \T{a} and \T{Maybe a}.
This proliferation of transitions is especially severe
for components with multiple type variables.
On the other hand, we have to be careful not to include \emph{too few}
instances. We cannot, for example, just limit ourselves
to the monomorphic types that are explicitly present in
the query (\T{a} and \T{L (M a)}), as this will preclude
the synthesis of terms that generate intermediate values
of some other type, \eg \T{L a} as returned by the component
\T{c}, thereby preventing the synthesizer from
finding solutions.

\mypara{Abstract Types}
To solve this problem, we introduce the notion of an \emph{abstract type}%
\footnote{Not to be confused with existing notions of \emph{abstract data type} and \emph{abstract class}.
We use ``abstract'' here is the sense of abstract interpretation~\cite{Cousot1977},
\ie an abstraction of a set of concrete types.},
which stands for (infinitely) many monomorphic instances.
We represent abstract types simply as polymorphic types,
\ie types with free type variables.
For example, the abstract type $\tau$ stands for the set of all types,
while \T{L _t} stands for the set $\{\T{L}~t \mid t\in Type\}$.
This representation supports different levels of detail:
for example, the type \T{L (M a)}
can be abstracted into itself, \T{L (M _t)}, \T{L _t}, or \T{_t}.

\mypara{Abstract Transition Nets}
A Petri net constructed out of abstract types,
which we dub an \emph{abstract transition net} (ATN),
can finitely represent all types in our universe,
and hence all possible solutions to the synthesis
problem.
%
The ATN construction is grounded in the theory of abstract
interpretation and ensures that the net \emph{soundly over-approximates}
the concrete type system, \ie that every well-typed program corresponds
to some valid path through the ATN.
\autoref{fig:abs-ref} (2) shows the ATN for our
running example with places \T{_t}, \T{L _t} and \T{a}.
In this ATN, the \emph{rightmost} \T{f} transition
takes \T{a} and \T{_t} as inputs and returns \T{a} as output.
This transition represents the set of monomophic types $\{\T{a} \to t \to \T{a} \mid t \in Type\}$
and \emph{over-approximates} the set of instances of \T{f}
where the first argument unifies with \T{a} and the second argument unifies with \T{_t}
(which in this case is a singleton set $\{\T{a -> M a -> a}\}$).
Due to the over-approximation, some of the ATN's paths
yield \emph{spurious} ill-typed solutions.
For example, via the highlighted path, this ATN
produces the term \T{f x1 (l x2)}, which is ill-typed
since the arguments to \T{f} have the types \T{a}
and \T{M (M a)}.

How do we pick the right level of detail for the ATN?
If the places are too abstract,
there are too many spurious solutions,
leading, in the limit, to a brute-force enumeration of programs.
If the places are too concrete, the net becomes too large,
and the search for valid paths is too slow.
Ideally, we would like to pick a minimal set of
abstract types that only make distinctions pertinent to the query at hand.

\begin{figure}
\input{petri_abstract}

\input{type_errors}
\caption{Three iterations of abstraction refinement: ATNs (above)
and corresponding solutions (below).
Some irrelevant transitions are omitted from the ATNs for clarity.
Solutions 1 and 2 are spurious, solution 3 is valid.
Each solution is annotated with its concrete typing (in red);
each spurious solution is additionally annotated with its proof
of untypeability (in blue).
These blue types are added to the ATN in the next iteration.
}\label{fig:abs-ref}
\end{figure}
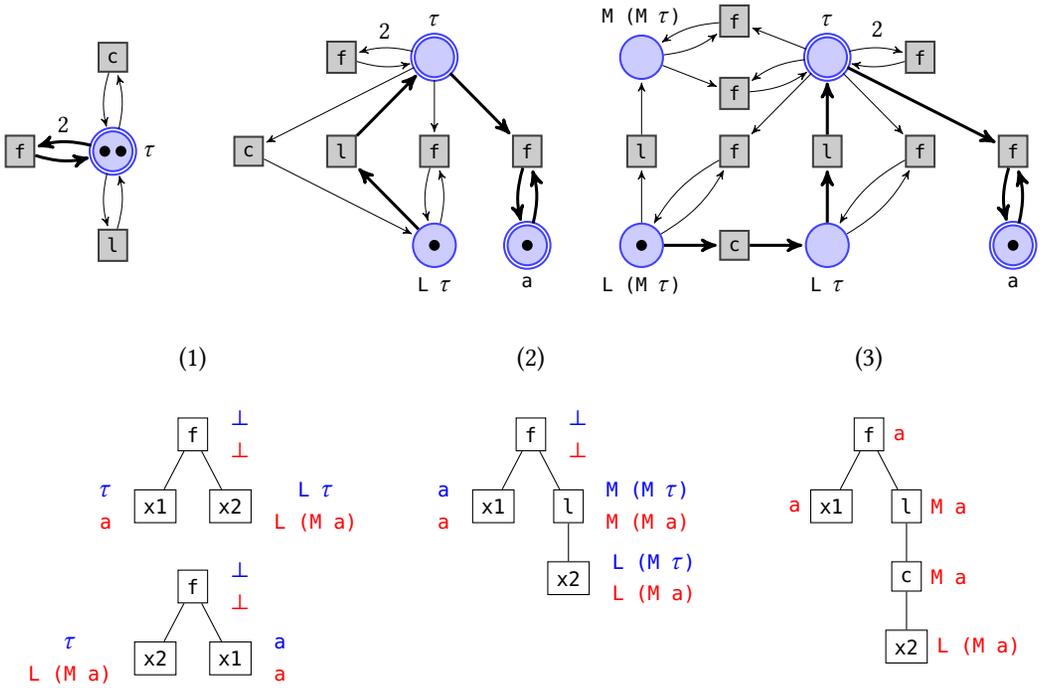

\mypara{Type-Guided Abstraction Refinement}
\tool solves this problem using an iterative process we call
\emph{type-guided abstraction refinement} (\tygar)
where an initial coarse abstraction is incrementally refined
using the information from the type errors found in spurious solutions.
Next, we illustrate \tygar using the running example
from \autoref{fig:abs-ref}.

\paragraph{Iteration 1}
We start with the coarsest possible abstraction,
where all types are abstracted to \T{_t}, yielding
the ATN in \autoref{fig:abs-ref} (1).
%
The shortest valid path is just $[\T{f}]$,
which corresponds to two programs: \T{f x1 x2} and \T{f x2 x1}.
Next, we type-check these programs to determine whether they are valid or spurious.
During type checking, we compute the principal type of each sub-term
and propagate this information bottom-up through the AST;
the resulting \emph{concrete typing} is shown in red at the bottom of \autoref{fig:abs-ref} (1).
Since both candidate programs are ill-typed
(as indicated by the annotation $\bot$ at the root of either AST),
the current path is spurious.
Although we could simply enumerate more valid paths
until we find a well-typed program,
such brute-force enumeration does not scale with the number of components.
Instead, we refine the abstraction so that this path
(and hopefully many similar ones) becomes invalid.

Our refinement uses the type error information
obtained while type-checking the spurious programs.
Consider \T{f x1 x2}:
the program is ill-typed because the concrete type of \T{x2}, \T{L (M a)},
does not unify with the second argument of \T{f}, \T{M _a}.
To avoid making this type error in the future,
we need to make sure that the \emph{abstraction} of \T{L (M a)}
also fails to unify with \T{M _a}.
To this end, we need to extend our ATN with new 
abstract types, that suffice to reject the 
program \T{f x1 x2}.
These new types will update the ATN with new 
places that will \emph{reroute} the transitions 
so that the path that led to the term \T{f x1 x2}
is no longer feasible.
We call this set of abstract types 
a \emph{proof of untypeability} of 
the program.
We could use \T{x2}'s concrete type 
\T{L (M a)} as the proof, but we want 
the proof to be as general as possible,
so that it can reject more programs.
To compute a better proof, the \tygar 
algorithm \emph{generalizes} the concrete 
typing of the spurious program, repeatedly 
weakening concrete types with fresh variables 
while still preserving untypeability.
In our example, the generalization step 
yields \T{_t} and \T{L _t} (see blue 
annotations in \autoref{fig:abs-ref}).
This general proof also rejects other programs 
that use a list as the second argument to \T{f},
such as \T{f x1 (c x2)}.
Adding the types from the untypeability 
proofs of both spurious programs to the 
ATN results in a refined net shown in 
\autoref{fig:abs-ref} (2).


\paragraph{Iteration 2}
The new ATN in \autoref{fig:abs-ref} (2) has no valid paths of length one,
but has the (highlighted) path $[\T{l}, \T{f}]$ of length two,
which corresponds to a single program \T{f x1 (l x2)}
(since the two tokens never cross paths).
This program is ill-typed,
so we refine the abstraction based on its untypeability,
as depicted at the bottom of \autoref{fig:abs-ref} (2).
To compute the proof of untypeability,
we start as before, by generalizing the concrete types of \T{f}'s arguments
as much as possible as long as the application remains ill-typed,
arriving at the types \T{a} and \T{M (M _t)}.
Generalization then propagates top-down through the AST:
in the next step, we compute the most general abstraction for the type of \T{x2}
such that \T{l x2} has type \T{M (M _t)}.
The generalization process stops at the leaves of the AST
(or alternatively when the type of some node cannot be generalized).
Adding the types \T{M (M _t)} and \T{L (M _t)} from the untypeability proof to the ATN
leads to the net in \autoref{fig:abs-ref} (3).

\paragraph{Iteration 3}
The shortest valid path in the third  ATN is $[\T{c}, \T{l}, \T{f}]$, which
corresponds to a well-typed program \T{f x1 (l (c x2))} (see the bottom of \autoref{fig:abs-ref} (3)),
which we return as the solution.

\subsection{Pruning Irrelevant Solutions via Demand Analysis}\label{sec:examples:demand}

Using a query type as the sole input to synthesis has its pros and cons.
On the one hand, types are programmer-friendly:
unlike input-output examples, which often become verbose and cumbersome for data other than lists,
types are concise and versatile,
and their popularity with Haskell programmers is time-tested by the \hoogle API search engine.
%
%
On the other hand,
a query type only partially captures the programmer's intent;
in other words, not all well-typed programs are equally desirable.
In our running example, the program \T{\\x1 x2 -> x1} has the right type,
but it is clearly uninteresting. 
Hence, the important challenge for \tool is:
how do we filter out uninteresting solutions
\emph{without} requiring additional input from the user?

\mypara{Relevant Typing} \label{sec:example:relevant}
\sypet offers an interesting approach to this problem:
they observe that a programmer is unlikely to include an argument in a query
if this argument is not required for the solution.
To leverage this observation, they propose to use a \emph{relevant} type system~\cite{Pierce2004},
which requires each variable to be used \emph{at least once},
making programs like \T{\\x1 x2 -> x1} ill-typed.
TTNs naturally enforce relevancy during search:
in fact, TTN reachability as described so far encodes a stricter \emph{linear} type system,
where all arguments must be used \emph{exactly once}.
%
%
This requirement can be relaxed by adding special ``copy''
transitions that consume one token from a place and produce two token in the same place.

\mypara{Demand Analysis} \label{sec:examples:demand}
Unfortunately, with expressive polymorphic components
the synthesizer discovers ingenious ways to circumvent the relevancy requirement.
For example, the terms \T{fst (x1, x2)}, \T{const x1 x2}, and \T{fromLeft x1 (Right x2)}
are all functionally equivalent to \T{x1},
even though they satisfy the letter of relevant typing.
%
To filter out solutions like these,
we use GHC's \emph{demand analysis}~\cite{SergeyVJB17} to post-process solutions returned by the ATN
and filter out those with unused variables.
Demand analysis is a whole-program analysis that peeks inside the component implementation,
and hence is able to infer in all three cases above that the variable \T{x2} is unused.
As we show in \autoref{sec:eval}, demand analysis significantly improves the quality of solutions.

\subsection{Higher-Order Functions}\label{sec:examples:hof}

Next we illustrate how ATNs scale up to account for higher-order functions
and type classes, using the component library in \autoref{fig:ho} (left),
which uses both of these features.

\mypara{Example: Iteration}
Suppose the user poses a query \T{(a -> a) -> a -> Int -> a},
with the intention to apply a function $g$
to an initial value $x$ some number of times $n$.
Perhaps surprisingly, this query can be solved using
components in \autoref{fig:ho} by creating a list with
$n$ copies of $g$, and then \emph{folding} function
application over that list with the seed $x$ -- that is,
via the term \T{\\g x n -> foldr ($) x (replicate n g)}.

\begin{figure}
\begin{minipage}{.47\textwidth}
\small
\begin{lstlisting}
-- | Function application
($) :: (_a -> _b) -> _a -> _b
-- | List with n copies of a value
replicate :: Int -> _a -> [_a]
-- | Fold a list
foldr :: (_a -> _b -> _b) -> _b -> [_a] -> _b
-- | Value stored in the option
fromJust :: Maybe _a -> _a
-- | Lookup element by key
lookup :: Eq _a => _a -> [(_a, _b)] -> Maybe _b
\end{lstlisting}
\end{minipage}
\hfill
\begin{minipage}{.52\textwidth}
\centering
\input{petri_ho}
\end{minipage}
\caption{(left) A library with higher-order functions and type-class constraints.
(center) Fragment of an ATN for the query \T{(a -> a) -> a -> Int -> a}.
(right) Fragment of an ATN for the query \T{Eq a => [(a,b)] -> a -> b}.}\label{fig:ho}
\end{figure}

Can we generate this solution using an ATN?
As described so far, ATNs only assign places to base (non-arrow) types,
and hence cannot synthesize terms that use higher-order components,
such as the application of \T{foldr} to the function \T{($)} above.
Initially, we feared that supporting higher-order components
would require generating \emph{lambda terms} within the
Petri net (to serve as their arguments) which would be
beyond the scope of this work.
However, in common cases like our example, the higher-order
argument can be written as a single variable (or component).
Hence, the full power of lambda terms is not required.

\mypara{HOF Arguments via Nullary Components}
We support the common use case --- where higher-order arguments
are just components or applications of components ---  simply by
desugaring a higher-order library into a first-order library
supported by ATN-based synthesis.
To this end, we
\begin{inparaenum}[(1)]
\item introduce a binary type constructor \T{F _a _b} to represent arrow types as if they were base types; and
\item for each component \T{c :: B1 -> ... -> Bn -> B} in the original library,
we add a \emph{nullary} component \T{'c :: F B1 (... F Bn B)}.
\end{inparaenum}
Intuitively, an ATN distinguishes between functions it \emph{calls}
(represented as transitions) and functions it uses as \emph{arguments}
to other functions (represented as tokens in corresponding \T{F} places).

\autoref{fig:ho} (center) depicts a fragment of an ATN for our example.
Note that the \T{($)} component gives rise both to a binary
transition \T{$}, which we would take if we were to apply this
component, and a nullary transition \T{'$}, which is actually
taken by our solution, since \T{($)} is used as an argument to \T{foldr}.
Since \T{F} is just an ordinary type constructor as far as the ATN is concerned,
all existing abstraction and refinement mechanisms apply to it unchanged:
for example, in \autoref{fig:ho} both \T{a -> a} and \T{(a -> a) -> a -> a} 
are abstracted into the same place \T{F _t _t}.
%

\mypara{Completeness via Point-Free Style}
While our method was inspired by the common use case
where the higher-order arguments were themselves
components, note that with a sufficiently rich
component library, \eg one that has representations
of the \T{S}, \T{K} and \T{I} combinators, our method
is \emph{complete} as every term that would have required
an explicit lambda-subterm for a function argument, can
now be written in a point-free style, only using
variables, components and their applications.


\subsection{Type classes}\label{sec:examples:tc}

Type classes are widely used in Haskell
to support \emph{ad-hoc} polymorphism~\cite{WadlerB89}.
For example, consider the type of component \T{lookup} in \autoref{fig:ho}:
this function takes as input a key $k$ of type \T{_a}
and a list of key-value pairs of type \T{[(_a, _b)]},
and returns the value that corresponds to $k$, if one exists.
In order to look up $k$, the function has to compare
keys for equality; to this end, its signature imposes
a \emph{bound} \T{Eq _a} on the type of keys, enforcing
that any concrete key type be an instance of the
\emph{type class} \T{Eq} and therefore be equipped
with a definition of equality.

Type classes are implemented by a translation to
parametric polymorphism called \emph{dictionary passing},
where each class is translated into a record
whose fields implement the different functions supported
by the type class.
Happily, \tool can use dictionary passing to desugar
synthesis with type classes into a synthesis problem
supported by ATNs.
For example, the type of \T{lookup} is desugared into
an unbounded type with an extra argument:
\T{EqD _a -> _a -> [(_a,_b)] -> _b}.
Here \T{EqD _a}, is a \emph{dictionary}: a record datatype
that stores the implementation of equality on \T{_a};
the exact definition of this datatype is unimportant,
we only care whether \T{EqD _a} for a given \T{_a}
is inhabited.

\mypara{Example: Key-Value Lookup}
As a concrete example, suppose the user wants to
perform a lookup in a key-value list \emph{assuming}
the key is present, and poses a query \T{Eq a => [(a,b)] -> a -> b}.
The intended solution to this query is \T{\\xs k -> fromJust (lookup k xs)},
\ie look up the key and then extract the value from the option, assuming it is nonempty.
A fragment of an ATN for this query is shown in \autoref{fig:ho} (right).
Note that the transition \T{l}---%
the instance of \T{lookup} with $\alpha\mapsto \T{a}, \beta\mapsto \T{b}$---%
has \T{EqD a} as one of its incoming edges.
This corresponds to our intuition about type classes:
in order to fire \T{l}, the ATN first has to prove that \T{a} satisfies \T{Eq},
or in other words, that \T{EqD a} is inhabited.
In this case, the proof is trivial:
because the query type is also desugared in the same way,
the initial marking contains a token in \T{EqD a}%
\footnote{As we explain in \autoref{sec:impl:desugar},
dictionaries can also be inhabited via instances and functional dependencies.}.
A welcome side-effect of relevant typing is
that any solution \emph{must use} the token
in \T{EqD a}, which matches our intuition
that the user would not specify the bound \T{Eq a}
if they did not mean to compare keys for equality.
This example illustrates that the combination
of (bounded) polymorphism and relevant typing
gives users a surprisingly powerful mechanism
to disambiguate their intent.
Given the query above (and a library of \componentCount components),
\tool returns the intended solution as the first result.
In contrast, given a monomorphic variant of this query 
\T{[(Int, b)] -> Int -> b} (where the key type is just 
an \T{Int}) \tool produces a flurry of irrelevant results, 
such as \T{\\xs k -> snd (xs !! k)}, which uses \T{k} 
as an \emph{index} into the list, and not as a key as 
we intended.

%% file: petri_concrete.tex
\begin{tikzpicture}[node distance=1.3cm,>=stealth',bend angle=15,auto]

  \begin{scope}
    \node [blank]          (bl1)                                          {};
    \node [place]          (Ma)  [above of=bl1, label=above:\T{M a}]      {};
    \node [place]          (La)  [below of=bl1, label=below:\T{L a}]      {};
    \node [blank]          (bl2) [right of=La]                            {};
    \node [final,tokens=1] (a)   [right of=bl2, label=below:\T{a}]        {};
    \node [blank]          (bl3) [left of=La]                             {};    
    \node [place,tokens=1] (LMa) [left of=bl3, label=below:\T{L (M a)}]   {};
    \node [blank]          (bl4) [left of=Ma]                             {};    
    \node [place]          (MMa) [left of=bl4, label=above:\T{M (M a)}]      {};

    \node [transition] (f1) [right of=bl1] {\T{f<a>}}
      edge [pre,bend left,sol]   (a)
      edge [pre,sol]             (Ma)
      edge [post,bend right,sol] (a);
      
    \node [transition] (f2) [left of=Ma] {\T{f<M a>}}
      edge [pre,bend left]       (Ma)
      edge [pre]                 (MMa)
      edge [post,bend right]     (Ma);
      
    \node [transition] (l2) [below of=Ma] {\T{l<a>}}
      edge [pre,sol]                 (La)
      edge [post,sol]                (Ma);
      
    \node [transition] (l3) [below of=MMa] {\T{l<M a>}}
      edge [pre]                 (LMa)
      edge [post]                (MMa);

    \node [transition] (c) [left of=La] {\T{c}}
      edge [pre,sol]                 (LMa)
      edge [post,sol]                (La);  
  




      
  \end{scope}

\end{tikzpicture}

%% file: petri_abstract.tex
\resizebox{\textwidth}{!}{
\begin{tikzpicture}[node distance=1.3cm,>=stealth',bend angle=15,auto]

  \begin{scope}
    \node [final,tokens=2] (top)   [label=right:\T{_t}]         {};

    \node [transition] (f) [left of=top] {\T{f}}
      edge [pre,bend left,sol] node {2}         (top)
      edge [post,bend right,sol]                (top);

    \node [transition] (c) [above of=top] {\T{c}}
      edge [pre,bend left]                            (top)
      edge [post,bend right]                           (top);

    \node [transition] (l) [below of=top] {\T{l}}
      edge [pre,bend left]                            (top)
      edge [post,bend right]                           (top);
  \end{scope}
  
  \begin{scope}[xshift=4.5cm]
    \node [blank] (bl1) {};
    \node [final] (top)   [above of=bl1, label=above:\T{_t}]         {};
    \node [place,tokens=1] (L)  [below of=bl1, label=below:\T{L _t}]      {};
    \node [final,tokens=1] (a)  [right of=L, label=below:\T{a}]      {};

    \node [transition] (f1) [left of=top] {\T{f}}
      edge [pre,bend left] node {2}       (top)
      edge [post,bend right]              (top);

    \node [transition] (f2) [below of=top] {\T{f}}
      edge [pre,bend left]       (L)
      edge [pre]                 (top)
      edge [post,bend right]     (L);

    \node [transition] (f3) [above of=a] {\T{f}}
      edge [pre,bend left,sol]       (a)
      edge [pre,sol]                 (top)
      edge [post,bend right,sol]     (a);


    \node [transition] (l2) [left of=bl1] {\T{l}}
      edge [pre,sol]                 (L)
      edge [post,sol]                (top);
            
      
    \node [transition] (c2) [left of=l2] {\T{c}}
      edge [pre]                  (top)
      edge [post]                (L);
      
  \end{scope}
    
  \begin{scope}[xshift=10cm]
    \node [blank] (bl1) {};
    \node [final] (top)   [above of=bl1, label=above:\T{_t}]         {};
    \node [place] (L)  [below of=bl1, label=below:\T{L _t}]      {};
    \node [blank] (bl2) [right of=L] {};
    \node [final,tokens=1] (a)  [right of=bl2, label=below:\T{a}]      {};
    \node [blank] (bl3) [left of=L] {};    
    \node [place,tokens=1] (LM)  [left of=bl3, label=below:\T{L (M _t)}]      {};
    \node [blank] (bl4) [left of=top] {};    
    \node [place] (MM)  [left of=bl4, label=above:\T{M (M _t)}]      {};

    \node [transition] (f1) [right of=top] {\T{f}}
      edge [pre,bend right] node [yshift=.5cm] {2}       (top)
      edge [post,bend left]                              (top);

    \node [transition] (f2) [right of=bl1] {\T{f}}
      edge [pre,bend left]       (L)
      edge [pre]                 (top)
      edge [post,bend right]     (L);

    \node [transition] (f3) [above of=a] {\T{f}}
      edge [pre,bend left,sol]       (a)
      edge [pre,sol]                 (top)
      edge [post,bend right,sol]     (a);
      
    \node [transition] (f4) [left of=top, yshift=.5cm] {\T{f}}
      edge [pre,bend left]       (MM)
      edge [pre]                 (top)
      edge [post,bend right]     (MM);
      
    \node [transition] (f6) [left of=top, yshift=-.5cm] {\T{f}}
      edge [pre]                 (MM)
      edge [pre,bend left]       (top)
      edge [post,,bend right]    (top);

    \node [transition] (f5) [left of=bl1] {\T{f}}
      edge [pre,bend left]       (LM)
      edge [pre]                 (top)
      edge [post,bend right]     (LM);


    \node [transition] (l2) [below of=top] {\T{l}}
      edge [pre,sol]                 (L)
      edge [post,sol]                (top);
      
    \node [transition] (l3) [below of=MM] {\T{l}}
      edge [pre]                 (LM)
      edge [post]                (MM);

    \node [transition] (c1) [left of=L] {\T{c}}
      edge [pre,sol]                 (LM)
      edge [post,sol]                (L);

  
  \end{scope}
  
\end{tikzpicture}
}

%% file: type_errors.tex
\begin{tikzpicture}[node distance=5mm and 5mm,>=stealth',auto]

  \begin{scope}
    \node [draw=black] (f) [label=right:\labels{$\bot$}{$\bot$}]  {\T{f}};
    \node [draw=black] (x1) [below= of f,xshift=-5mm,label=left:\labels{\T{_t}}{\T{a}}]  {\T{x1}};
    \node [draw=black] (x2) [below= of f,xshift=5mm,label=right:\labels{\T{L _t}}{\T{L (M a)}}]  {\T{x2}};
    
    \path (f) edge (x1);
    \path (f) edge (x2);
    
    \node (label1) [above of=f,yshift=5mm] {(1)};
  \end{scope}
  
  \begin{scope}[yshift=-2cm]
    \node [draw=black] (f) [label=right:\labels{$\bot$}{$\bot$}]  {\T{f}};
    \node [draw=black] (x2) [below= of f,xshift=-5mm,label=left:\labels{\T{_t}}{\T{L (M a)}}]  {\T{x2}};
    \node [draw=black] (x1) [below= of f,xshift=5mm,label=right:\labels{\T{a}}{\T{a}}]  {\T{x1}};
    
    \path (f) edge (x1);
    \path (f) edge (x2);
    
  \end{scope}

  \begin{scope}[xshift=4.5cm]
    \node [draw=black] (f)  [label=right:\labels{$\bot$}{$\bot$}]  {\T{f}};
    \node [draw=black] (x1) [below= of f,xshift=-5mm,label=left:\labels{\T{a}}{\T{a}}]  {\T{x1}};
    \node [draw=black] (l)  [below= of f,xshift=5mm,label=right:\labels{\T{M (M _t)}}{\T{M (M a)}}]  {\T{l}};
    \node [draw=black] (x2) [below= of l,label=right:\labels{\T{L (M _t)}}{\T{L (M a)}}]  {\T{x2}};
    
    \path (f) edge (x1);
    \path (f) edge (l);
    \path (l) edge (x2);
  
    \node (label2) [above of=f,yshift=5mm] {(2)};
  \end{scope}
    
  \begin{scope}[xshift=9cm]
    \node [draw=black] (f)  [label=right:\labelc{\T{a}}]  {\T{f}};
    \node [draw=black] (x1) [below= of f,xshift=-5mm,label=left:\labelc{\T{a}}]  {\T{x1}};
    \node [draw=black] (l)  [below= of f,xshift=5mm,label=right:\labelc{\T{M a}}]  {\T{l}};
    \node [draw=black] (c)  [below= of l,label=right:\labelc{\T{M a}}]  {\T{c}};
    \node [draw=black] (x2) [below= of c,label=right:\labelc{\T{L (M a)}}]  {\T{x2}};
    
    \path (f) edge (x1);
    \path (f) edge (l);
    \path (l) edge (c);
    \path (c) edge (x2);
  
    \node (label3) [above of=f,yshift=5mm] {(3)};
  \end{scope}
  
\end{tikzpicture}

%% file: petri_ho.tex
\resizebox{\textwidth}{!}{
\begin{tikzpicture}[node distance=1.3cm,>=stealth',bend angle=15,auto]
    
  \begin{scope}
    \node [blank] (bl1) {};
    \node [place,tokens=1] (Int) [above of=bl1, label=above:\T{Int}]         {};
    \node [final,tokens=1] (a)   [right of=Int, label=above:\T{a}]     {};
    \node [place,tokens=1] (F)   [below of=bl1, label=below:\T{F _t _t}]         {};
    \node [place] (L)   [right of=F, label=below:\T{L _t}]         {};
    
    \node [transition] (f) [below of=a] {\T{foldr}}
      edge [pre,bend left,sol]              (a)
      edge [pre,sol]                        (F)
      edge [pre,sol]                        (L)
      edge [post,bend right,sol]            (a);
      
    \node [transition] (r) [below of=Int] {\T{rep}}
      edge [pre,sol]                        (Int)
      edge [pre,sol]                        (F)
      edge [post,sol]                       (L);
      
    \node [transition] (app1) [left of=F] {\T{'$}}
      edge [post,sol]                       (F);
      
    \node [transition] (app2) [left of=r] {\T{$}}
      edge [pre,bend left]                  (Int)
      edge [pre]                            (F)
      edge [post,bend right]                (Int);

  \end{scope}
  
  \begin{scope}[xshift=3cm]  
    \node [blank] (bl1) {};
    \node [place,tokens=1] (a) [above of=bl1, label=above:\T{a}]         {};
    \node [place,tokens=1] (L) [right of=a,   label=above:\T{L (P a b)}]     {};
    \node [place,tokens=1] (Eq)   [below of=bl1, label=below:\T{EqD a}]         {};
    \node [place] (M)   [right of=Eq, label=below:\T{M b}]         {};
    \node [final] (b)  [right of=M, label=below:\T{b}] {};
    
    \node [transition] (l) [below of=a] {\T{l}}
      edge [pre,sol]                        (a)
      edge [pre,sol]                        (Eq)
      edge [pre,sol]                        (L)
      edge [post,sol]                       (M);
      
    \node [transition] (fJ) [below of=L] {\T{fJ}}
      edge [pre,sol]                        (M)
      edge [post,sol]                       (b);
      
      

  \end{scope}
  
\end{tikzpicture}
}

%% file: checking.tex
\section{Abstract Type Checking}

\label{sec:algo}
\label{sec:check}

Next, we formally define the syntax of our target 
language \corelang and its type system, and use the 
framework of \emph{abstract interpretation} to develop
an algorithmic \emph{abstract} type system for \corelang.
This framework allows us to parameterize the checker by the 
desired level of detail, crucially enabling our novel 
\tygar synthesis algorithm formalized in \autoref{sec:synth}.

\subsection{The \corelang Language}\label{sec:algo:concrete}

\corelang is a simple first-order language with a 
prenex-polymorphic type system, whose syntax and 
typing rules are shown in \autoref{fig:lang}.
We stratify the terms into \emph{application} terms 
which comprise variables $x$, library components $c$ 
and applications; 
and \emph{normal-form} terms which 
are lambda-abstractions over application terms.
%

The \emph{base} types $B$ include type 
variables $\tau$, as well as applications 
of a type constructor to zero or more base 
types $C\ \many{B}$.
We write $\many{X}$ to denote zero or more 
occurrences of a syntactic element $X$.
Types $T$ include base types and first-order 
function types (with base-typed arguments).
Syntactic categories $b$ and $t$ are the 
ground counterparts to $B$ and $T$ 
(\ie they contain no type variables).
A \emph{component library} $\Lambda$ is 
a finite map from a set of components $c$ 
to the components' poly-types. 
A \emph{typing environment} $\Gamma$ is a map 
from variables $x$ to their ground base types.
%
%
A \emph{substitution} $\sigma = [\subst{\tau_1}{B_1}, \ldots, \subst{\tau_n}{B_n}]$ 
is a mapping from type variables to base types
that maps each $\tau_i$ to $B_i$ and is identity elsewhere.
We write $\sapp{\sigma}{T}$ to denote the application of $\sigma$ to type $T$, 
which is defined in a standard way.

A \emph{typing judgment} $\jtyping{\Gamma}{\nex}{t}$ is only 
defined for ground types $t$.
Polymorphic components are instantiated into ground 
monotypes by the \textsc{Comp} rule, which angelically 
picks ground base types to substitute for all the 
universally-quantified type variables in the component signature
(the rule implicitly requires that $\sapp{\sigma}{T}$ be ground).
%


\begin{figure}
\begin{minipage}{.45\textwidth}
\small
\textbf{Syntax}
$$
\begin{array}{llll}
e &::= x \mid c \mid \eapp{e}{e}                         & \text{\emph{Application Terms}}\\  
\nex & ::= e \mid \elam{x}{\nex}                         & \text{\emph{Normal-Form Terms}}\\
b &::= C\ \many{b}                                       & \text{\emph{Ground Base Types}}\\  
t &::= b \mid b \to t                                    & \text{\emph{Ground Types}}\\  
B &::= \tau \mid C\ \many{B}                             & \text{\emph{Base Types}}\\  
T &::=  B \mid B \to T                                   & \emph{Types}\\  
P &::=  \many{\forall\tau}.T                             & \text{\emph{Polytypes}}\\
\Gamma &::=  \cdot \mid x:b,\Gamma                       & \text{\emph{Environments}}\\
\sigma &::=  [\many{\subst{\tau}{B}}]                    & \text{\emph{Substitutions}}
\end{array}
$$
\end{minipage}
\begin{minipage}{.45\textwidth}
\small
\textbf{Typing}\quad$\boxed{\jtyping{\Gamma}{\nex}{t}}$
\begin{gather*}
\inference[\textsc{T-Var}]
{\Gamma(x) = b}
{\jtyping{\Gamma}{x}{b} }
\\
\inference[\textsc{T-Comp}]
{\ty{c} = \many{\forall\tau}.T}
{\jtyping{\Gamma}{c}{\sapp{\sigma}{T}} }
\\
\inference[\textsc{T-App}]
{ \jtyping{\Gamma}{e_1}{b \to t} &  
  \jtyping{\Gamma}{e_2}{b}}
{\jtyping{\Gamma}{\eapp{e_1}{e_2}}{t}}
\\
\inference[\textsc{T-Fun}]
{ \jtyping{\Gamma, x:b}{\nex}{t} }
{\jtyping{\Gamma}{\elam{x}{\nex}}{b \to t}}
\end{gather*}
\end{minipage}
\caption{\corelang: syntax and declarative type system.}\label{fig:lang}
\end{figure}

\subsection{Type Checking as Abstract Interpretation}


\mypara{Type subsumption lattice}
We say that type $T'$ \emph{is more specific than} type $T$ 
(or alternatively, that $T$ \emph{is more general than} or 
\emph{subsumes} $T'$) written $T' \sub T$, iff there exists 
$\sigma$ such that $T' = \sapp{\sigma}{T}$.
The relation $\sub$ is a partial order on types.
For example, in a library with two nullary type 
constructors \T{A} and \T{B}, and a binary type 
constructor \T{P}, we have $\T{P A B} \sub \T{P _a B} \sub \T{P _a _b} \sub \tau$.
This partial order induces an equivalence relation $T_1 \equiv T_2 \triangleq T_1 \sub T_2 \wedge T_2 \sub T_1$
(equivalence up to variable renaming).
%
The order (and equivalence) relation extends to substitutions in a standard way:
$\sigma' \sub \sigma \triangleq \exists \rho . \forall \tau . \sapp{\sigma'}{\tau} = \sapp{\rho}{\sapp{\sigma}{\tau}}$.

We augment the set of types with a special bottom 
type $\bot$ 
that is strictly more specific than every other base type;
we also consider a bottom substitution $\sigma_{\bot}$ 
and define $\sapp{\sigma_{\bot}}{B} = \bot$ for any $B$.
A \emph{unifier} of $B_1$ and $B_2$ is a substitution $\sigma$ 
such that $\sapp{\sigma}{B_1} = \sapp{\sigma}{B_2}$;
note that $\sigma_{\bot}$ is a unifier for any two types.
The most general unifier (MGU) is unique up to $\equiv$,
and so, by slight abuse of notation, 
we write it as a function $\mgu{B_1,B_2}$.
We write $\mgu{\many{B_1,B_2}}$ 
for the MGU of a sequence of type pairs, where the MGU 
of an empty sequence is the identity substitution ($\mgu{\cdot} = []$).
The \emph{meet} of two base types is defined as 
$B_1 \meet B_2 = \sapp{\sigma}{B_1} (= \sapp{\sigma}{B_2})$,
where $\sigma = \mgu{B_1, B_2}$.
For example, $\T{P _a B} \meet \T{P A _b} = \T{P A B}$ 
while $\T{P _a B} \meet \T{P _b A} = \bot$.
The \emph{join} of two base types can be defined as their 
anti-unifier, but we elide a detailed discussion as joins 
are not required for our purposes. 


We write $\basebots = \bases \cup \{\bot\}$ for the 
set of base types augmented with $\bot$.
Note that $\langle \basebots, \sub, \join, \meet \rangle$ 
is a lattice with bottom element $\bot$ and top 
element $\tau$ and is isomorphic to \citet{plotkin1970lattice}'s 
\emph{subsumption lattice} on first-order logic terms.

%
%

\mypara{Type Transformers}
A component signature can be interpreted 
as a partial function that maps (tuples of) 
ground types to ground types.
For example, intuitively, a component \T{l :: forall _b.L _b -> M _b} 
maps \T{L A} to \T{M A}, \T{L (M A)} to \T{M (M A)}, and \T{A} to $\bot$.
This gives rise to \emph{type transformer} 
semantics for components, which is similar 
to predicate transformer semantics in predicate 
abstraction and SYNGAR~\cite{WangDS18}, but 
instead of being designed by a domain expert
can be derived automatically from the component 
signatures.

More formally,
we define a \emph{fresh instance} of a polytype
$\fresh{\many{\forall\tau}.T} \triangleq \sapp{[\many{\subst{\tau}{\tau'}}]}{T}$,
where $\many{\tau'}$ are fresh type variables.
Let $c$ be a component and $\fresh{\ty{c}} = \many{B'_i} \to B'$;
then a \emph{type transformer} for $c$ is a function $\sem{c}_{\Lambda}\colon \many{\basebots} \to \basebots$ defined as follows:
$$
\sem{c}_{\Lambda}(\many{B_i}) = \sapp{\sigma}{B'} \quad \text{where}\ \sigma = \mgu{\many{B_i, B'_i}}
$$
%
%
We omit the subscript $\Lambda$ where the library is clear from the context.
For example, for the component \T{l} above:
$\sem{\T{l}}(\T{L (M _t)}) = \T{M (M _t)}$, 
$\sem{\T{l}}(\tau) = \T{M}\ \tau_1$ (where $\tau_1$ is a fresh type variable), 
and $\sem{\T{l}}(\T{A}) = \bot$ (because $\mgu{\T{L}\ \tau_2, \T{A}} = \sigma_{\bot}$).
We can show that this type transformer is \emph{monotone}:
applying it to more specific types yield a more specific type.
The transformer is also \emph{sound} 
in the sense that in any concrete type derivation where 
the argument to \T{l} is more specific than some $B$, its 
result is guaranteed to be more specific than $\sem{\T{l}}(B)$.

\begin{lemma}[Trans. Monotonicity]\label{lemma:trans-monotone}
If $\many{B^1_i \sub B^2_i}$ then $\sem{c}(\many{B^1_i}) \sub \sem{c}(\many{B^2_i})$.
\end{lemma}

\begin{lemma}[Trans. Soundness]\label{lemma:trans-sound}
If $\fresh{\ty{c}} = \many{B_i} \to B$ and 
$\many{\sapp{\sigma}{B_i} \sub B'_i}$
then $\sapp{\sigma}{B} \sub \sem{c}(\many{B'_i})$.
\end{lemma}

The proofs of these and following results can be found in \autoref{appendix:props}.



\begin{figure}
\textbf{(Abstract) Type Inference}\quad$\boxed{\jainfer{\Gamma}{e}{B}{\memph{\abset}}}$ \\[0.1in]
\begin{gather*}
\inference[\textsc{I-Var}]
{\Gamma(x) = b}
{\jainfer{\Gamma}{x}{\memph{\alpha_{\abset}(\mnoemph{b})}}{\memph{\abset}} }
\quad\quad
\inference[\textsc{I-App}]
{ \many{\jainfer{\Gamma}{e_i}{B_i}{\memph{\abset}}} }
{\jainfer{\Gamma}{\eapp{c}{\many{e_i}}}{\memph{\alpha_{\abset}\left(\mnoemph{\sem{c}(\many{B_i})}\right)}}{\memph{\abset}}}
\end{gather*} 
\\[0.2in]
\textbf{(Abstract) Type Checking}\quad$\boxed{\jacheck{\Gamma}{\nex}{t}{\memph{\abset}}}$ \\[0.1in]
\begin{gather*}
\inference[\textsc{C-Fun}]
{ \jacheck{\Gamma, x:b}{\nex}{t}{\memph{\abset}} }
{\jacheck{\Gamma}{\elam{x}{\nex}}{b \to t}{\memph{\abset}}}
\quad\quad
\inference[\textsc{C-Base}]
{\jainfer{\Gamma}{e}{B}{\memph{\abset}} & b \sub B}
{\jacheck{\Gamma}{e}{b}{\memph{\abset}} }
\end{gather*}
\caption{Abstract type checking for \corelang. 
Treating $\alpha_{\abset}$ as the identity function yields concrete type checking.}\label{fig:abstract-typing}
\end{figure}

\mypara{Bidirectional Typing}
We can use type transformers to define algorithmic type checking for \corelang,
as shown in \autoref{fig:abstract-typing}.
For now, ignore the parts of the rules highlighted in red,
or, in other words, assume that $\alpha_{\abset}$ is the identity function;
the true meaning of this function is explained in the next section.
As is standard in bidirectional type checking~\cite{PierceTu00},
the type system is defined using two judgments:
the \emph{inference judgment} \jtinfer{\Gamma}{e}{B} generates the (base) type $B$ from the term $e$,
while the \emph{checking judgment} \jtcheck{\Gamma}{\nex}{t} checks $\nex$ against a known (ground) type $t$.
%
%
Algorithmic typing assumes that the term is in $\eta$-long form, 
\ie there are no partial applications.
During type checking, the outer $\lambda$-abstractions are handled by the checking rule \textsc{C-Fun},
and then the type of inner application term is inferred 
and compared with the given type $b$ in \textsc{C-Base}.

The only interesting case is the inference rule \textsc{I-App},
which handles (uncurried) component applications using their 
corresponding type transformers.
Nullary components are handled by the same rule
(note that in this case $\sem{c} = \fresh{\ty{c}}$).
This type system is algorithmic, because we have eliminated the angelic choice of polymorphic component instantiations
(recall the \textsc{T-Comp} rule in the declarative type system).
Moreover, type inference for application terms can be thought of as abstract interpretation,
where the abstract domain is the type subsumption lattice:
for any application term $e$, the inference computes its ``abstract value'' $B$
(known in type inference literature as its \emph{principal type}).
%
%
We can show that 
the algorithmic system is sound and 
complete with respect to the declarative one.

\begin{theorem}[Type Checking is Sound and Complete]\label{thm:algo-sound}
\jtyping{\cdot}{\nex}{t} iff \jtcheck{\cdot}{\nex}{t}.
\end{theorem}

\subsection{Abstract Typing}

The algorithmic typing presented so far is 
just a simplified version of Hindley-Milner type inference.
%
However, casting type inference as abstract interpretation
gives us the flexibility to tune the \emph{precision} of 
the type system by restricting the abstract domain to a 
\emph{sub-lattice} of the full type subsumption lattice.
This is similar to predicate abstraction, 
where precision is tuned by restricting 
the abstract domain to boolean combinations 
of a finite set of predicates.

\mypara{Abstract Cover}
An \emph{abstract cover} $\abset = \{A_1, \ldots, A_n\}$ 
is a set of base types $A_i \in \basebots$
that contains $\tau$ and $\bot$, and is a 
sub-lattice of the type subsumption lattice 
(importantly, it is closed under $\meet$).
%
%
For example, in a library with a nullary constructor \T{A} and two unary constructors \T{L} and \T{M},
$\abset_0 = \{\tau, \bot\}$,
$\abset_1 = \{\tau, \T{A}, \T{L _t}, \bot\}$,
and $\abset_2 = \{\tau, \T{A}, \T{L _t}, \T{L (M _t)}, \T{M (M _t)}, \bot\}$
are abstract covers.
Note that in a cover, the scope of a type variable is each individual base type,
so the different instances of $\tau$ above are unrelated.
We say that an abstract cover $\abset'$ \emph{refines} a cover $\abset$ ($\abset' \refines \abset$)
if $\abset$ is a sub-lattice of $\abset'$.
In the example above, $\abset_2 \refines \abset_1 \refines \abset_0$.

\mypara{Abstraction function}
Given an abstract cover $\abset$, 
the \emph{abstraction} $\alpha_{\abset}\colon \basebots \to \basebots$ of a base type $B$
is defined as the most specific type in $\abset$ that subsumes $B$:
$$
\alpha_{\abset}(B) = A\in \abset\ \text{such that}\ B\sub A\ \text{and}\ \forall A'\in \abset.B\sub A'\Rightarrow A\sub A'
$$
We can show that $\alpha_{\abset}(B)$ is unique,
because $\abset$ is closed under meet.
%
%
In abstract interpretation, it is customary to define a dual \emph{concretization} function.
In our case, the abstract domain $\abset$ is a sub-lattice of the concrete domain $\basebots$,
and hence our concretization function is the identity function $\mathit{id}$.
It is easy to show that $\alpha_{\abset}$ and $\mathit{id}$ form a \emph{Galois insertion},
because $B \sub \mathit{id}(\alpha_{\abset}(B))$ and $A = \alpha_{\abset}(\mathit{id}(A))$
both hold by definition of $\alpha_{\abset}$.

\mypara{Abstract Type Checking}
Armed with the definition of abstraction function,
let us now revisit \autoref{fig:abstract-typing} and consider the highlighted parts we omitted previously.
The two abstract typing judgments---for checking and inference---are parameterized by the abstract cover. 
The only interesting changes are in the \emph{abstract type inference} judgment $\jainfer{\Gamma}{e}{B}{\abset}$,
which applies the abstraction function to the inferred type at every step.
For example, recall the covers $\abset_1$ and $\abset_2$ defined above,
and consider a term \T{l xs} where $\ty{\T{l}} = \forall \beta.\T{L}\ \beta \to \T{M}\ \beta$
and $\Gamma(\T{xs}) = \T{L (M A)}$.
Then in $\abset_1$ we infer \jainfer{\Gamma}{\T{l xs}}{\tau}{\abset_1},
since $\alpha_{\abset_1}(\T{L (M A)}) = \T{L _t}$
and $\sem{\T{l}}(\T{L _t}) = \T{M _t}$, but \T{M _t} is abstracted to $\tau$.
However, in $\abset_2$ we infer \jainfer{\Gamma}{\T{l xs}}{\T{M (M _t)}}{\abset_2},
since $\alpha_{\abset_2}(\T{L (M A)}) = \T{L (M _t)}$,
and $\sem{\T{l}}(\T{L (M _t)}) = \T{M (M _t)}$, 
which is abstracted to itself.

We can show that abstraction preserves typing: 
\ie $\nex$ has type $t$ in an abstraction $\abset$ 
whenever it has type $t$ in a \emph{more refined} 
abstraction $\abset' \refines \abset$:
\begin{theorem}[Typing Preservation]\label{thm:refine}
If $\abset' \refines \abset$
and $\jacheck{\Gamma}{\nex}{t}{\abset'}$
then $\jacheck{\Gamma}{\nex}{t}{\abset}$.
\end{theorem}

As $\basebots \refines \abset$ for any $\abset$, 
the above \autoref{thm:refine} implies that abstract 
typing conservatively \emph{over-approximates} concrete 
typing:

\begin{corollary}\label{thm:over-typecheck}
If $\jtcheck{\cdot}{\nex}{t}$ then $\jacheck{\cdot}{\nex}{t}{\abset}$.
\end{corollary}


%% file: synthesis.tex
\section{Synthesis}
\label{sec:synth}

Next, we formalize the concrete and abstract 
synthesis problems, and use the notion of 
abstract type checking from \autoref{sec:check} 
to develop the \tygar synthesis algorithm, which 
solves the (concrete) synthesis problem 
by solving a sequence of abstract synthesis 
problems with increasing detail.

\mypara{Synthesis Problem}
A \emph{synthesis problem} $(\Lambda, t)$ 
is a pair of a component library and 
\emph{query type}. 
A \emph{solution} to the synthesis problem 
is a normal-form term $\nex$ such that 
$\jtyping{\cdot}{\nex}{t}$.
Note that the normal-form requirement 
does not restrict the solution space:
\corelang has no higher-order functions 
or recursion, hence any well-typed program 
has an equivalent $\eta$-long $\beta$-normal form. 
We treat the query type as a monotype 
without loss of generality: any query 
polytype $\many{\forall\tau}.T$ is 
equivalent to $\sapp{[\many{\subst{\tau}{C}}]}{T}$
where $\many{C}$ are fresh nullary 
type constructors.
%
%
The synthesis problem in \corelang is 
\emph{semi-decidable}: if a solution $\nex$ 
exists, it can be found by enumerating programs 
of increasing size.
%
Undecidability follows from a reduction from 
Post's Correspondence Problem (see \autoref{appendix:props}).

\mypara{Abstract Synthesis Problem}
An \emph{abstract synthesis problem} $(\Lambda, t, \abset)$ 
is a triple of a component library, query type, and 
abstract cover.
A \emph{solution} to the abstract synthesis problem 
is a program term $\nex$ such that $\jacheck{\cdot}{\nex}{t}{\abset}$. 
We can use \autoref{thm:over-typecheck} and \autoref{thm:algo-sound}, 
to show that any solution to a concrete synthesis problem is 
also a solution to any of its abstractions:

\begin{theorem}\label{thm:over}
If $\nex$ is a solution to $(\Lambda, t)$,
then $\nex$ is also a solution to $(\Lambda, t, \abset)$.
\end{theorem}


\subsection{Abstract Transition Nets}\label{sec:algo:atn}

Next we discuss how to construct an abstract transition net (ATN)
for a given abstract synthesis problem $(\Lambda, t, \abset)$,
and use ATN reachability to find a solution to this synthesis problem.

\mypara{Petri Nets}
A \emph{Petri net} $N$ is a triple $(P, T, E)$,
where $P$ is a set of places,
$T$ is a set of transitions,
$E\colon (P\times T)\cup(T\times P)\to \nat$ is a matrix of edge multiplicities
(absence of an edge is represented by a zero entry).
A \emph{marking} of a Petri net is a mapping $M\colon P\to \nat$
that assigns a non-negative number of tokens to every place.
A \emph{transition firing} is a triple $M_1 \steps{t} M_2$,
such that for all places $p$: $M_1(p) \geq E(p, t) \wedge M_2(p) = M_1(p) - E(p, t) + E(t, p)$.
A sequence of transitions $t_1,\ldots,t_n$ is a \emph{path} between $M$ and $M'$
if $M \steps{t_1} M_1 \ldots M_{n - 1} \steps{t_n} M'$ is a sequence of transition firings.

\mypara{ATN Construction}
Consider an abstract synthesis problem $(\Lambda, t, \abset)$,
where $t = b_1 \to \ldots \to b_n \to b$.
%
An \emph{abstract transition net} $\net(\Lambda, t, \abset)$ 
is a 5-tuple $(P, T, E, I, F)$, where $(P, T, E)$ is a 
Petri net, $I\colon P\to \nat$ is a multiset of \emph{initial places}
and $F \subseteq P$ is a set of \emph{final places} defined as follows:
\begin{enumerate}
\item the set of \emph{places} $P = \abset \setminus \setof{\bot}$;
\item \emph{initial places} are abstractions of query arguments: 
      for every $i \in [1, n]$, add 1 to $I(\alpha_{\abset}(b_i))$;
\item \emph{final places} are all places that subsume the query result: 
      $F = \setof{A \in P \mid b \sub A}$.
\item \label{step:trans}
      for each component $c \in \Lambda$ 
      and for each tuple $A, A_1, \ldots, A_m \in P$, where $m$ is the arity of $c$,
      add a \emph{transition} $t$ to $T$ iff $\alpha_{\abset}\left(\sem{c}(A_1, \ldots, A_m)\right) \equiv A$;
      set $E(t, A) = 1$ and add 1 to $E(A_j, t)$ for every $j \in [1, m]$;
\item for each initial place $\{p \in P \mid I(p) > 0\}$, 
      add a self-loop \emph{copy transition} $\kappa$ to $T$, setting $E(p, \kappa) = 1$ and $E(\kappa, p) = 2$,
      and a self-loop \emph{delete transition} $\delta$ to $T$, setting $E(p, \delta) = 1$ and $E(\delta, p) = 0$.
\end{enumerate}
Given an ATN $\net = (P, T, E, I, F)$,
%
$M_F$ is a \emph{valid final marking} if it assigns exactly one token to some final place: 
$\exists f\in F . M_F(f) = 1  \wedge \forall p \in P . p \neq f \Rightarrow M_F(p) = 0$.
A path $\pi = [t_1, \ldots, t_n]$ is a \emph{valid path} of the ATN 
($\pi \models \net$),
if it is a path in the Petri net $(P, T, E)$ 
from the marking $I$ to some valid final marking $M_F$.

\mypara{From Paths to Programs}
Any valid path $\pi$ corresponds to a set of 
normal-form terms $\fromPath{\pi}$.
The mapping from paths to programs has been 
defined in prior work on \sypet, so we do not 
formalize it here. 
Intuitively, multiple programs arise  because 
a path does not distinguish between different 
tokens in one place and has no notion of order 
of incoming edges of a transition.
%
%

\mypara{Guarantees}
ATN reachability is both sound and complete with respect to (abstract) typing:

\begin{theorem}[ATN Completeness]\label{thm:atn-complete}
If $\jacheck{\cdot}{\nex}{t}{\abset}$ and $\nex \in \fromPath{\pi}$
then $\pi \models \net(\Lambda, t, \abset)$.
\end{theorem}

\begin{theorem}[ATN Soundness]\label{thm:atn-sound}
If $\pi \models \net(\Lambda, t, \abset)$,
then 
$\exists \nex \in \fromPath{\pi}$ s.t. $\jacheck{\cdot}{\nex}{t}{\abset}$.
\end{theorem}

\mypara{Abstract Synthesis Algorithm}
\autoref{alg:synthesis} (left) presents an algorithm
for solving an abstract synthesis problem $(\Lambda, t, \abset)$.
The algorithm first constructs the ATN $\net(\Lambda, t, \abset)$.
Next, the function \textproc{\shortestpath} uses a constraint solver 
to find a shortest valid path $\pi \models \net$%
\footnote{\autoref{sec:impl:encoding} details our encoding of ATN reachability into constraints.}.
From \autoref{thm:atn-complete}, 
we know that if no valid path exists (no final marking is reachable from any initial marking),
then the abstract synthesis problem has no solution,
so the algorithm returns $\bot$.
Otherwise, it enumerates all programs $\nex \in \fromPath{\pi}$
and type-checks them abstractly, until it encounters an $\nex$ 
that is abstractly well-typed (such an $\nex$ must exists per 
\autoref{thm:atn-sound}).

\begin{figure}
\begin{minipage}{.45\textwidth}
  \small
  \begin{algorithmic}[1]
    \Require{Abstract synthesis problem ($\Lambda, t, \abset$)}
    \Ensure{Solution $e$ or $\bot$ if no solution}  
    \Statex
    \Function{\abssynth}{$\Lambda, t, \abset$}
      \Let{$\net$}{$\net(\Lambda, t, \abset)$}
      \Let{$\pi$}{\Call{\shortestpath}{$\net$}}
      \If{$\pi = \bot$}
        \State \Return{$\bot$}
      \Else 
        \For{$\nex \in \fromPath{\pi}$}
          \If{$\jacheck{\cdot}{\nex}{t}{\abset}$}
            \State \Return{$\nex$}
          \EndIf
        \EndFor             
      \EndIf
      \Statex
    \EndFunction
  \end{algorithmic}
\end{minipage}
\begin{minipage}{.45\textwidth}
  \small
  \begin{algorithmic}[1]
    \Require{Synthesis problem ($\Lambda, t$), initial cover $\abset_0$}
    \Ensure{Solution $\nex$ or $\bot$ if no solution}  
    \Statex
    \Function{Synthesize}{$\Lambda, t, \abset_0$}
      \Let{$\abset$}{$\abset_0$}
      \While{true}
        \Let{$\nex$}{\Call{\abssynth}{$\Lambda, t, \abset$}}\label{algo:cand}
        \If{$\nex = \bot$}
          \State \Return{$\bot$}
        \ElsIf{$\jtcheck{\cdot}{\nex}{t}$}\label{algo:tc}
          \State \Return{$\nex$}
        \Else
          \Let{$\abset$}{\Call{Refine}{$\abset, \nex, t$}}
        \EndIf
      \EndWhile
    \EndFunction
  \end{algorithmic}
\end{minipage}  
\caption{(left) Algorithm for the abstract synthesis problem. 
(right) The \tygar algorithm.}\label{alg:synthesis}
\end{figure}

\mypara{ATN versus TTN}
Our ATN construction is inspired by but different from 
the TTN construction in \sypet \cite{WangDS18}.
In the monomorphic setting of \sypet, it suffices to 
add a single transition per component. 
To account for our \emph{polymorphic} components,
we need a transition for every \emph{abstract instance}
of the component's polytype.
To compute the set of abstract instances, we consider 
all possible $m$-tuples of places, and for each, we 
compute the result of the abstract type transformer 
$\alpha_{\abset}\left(\sem{c}(A_1, \ldots, A_m)\right)$.
This result is either $\bot$, in which case no transition 
is added, or some $A \in P$, in which case we add 
a transition from $A_1, \ldots, A_m$ to $A$.

Due to abstraction, unlike \sypet, where the final 
marking contains a single token in the result type 
$b$, we must allow for several possible final markings.
Specifically, we allow the token to end up in any 
place $A$ that subsumes $b$, not just in its most 
precise abstraction $\alpha_{\abset}(b)$.
This is because, like any abstract interpretation, 
abstract type inference might lose precision, and 
so requiring that it infer the most precise type 
$\alpha_{\abset}(b)$ for the solution would lead 
to incompleteness.


\mypara{Enforcing Relevance}
Finally, consider copy transitions $\kappa$ and delete transitions $\delta$:
in this section, we describe an ATN that 
implements a simple, structural type system,
where each function argument can be used zero 
or more times.
Hence we allow the ATN to duplicate tokens in the initial marking $I$
using $\kappa$ transitions 
and discard them using $\delta$ transitions.
We can easily adapt the ATN definition to 
implement a relevant type system by eliminating the $\delta$ transitions
(this is what our implementation does, see \autoref{sec:impl:encoding});
a linear type system can be supported by eliminating both.


\subsection{The \tygar Algorithm}\label{sec:synthesis:tygar}

The abstract synthesis algorithm from
\autoref{alg:synthesis} either returns 
$\bot$, indicating that there is no solution 
to the synthesis problem, or a term $\nex$ 
that is abstractly well-typed.
However, this term may not be (concretely) 
well-typed, and hence, may not be a solution 
to the synthesis problem.
We now turn to the core of our technique: the 
\emph{type-guided abstraction refinement} (\tygar) 
algorithm which iteratively refines an abstract 
cover $\abset$ (starting with some $\abset_0$) until 
it is specific enough that a solution to an 
abstract synthesis problem is also well-typed 
in the concrete type system. 

\autoref{alg:synthesis} (right) describes the 
pseudocode for the \tygar procedure which takes 
as input a (concrete) synthesis problem $(\Lambda, t)$ 
and an initial abstract cover $\abset_0$,
and either returns a solution $\nex$ to the 
synthesis problem or $\bot$ if $t$ cannot be 
inhabited using the components in $\Lambda$.
In every iteration, \tygar first solves the 
abstract synthesis problem at the current 
level of abstraction $\abset$, using the 
previously defined algorithm \textproc{\abssynth}.
If the abstract problem has no solution,
then neither does the concrete one 
(by \autoref{thm:over}), so the 
algorithm returns $\bot$.
Otherwise, the algorithm type-checks the 
term $\nex$ against the concrete query type.
If it is well-typed, then $\nex$ is a 
solution to the synthesis problem $(\Lambda, t)$;
otherwise $\nex$ is \emph{spurious}.

\mypara{Refinement}
The key step in the \tygar algorithm 
is the procedure \textproc{Refine},
which takes as input the current cover 
$\abset$ and a spurious program $\nex$
and returns a refinement $\abset'$ of 
the current cover ($\abset' \refines \abset$)
such that $\nex$ is abstractly ill-typed 
in $\abset'$ ($\jancheck{\cdot}{\nex}{t}{\abset'}$).
Procedure \textproc{Refine} is detailed 
in \autoref{sec:algo:refine}, but the 
declarative description above suffices 
to see how it helps the synthesis algorithm 
make progress: in the next iteration, 
\textproc{\abssynth} cannot return the 
same spurious program $\nex$, as it no 
longer type-checks abstractly.
Moreover, the intuition is that along with 
$\nex$ the refinement rules out many other 
spurious programs that are ill-typed 
``for a similar reason''.

\mypara{Initial Cover}
The choice of initial cover $\abset_0$
has no influence on the correctness of 
the algorithm.
A natural choice is the most general 
cover $\abset_{\top} = \{\tau, \bot\}$.
In our experiments (\autoref{sec:eval})
we found that synthesis is more efficient 
if we pick the initial cover $\abset_Q(\many{b_i}\to b) = \close{\setof{\tau, \many{b_i}, b, \bot}}$%
\footnote{Here $\close{\abset}$ closes the cover under meet, 
as required by the definition of sublattice.},
which represents the query type $t = \many{b_i}\to b$ concretely.
Intuitively, the reason is that the 
distinctions between the types in $t$
are very likely to be important for 
solving the synthesis problem, 
so there is no need to make 
the algorithm re-discover
them from scratch.

\mypara{Soundness and Completeness}
\textproc{Synthesize} is a 
semi-algorithm for the synthesis problem in \corelang. 

\begin{theorem}[Soundness]
If $\text{\textproc{Synthesize}}(\Lambda, t, \abset_0)$
returns $\nex$ then $\jtyping{\cdot}{\nex}{t}$.
\end{theorem}
\begin{proof}[Proof Sketch]
This follows trivially from the type check in line~\ref{algo:tc} of the algorithm.
\end{proof}

\begin{theorem}[Completeness]
If $\exists \nex.\ \jtyping{\cdot}{\nex}{t}$
then $\text{\textproc{Synthesize}}(\Lambda, t, \abset_0)$
returns some $\nex' \neq \bot$.
\end{theorem}
\begin{proof}[Proof Sketch]
Let $\nex_0$ be some shortest solution 
to $(\Lambda, t)$ and let $k$ be the 
number of all syntactically valid 
programs of the same or smaller size than $\nex_0$
(here, the size of the program is the number of component applications).
Line~\ref{algo:cand}
cannot return $\bot$ or a program 
$\nex$ that is larger than 
$\nex_0$, since $\nex_0$ is abstractly 
well-typed at any $\abset$ by 
\autoref{thm:over-typecheck},
and \textproc{\abssynth} always 
returns a shortest abstractly
well-typed program, when one 
exists by \autoref{thm:atn-complete}.
Line~\ref{algo:cand} also cannot 
return the same solution twice
by the property of \textproc{Refine}.
Hence the algorithm must find 
a solution in at most $k$ iterations.
\end{proof}

When there is no solution, our algorithm might not terminate.
This is unavoidable, since the synthesis problem is only semi-decidable,
as we discussed at the beginning of this section.
In practice, we impose an upper bound on the length of the solution,
which then guarantees termination.

\subsection{Refining the Abstract Cover}\label{sec:algo:refine}

\begin{figure}
\centering
\begin{minipage}{.48\textwidth}
  \small
  \begin{algorithmic}[1]
    \Require{$\abset, \nex, t$ s.t. $\jntyping{\cdot}{\nex}{t}$}
    \Ensure{$\abset' \refines \abset$ s.t. $\jancheck{\cdot}{\nex}{t}{\abset'}$}  
    \Function{Refine}{$\abset, \elam{\many{x_i}}{\ebody}, \many{b_i}\to b$}
      \Let{$\Lambda$}{$\Lambda \cup (r :: b\to b)$}
      \Let{$e^*$}{$\eapp{r}{\ebody}$}
      \For{$e_j \in \subterms{e^*}$}\label{algo:refine:p_start}
        \State{\jtinfer{\many{x_i:b_i}}{e_j}{U[e_j]}}\label{algo:refine:p_init}
      \EndFor
      \Let{$U$}{\Call{Generalize}{$U, e^*$}}\label{algo:refine:p_end}
      \State \Return{$\close{\abset \cup \rng{U}}$}\label{algo:refine:return}
    \EndFunction
  \end{algorithmic}
\end{minipage}
\begin{minipage}{.5\textwidth}
  \small
  \begin{algorithmic}[1]
    \Require{$U, e$ s.t. $\inv_1 \wedge \inv_2 \wedge \inv_3$}
    \Ensure{$U'$ s.t. $\inv_1 \wedge \inv_2 \wedge \inv_3$}  
    \Function{Generalize}{$U, e$}
        \If{$e = x$}
          \State \Return{$U$}
        \ElsIf{$e = \eapp{c}{\many{e_j}}$}\label{algo:gen:rec_start}
          \Let{$\many{B_j}$}{weaken $\many{U[e_j]}$ while $\sem{c}(\many{B_j}) \sub U[e]$}
          \Let{$U'$}{$U[\many{e_j \mapsto B_j}]$}
          \LineFor{$e_j$}{\Call{Generalize}{$U', e_j$}}\label{algo:gen:rec_end}
        \EndIf
    \EndFunction
  \end{algorithmic}
\end{minipage}  
\caption{Refinement algorithm.}\label{alg:refinement}
\end{figure}

This section details the refinement step of the \tygar algorithm.
The pseudocode is given in \autoref{alg:refinement}.
The top-level function \textproc{Refine}($\abset, \nex, t$)
takes as inputs an abstract cover $\abset$, a term $\nex$, and a goal type $t$,
such that $\nex$ is ill-typed concretely ($\jntyping{\cdot}{\nex}{t}$),
but well-typed abstractly ($\jacheck{\cdot}{\nex}{t}{\abset}$).
It produces a refinement of the cover $\abset' \refines \abset$, 
such that $\nex$ is ill-typed abstractly in that new 
cover (\jancheck{\cdot}{\nex}{t}{\abset'}).

\mypara{Proof of untypeability}
At a high-level, \textproc{Refine} works by constructing a \emph{proof of untypeability} of $\nex$,
\ie a mapping $U \colon \mathbf{e}\to \basebots$ from subterms of $\nex$ to types,
such that if $\rng{U}\subseteq \abset'$ then \jancheck{\cdot}{\nex}{t}{\abset'}
(in other words, the types in $U$ contain enough information to reject $\nex$).
Once $U$ is constructed, line~\ref{algo:refine:return} adds its range to $\abset$,
and then closes the resulting set under meet.

Let us now explain how $U$ is constructed.
Let $\nex \doteq \elam{\many{x_i}}{\ebody}$, $t \doteq \many{b_i}\to b$, and $\Gamma \doteq \many{x_i\colon b_i}$.
There are two reasons why $\nex$ might not type-check against $t$:
either $\ebody$ on its own is ill-typed 
or it has a non-bottom type that nevertheless does not subsume $b$. 
To unify these two cases, \textproc{Refine} constructs a new application term $e^* = \eapp{r}{\ebody}$,
where $r$ is a dedicated component of type $b \to b$;
such $e^*$ is guaranteed to be ill-typed on its own: \jtinfer{\Gamma}{e^*}{\bot}.
Lines~\ref{algo:refine:p_start}--\ref{algo:refine:p_init} initialize $U$ 
for each subterm of $e^*$ with the result of concrete type inference.
%
At this point $U$ already constitutes a valid proof of untypeability,
but it contains too much information;
in line~\ref{algo:refine:p_end} the call to \textproc{Generalize}
removes as much information from $U$ as possible 
while maintaining enough to prove that $e^*$ is ill-typed.
%
More precisely, \textproc{Generalize} maintains three crucial invariants
that together guarantee that $U$ is a proof of untypeability:
\begin{description}
\item[$\inv_1$:] \emph{($U$ subsumes concrete typing)} For any $e \in \subterms{e^*}$, 
if $\jtinfer{\Gamma}{e}{B}$, then $B \sub U[e]$;
\item[$\inv_2$:] \emph{($U$ abstracts type transformers)} For any application subterm $e = \eapp{c}{\many{e_j}}$,
$\sem{c}(\many{U[e_j]}) \sub U[e]$;
\item[$\inv_3$:] \emph{($U$ proves untypeability)} $U[e^*] = \bot$.
\end{description}

\begin{lemma}\label{lemma:invp}
If $\inv_1 \wedge \inv_2 \wedge \inv_3$ then $U$ is a proof of untypeability:
if $\rng{U}\subseteq \abset'$ then \jancheck{\cdot}{\nex}{t}{\abset'}.
\end{lemma}
\begin{proof}[Proof Sketch]
We can show by induction on the derivation that for any $\abset' \supseteq \rng{U}$
and node $e$,
$\jainfer{\Gamma}{e}{B}{\abset'} \sub U[e]$
(base case follows from $\inv_1$, and inductive case follows from $\inv_2$).
Hence, $\jainfer{\Gamma}{e^*}{B}{\abset'} \sub U[e^*] = \bot$ (by $\inv_3$),
so $\jainfer{\Gamma}{\ebody}{B}{\abset'} \not\sub b$, and \jancheck{\cdot}{\nex}{t}{\abset'}.
\end{proof}

\mypara{Correctness of \textproc{Generalize}}
Now that we know that invariants $\inv_1$--$\inv_3$ are sufficient for correctness,
let us turn to the inner workings of \textproc{Generalize}.
This function starts with the initial proof $U$ (concrete typing),
and recursively traverses the term $e^*$ top-down.
At each application node $e = \eapp{c}{\many{e_j}}$
it \emph{weakens} the argument labels $\many{U[e_j]}$
(lines~\ref{algo:gen:rec_start}--\ref{algo:gen:rec_end}).
The weakening step performs \emph{lattice search} 
to find more general values for $\many{U[e_j]}$ allowed by $\inv_2$.
More concretely, each new value $B_j$ starts out as the initial value of $\many{U[e_j]}$;
at each step, weakening picks one $B_j \neq \bot$ and moves it upward in the lattice 
by replacing a ground subterm of $B_j$ with a type variable;
the step is accepted as long as $\sem{c}(\many{B_j}) \sub U[e]$.
The search terminates when there is no more $B_j$ that can be weakened.
Note that in general there is no unique most general value for $\many{B_j}$,
we simply pick the first value we find that cannot be weakened any further.
The correctness of the algorithm does not depend on the choice of $\many{B_j}$,
and only rests on two properties:
\begin{inparaenum}[(1)]
\item $\many{U[e_j] \sub B_j}$ and 
\item $\sem{c}(\many{B_j}) \sub U[e]$.
\end{inparaenum}

We can show that \textproc{Generalize} maintains the invariants $\inv_1$--$\inv_3$.
$\inv_1$ is maintained by property (1) of weakening
(we start from concrete types and only move up in the lattice).
$\inv_2$ is maintained between $e$ and its children $\many{e_j}$ by property (2) of weakening,
and between each $e_j$ and its children because the label of $e_j$ only goes up.
Finally,
$\inv_3$ is trivially maintained since we never update $U[e^*]$.


\begin{figure}
\begin{minipage}{.3\textwidth}
\centering
\input{refinement}
\caption{\textproc{Refine} in the second iteration of the running example.}\label{fig:refine}
\end{minipage}
\hfill\vline\hfill
\begin{minipage}{.65\textwidth}
\centering
\input{unsat}
\vspace{-.4cm}
\caption{\textproc{Synthesize} on an unsatisfiable problem.}\label{fig:unsat}
\end{minipage}
\end{figure}

\mypara{Example 1}
Let us walk through the refinement step in iteration 2 of our running example from \autoref{sec:examples:abstract}.
As a reminder,
$\ty{\T{f}} = \forall \alpha . \alpha \to \T{M}\ \alpha \to \alpha$ and
$\ty{\T{l}} = \forall \beta . \T{L}\ \beta\to \T{M}\ \beta$.
Consider a call to $\textproc{Refine}(\abset, \nex, t)$,
where $\abset = \{\tau, \T{A}, \T{L _t}, \bot\}$,
$\nex = \lambda x_1\ x_2 . \T{f}\ x_1\ (\T{l}\ x_2)$
and $t = \T{A -> L (M A) -> A}$.
Let us denote $\Gamma = x_1:\T{A},x_2:\T{L (M A)}$.
It is easy to see that $\nex$ is ill-typed concretely but well-typed abstractly,
since, as explained above, \jainfer{\Gamma}{\T{l}\ x_2}{\tau}{\abset},
and hence \jainfer{\Gamma}{\T{f}\ x_1\ (\T{l}\ x_2)}{\T{A}}{\abset}.
\textproc{Refine} first constructs $e^* = \eapp{r}{\ebody}$;
the AST for this term is shown on \autoref{fig:refine} (left).
It then initializes the mapping $U$ with concrete inferred types,
depicted as red labels;
as expected $U[e^*] = \bot$.
The blue labels show $U'$ obtained by calling \textproc{Generalize}
through the following series of recursive calls:
\begin{itemize}
\item In the initial call to \textproc{Generalize}, the term $e$ is $\eapp{r}{\ebody}$;
although it is an application, we do not weaken the label for $\ebody$
since its concrete type is $\bot$, which cannot be weakened.
\item We move on to $\ebody = \eapp{\eapp{\T{f}}{x_1}}{l}$ with $U[x_1] = \T{A}$ and $U[l] = \T{M (M A)}$.
The former type cannot be weakened: an attempt to replace \T{A} with $\tau$
causes $\sem{\T{f}}$ to produce $\T{M A} \not\sub \bot$.
The latter type can be weakened by replacing \T{A} with $\tau$
(since $\sem{\T{f}}(\T{A}, \T{M (M _t)}) = \bot$), but no further.
\item The first child of \T{f}, $x_1$, is a variable so $U$ remains unchanged.
\item For the second child of \T{f}, $l = \eapp{\T{l}}{x_2}$,
\T{l}'s signature allows us to weaken $U[x_2]$ to \T{L (M _t)} but no further,
since $\sem{\T{l}}(\T{L (M _t)}) = \T{M (M _t)}$
but $\sem{\T{l}}(\T{L _t}) = \T{M _t} \not\sub \T{M (M _t)}$.
\item Since $x_2$ is a variable, \textproc{Generalize} terminates.  
\end{itemize}

\mypara{Example 2}
We conclude this section with an end-to-end application of \tygar
to a very small but illustrative example.
Consider a library $\Lambda$ with three type constructors, \T{Z}, \T{U}, and \T{B}
(with arities 0, 1, and 2, respectively),
and two components, \T{f} and \T{g}, such that:
$\Lambda(\T{f}) = \forall\alpha . \T{B}\ \alpha\ \alpha$
and $\Lambda(\T{g}) = \forall\beta . \T{B}\ (\T{U}\ \beta)\ \beta \to \T{Z}$.
Consider the synthesis problem $(\Lambda, \T{Z})$, which has no solutions:
the only way to obtain a \T{Z} is from \T{g},
which requires a \T{B} with \emph{distinct} parameters,
but we can only construct a \T{B} with \emph{equal} parameters (using \T{f}). 
Assume that the initial abstract cover is $\abset_0 = \setof{\tau, \bot}$, 
as shown in the upper left of \autoref{fig:unsat}.
$\textproc{SynAbstract}(\Lambda, \T{Z}, \abset_0)$ returns a program \T{f},
which is spurious, hence we invoke $\textproc{Refine}(\abset_0, \T{f}, \T{Z})$.
The results of concrete type inference are shown as red labels in \autoref{fig:unsat};
in particular, note that because \T{f} is a nullary component, 
$\sem{\T{f}}$ is simply a fresh instance of its type, here \T{B _t _t},
which can be generalized to \T{B _a _b}:
the root cause of the type error is that $r$ does not accept a \T{B}.
In the second iteration, $\abset_0 = \setof{\tau, \T{B _a _b}, \bot}$
and $\textproc{SynAbstract}(\Lambda, \T{Z}, \abset_1)$ returns \T{g f},
which is also spurious.
In this call to \textproc{Refine}, however, the concrete type of \T{f} can no longer be generalized:
the root cause of the type error is that $g$ accepts a \T{B} with distinct parameters.
Adding \T{B _t _t} to the cover, results in the ATN on the right,
which does not have a valid path
($\textproc{SynAbstract}$ returns $\bot$).

There are three interesting points to note about this example.
\begin{inparaenum}[(1)]
\item In general, even concrete type inference may produce non-ground types,
for example: \jtinfer{\cdot}{\T{f}}{\T{B}\ \tau\ \tau}.
\item \textproc{Synthesize} can \emph{sometimes} detect that there is no solution, 
even when the space of all possible ground base types is infinite.
\item To prove untypeability of \T{g f},
our abstract domain must be able to express non-linear type-level terms
(\ie types with repeated variables, like $\T{B}\ \tau\ \tau$);
we could not, for example, replace type variables with a single construct \T{?}, 
as in gradual typing~\cite{Siek06gradualtyping}.
\end{inparaenum}

%% file: refinement.tex
\begin{tikzpicture}[node distance=5mm and 5mm,>=stealth',auto]
    \node [draw=gray]  (r)  [label=right:\labels{$\bot$}{$\bot$}]  {$r$};
    \node [draw=black] (f)  [below= of r,label=right:\labels{$\bot$}{$\bot$}]  {\T{f}};
    \node [draw=black] (x1) [below= of f,xshift=-5mm,label=left:\labels{\T{A}}{\T{A}}]  {\T{x1}};
    \node [draw=black] (l)  [below= of f,xshift=5mm,label=right:\labels{\T{M (M _t)}}{\T{M (M A)}}]  {\T{l}};
    \node [draw=black] (x2) [below= of l,label=right:\labels{\T{L (M _t)}}{\T{L (M A)}}]  {\T{x2}};
    
    \path (r) edge (f);
    \path (f) edge (x1);
    \path (f) edge (l);
    \path (l) edge (x2);
\end{tikzpicture}

%% file: unsat.tex
\resizebox{\textwidth}{!}{
\begin{tikzpicture}[node distance=1cm,>=stealth',bend angle=15,auto]

  \begin{scope}
    \node [final] (top)   [label=right:\T{_t}]         {};

    \node [transition] (f) [above of=top] {\T{f}}
      edge [post,sol]        (top);

    \node [transition] (g) [left of=top] {\T{g}}
      edge [pre,bend left]          (top)
      edge [post,bend right]        (top);
  \end{scope}
  
  \begin{scope}[yshift=-1.3cm,node distance=.8cm]
  
    \node [draw=gray]  (r)  [label=right:\labels{$\bot$}{$\bot$}]  {$r$};        
    \node [draw=black] (f)  [below of=r,label=left:\labels{\T{B _a _b}}{\T{B _t _t}}]  {\T{f}};
    \path (r) edge (f);
    
  \end{scope}
    
  \begin{scope}[xshift=2cm]
    \node [final] (top)   [label=below:\T{_t}]         {};
    \node [blank]          (bl2) [right of=top]                            {};
    \node [place] (B)  [right of=bl2, label=below:\T{B _a _b}]      {};    

    \node [transition] (f) [above of=B] {\T{f}}
      edge [post,sol]        (B);

    \node [transition] (g1) [above of=top] {\T{g}}
      edge [pre,bend left]          (top)
      edge [post,bend right]        (top);
      
    \node [transition] (g2) [right of=top] {\T{g}}
      edge [pre,sol]          (B)
      edge [post,sol]         (top);      
  \end{scope}
  
  \begin{scope}[xshift=3cm,yshift=-1.3cm,node distance=.8cm]
  
    \node [draw=gray]  (r)  [label=right:\labels{$\bot$}{$\bot$}]  {$r$};        
    \node [draw=black] (g)  [below of=r,label=left:\labels{$\bot$}{$\bot$}]  {\T{g}};
    \node [draw=black] (f)  [below of=g,label=right:\labels{\T{B _t _t}}{\T{B _t _t}}]  {\T{f}};
    \path (r) edge (g);
    \path (g) edge (f);
    
  \end{scope}

  \begin{scope}[xshift=6cm]
    \node [final] (top)   [label=below:\T{_t}]         {};
    \node [blank]          (bl2) [right of=top]                            {};
    \node [place] (Btt)  [above of=bl2, label=above:\T{B _t _t}]      {};    
    \node [place] (B)  [right of=bl2, label=below:\T{B _a _b}]      {};    

    \node [transition] (f) [right of=Btt] {\T{f}}
      edge [post]      (Btt);

    \node [transition] (g1) [above of=top] {\T{g}}
      edge [pre,bend left]          (top)
      edge [post,bend right]        (top);
      
    \node [transition] (g2) [right of=top] {\T{g}}
      edge [pre]          (B)
      edge [post]         (top);      
  \end{scope}
\end{tikzpicture}
}

%% file: implementation.tex
\section{Implementation}\label{sec:impl}

We have implemented the \tygar synthesis algorithm in Haskell,
in a tool called \tool.
The tool relies on the Z3 SMT solver~\cite{deMoura-Bjorner:TACAS08}
to find paths in the ATN.
This section focuses on interesting implementation details, such as
desugaring Haskell libraries into first-order components accepted by \tygar,
an efficient and incremental algorithm for ATN construction,
and the SMT encoding of ATN reachability.

\subsection{Desugaring Haskell Types}\label{sec:impl:desugar}

The Haskell type system is significantly more expressive than 
that of our core language \corelang, and many of its advanced 
features are not supported by \tool.
%
However, two type system features are ubiquitous in Haskell:
higher-order functions and type classes.
As we illustrated in \autoref{sec:examples:hof} and \autoref{sec:examples:tc},
\tool handles both features by desugaring them into \corelang.
Next, we give more detail on how \tool 
translates a Haskell synthesis problem $(\hask{\Lambda}, \hask{t})$
into a \corelang synthesis problem $(\Lambda, t)$:
\begin{enumerate}
\item $\Lambda$ includes a fresh binary type constructor \T{F _a _b}
(used to represent function types).

\item Every declaration of type class \T{C _t} 
with methods $m_i :: \forall \tau . T_i$ in $\hask{\Lambda}$
gives rise to a type constructor \T{CD _t} (the dictionary type) 
and components $m_i :: \forall \tau . \T{CD _t} \to T_i$ in $\Lambda$.
For example, a type class declaration \T{class Eq _a where (==) :: a -> a -> Bool}
creates a fresh type constructor \T{EqD _a}
and a component \T{(==) :: EqD _a -> _a -> _a -> Bool}.
%

\item Every instance declaration \T{C B} in $\hask{\Lambda}$ 
produces a component that returns a dictionary \T{CD B}.
So \T{instance Eq Int} creates a component \T{eqInt :: EqD Int},
while a subclass instance like \T{instance Eq a => Eq [a]}
creates a component 
\T{eqList :: EqD a -> EqD [a]}.
Note that the exact implementation of the type class 
methods inside the instance is irrelevant; all we care 
about is that the instance inhabits the type class dictionary.

\item For every component $c$ in $\hask{\Lambda}$, 
we add a component $c$ to $\Lambda$ and define $\Lambda(c) = \fo\left(\hask{\Lambda}(c)\right)$,
where the translation function $\fo$, which eliminates type class constraints
and higher-order types, is defined as follows:
\begin{align*}
\fo\left(\forall \many{\tau} . (\T{C}_1\ \tau_1,\ldots, \T{C}_n\ \tau_n)\Rightarrow T\right) &= \forall \many{\tau} . \T{CD}_1\ \tau_1 \to \ldots \to \T{CD}_n\ \tau_n \to \fo(T) \span \span\\
\fo(T_1 \to T_2) &= \toBase{T_1} \to \fo(T_2)             & \fo(B) &= B\\
\toBase{T_1 \to T_2} &= \T{F}\ \toBase{T_1}\ \toBase{T_2} & \toBase{B} &= B
\end{align*}
For example, Haskell components on the left are translated into \corelang components on the right:
\begin{tabular}{lll}
\T{member :: Eq _a => _a -> [_a] -> Bool}  & \quad\quad & \T{member :: EqD _a -> _a -> [_a] -> Bool}\\
\T{any    :: (_a -> Bool) -> [_a] -> Bool} & \quad\quad & \T{any    :: F _a Bool -> [_a] -> Bool}\\
\end{tabular}

\item For every non-nullary component and type class method $c$ in $\hask{\Lambda}$,
we add a nullary component $c'$ to $\Lambda$ and define $\Lambda(c') = \toBase{\Lambda(c)}$.
For example: \T{any' :: F (F _a Bool) (F [_a] Bool)}.

\item Finally, the \corelang query type $t$ is defined as $\fo(\hask{t})$.
\end{enumerate}

\mypara{Limitations}
Firstly, in modern Haskell, type classes often constrain \emph{higher-kinded} type variables;
for example, the \T{Monad} type class in the signature \T{return :: Monad m => a -> m a}
is a constraint on \emph{type constructors} rather than \emph{types}.
Support for higher-kinded type variables is beyond the scope of this paper.
Secondly, in theory our encoding of higher-order functions 
(\autoref{sec:examples:hof}) is complete, as any program 
can be re-written in \emph{point-free style}, 
\ie without lambda terms, using an appropriate 
set of components \cite{Barendregt} including 
an \emph{apply} component \T{($) :: F _a _b -> _a -> _b} 
that enables synthesizing terms containing 
partially applied functions. 
However, in practice we found that adding a 
nullary version for every component  
significantly increases the size of 
the search space and is infeasible 
for component libraries of nontrivial size.
Hence, in our evaluation we 
only generate nullary variants of 
a selected subset of popular components.




\subsection{ATN Construction}\label{sec:impl:construct}

\mypara{Incremental updates}
\autoref{sec:algo:atn} shows how to construct an ATN
given an abstract synthesis problem $(\Lambda,t, \abset)$.
%
However, computing the set of ATN transitions and edges from scratch in each refinement iteration is expensive.
We observe that each iteration only makes small changes to the abstract cover,
which translate to small changes in the ATN. 

Let $\abset$ be the old abstract cover 
and $\abset' = \abset \cup \{\anew\}$ be the new abstract cover
(if a refinement step adds multiple types to $\abset$, we can consider them one by one).
Let $\parents$ be the direct successors of $\anew$ in the $\sub$ partial order;
for example, in the cover $\{\tau, \T{P _a _b}, \T{P A _b}, \T{P _a B}, \T{P A B}, \bot\}$,
the parents of $\T{P A B}$ are $\{\T{P A _b}, \T{P _a B}\}$.
Intuitively, adding $\anew$ to the cover can \emph{add} new transitions
and \emph{re-route} some existing transitions.
A transition is re-routed if a component $c$ returns a more precise type under $\abset'$ than it did under $\abset$,
given the same types as arguments.
Our insight is that the only candidates for re-routing 
are those transitions that return one of the types in $\parents$.
Similarly, all new transitions can be derived from those that take one of the types in $\parents$ as an argument.
More precisely, starting from the old ATN, we update its transitions $T$ and edges $E$
as follows:
\begin{enumerate}
    \item Consider a transition $t \in T$ that represents the abstract instance
    $\alpha_{\abset}\left(\sem{c}(\many{A_i})\right) = A$
    such that $A \in \parents$;
    if $\alpha_{\abset'}\left(\sem{c}(\many{A_i})\right) = \anew$,
    set $E(t, A) = 0$ and $E(t, \anew) = 1$.
    
    \item Consider a transition $t \in T$ that represents the abstract instance
    $\alpha_{\abset}\left(\sem{c}(\many{A_i})\right) = A$
    such that at least one $A_i \in \parents$;
    consider $\many{A'_i}$ obtained from $\many{A_i}$ by substituting at least one $A_i \in \parents$ with $\anew$;
    if $\alpha_{\abset'}\left(\sem{c}(\many{A'_i})\right) = A' \neq \bot$,
    add a new transition $t'$ to $T$, 
    set $E(t',A') = 1$ and add 1 to $E(A'_i, t')$ for each $A'_i$.
\end{enumerate}

\mypara{Transition coalescing}
The ATN construction algorithm in \autoref{sec:algo:atn} adds a separate transition 
for each abstract instance of each component in the library.
Observe, however, that different components may share the same abstract instance: for example in \autoref{fig:abs-ref} (1),
both $c$ and $l$ have the type $\T{_t -> _t}$.
Our implementation \emph{coalesces} equivalent transitions:
an optimization known in the literature as \emph{observational equivalence reduction}~\cite{WangDS18,AlurRU17}.
More precisely, we do not add a new transition 
if one already exists in the net with the same incoming and outgoing edges. 
Instead, we keep track of a mapping from each transition to a set of components. 
Once a valid path $[t_1, \ldots, t_n]$ is found,
where each transition $t_i$ represents a set of components, 
we select an arbitrary component from each set to construct the candidate program. 
In each refinement iteration, 
the transition mapping changes as follows:
\begin{enumerate}
    \item new component instances are coalesced into new groups and added to the map,
    each new group is added as a new ATN transition;
    \item if a component instance is re-routed,
    it is removed from the corresponding group;
    \item transitions with empty groups are removed from the ATN.
\end{enumerate}

\subsection{SMT Encoding of ATN Reachability}\label{sec:impl:encoding}

Our encoding differs slightly from that in previous work on \sypet.
Most notably, we use an SMT (as opposed to SAT) encoding,
in particular, representing transition firings as integer variables. 
This makes our encoding more compact,
which is important in our setting, since, unlike \sypet,
we cannot pre-compute the constraints for a component library
and use them for all queries.
%

\input{encoding}

%% file: encoding.tex
\newcommand{\tokens}[2]{\ensuremath{\T{tok}^{#1}_{#2}}}
\newcommand{\fire}[1]{\ensuremath{\T{fire}_{#1}}}

\mypara{ATN Encoding} 
Given a ATN $\net = (P, T, E, I, F)$, 
we show how to build an SMT formula $\phi$ that encodes all valid paths of a given length $\ell$;
the overall search will then proceed by iteratively increasing the length $\ell$.
We encode the number of tokens in each place $p \in P$ 
at each time step $k \in [0,\ell]$ as an integer variable \tokens{p}{k}.
%
We encode the transition firing at each time step $k \in [0,\ell)$ 
as an integer variable \fire{k} so that $\fire{k} = t$ indicates that 
the transition $t$ is fired at time step $k$.
For any $x \in \{P \cup T\}$, 
let the \emph{pre-image}  of $x$ be $\pre(x) = \{y \in P \cup T \mid E(y, x) > 0\}$
and the \emph{post-image} of $x$ be $\post(x) = \{y \in P \cup T \mid E(x, y) > 0\}$.

The formula $\phi$ is a conjunction of the following constraints:
\begin{enumerate}
\item At each time step, a valid transition is fired:
$
\bigwedge_{k=0}^{\ell-1} 1 \leq \fire{k} \leq |T|
$

\item If a transition $t$ is fired at time step $k$ 
then all places $p \in \pre(t)$ have sufficiently many tokens:
$
\bigwedge_{k=0}^{\ell-1} \bigwedge_{t=1}^{|T|} \fire{k} = t \implies \bigwedge_{p \in \pre(t)} \tokens{p}{k} \geq E(p, t)
$

\item If a transition $t$ is fired at time step $k$ 
then all places $p \in \pre(t) \cup post(t)$ will have their markings updated at time step $k+1$:
$
\bigwedge_{k=0}^{\ell-1} \bigwedge_{t=1}^{|T|} \fire{k} = t \implies \bigwedge_{p \in \pre(t) \cup \post(t)} \tokens{p}{k + 1} = \tokens{p}{k} - E(p, t) + E(t, p)
$

\item If none of the outgoing or incoming transitions of a place $p$ are fired
at time step $k$, then the marking in $p$ does not change:
$
\bigwedge_{k=0}^{\ell-1}\bigwedge_{p \in P}(\bigwedge_{t\in \pre(p)\cup \post(p)} \fire{k} \neq t) \implies \tokens{p}{k + 1} = \tokens{p}{k}
$

\item The initial marking is $I$: $\bigwedge_{p \in P} \tokens{p}{0} = I(p)$.

\item\label{enc:final} The final marking is valid:
$\bigvee_{f \in F} \left( \tokens{f}{\ell} = 1 \wedge \bigwedge_{p \in P \setminus \{f\}} \tokens{p}{\ell} = 0\right)$.

\end{enumerate}

\mypara{Optimizations}
Although the validity of the final marking can be encoded as in~(\ref{enc:final}) above,
we found that quality of solutions improves
if instead we iterate through $f \in F$ in the order from \emph{most to least precise};
in each iteration we enforce $\tokens{f}{\ell} = 1$ (and $\tokens{p}{\ell} = 0$ for $p \neq f$),
and move to the next place if no solution exists.
Intuitively, this works because paths that end in a more precise place lose less information, 
and hence are more likely to correspond to concretely well-typed programs.

As we mentioned in \autoref{sec:synth},
our implementation adds copy transitions but not delete transitions to the ATN,
thereby enforcing relevant typing.
We have also tried an alternative encoding of relevant typing,
which forgoes copy transitions, 
and instead allows the initial marking to contain extra tokens in initial places:
$\bigwedge_{p \in \{P \mid I(p) > 0\}} \tokens{p}{0} \geq I(p)$ and 
$\bigwedge_{p \in \{P \mid I(p) = 0\}} \tokens{p}{0} = 0$.
Although this alternative encoding often produces solutions faster
(due to shorter paths),
we found that the quality of solutions suffers.
We conjecture that the original encoding works well,
because it \emph{biases} the search towards linear consumption of resources,
which is common for desirable programs.

%% file: evaluation.tex
\section{Evaluation}\label{sec:eval}

Next, we describe an empirical evaluation of two research questions of \tool:

\begin{itemize}
  \item \textbf{Efficiency:} Is \tygar able to find well-typed programs quickly?
  \item \textbf{Quality of Solutions:} Are the synthesized code snippets interesting?

\end{itemize}

\mypara{Component library}
We use the same set of \componentCount components in all experiments.
To create this set, we started with all components from 12 popular Haskell library modules,%
\footnote{
\T{Data.Maybe},
\T{Data.Either},
\T{Data.Int},
\T{Data.Bool},
\T{Data.Tuple},
\T{GHC.List},
\T{Text.Show},
\T{GHC.Char},
\T{Data.Int},
\T{Data.Function},
\T{Data.ByteString.Lazy},
\T{Data.ByteString.Lazy.Builder}.
}
and excluded seven components\footnote{
\T{id},
\T{const},
\T{fix},
\T{on},
\T{flip},
\T{&},
\T{(.)}.}
that are highly-polymorphic yet redundant 
(and hence slowed down the search with no added benefit).


\mypara{Query Selection}
We collected \benchmarkCount benchmark queries from three sources:
\begin{enumerate}
\item \emph{\hoogle.}
We started with all queries made to \hoogle between 1/2015 and 2/2019.
Among the 3.8M raw queries,
71K were syntactically unique,
and only 60K could not be exactly solved by \hoogle.
Among these, many were syntactically ill-formed (\eg \T{FromJSON a -> Parser a ->}) or unrealizable (\eg \T{a -> b}).
We wanted to discard such invalid queries, but had no way to identify unrealizable queries automatically.
Instead we decided to reduce the number of queries by selecting only popular ones
(those asked at least \emph{five} times),
leaving us with 1750 queries, and then we pruned invalid queries manually, leaving us with 180 queries.
Finally, out of the 180 remaining queries, only \textbf{\hoogleOnlyBms} were realizable with our selected component set.
%
\item \emph{\sover.}
We first collected all Haskell-related questions from \sover,
ranked them by their view counts, and examined the first 500.
Out of 15 queries with implementations,
we selected \textbf{6} that were realizable with our component set.
%
\item \emph{Curated.}
Since we were unable to find many API-related Haskell questions
on \sover,
and \hoogle queries do not come with expected solutions and also tend to be easy,
we supplemented the benchmark set with \textbf{\ourBms} queries
from our own experience as Haskell programmers.
\end{enumerate}
%
The resulting benchmark set can be found in \autoref{fig:results}.

\mypara{Experiment Platform}
We ran all experiments on a machine with an Intel Core i7-3770
running at 3.4Ghz with 32Gb of RAM. The platform ran Debian
10, GHC 8.4.3, and Z3 4.7.1.

\subsection{Efficiency}

\input{results}
\mypara{Setup}
To evaluate the efficiency of \tool,
we run it on each of the \benchmarkCount queries,
and report the time to synthesize the first well-typed solution that passes the demand analyzer
(\autoref{sec:examples:demand}).
We set the timeout to \timeout seconds and take the median
time over three runs to reduce the uncertainty generated by using an SMT solver.
%
%
To assess the importance of \tygar,
we compare five variants of \tool:
\begin{enumerate}
  \item \mFixType: we \emph{monomorphise} the component library
       by instantiating all type constructors with all
       types up to an unfolding depth of one and do not use refinement.

  \item \mNoGar: we build the ATN from the abstract cover $\abset_Q$,
        which precisely represents types from the query (defined in \autoref{sec:synthesis:tygar}).
        We do not use refinement, and instead \emph{enumerate}
        solutions to the abstract synthesis problem until one type checks concretely.
        Hence, this variant uses our abstract typing but does not use \tygar.

  \item \mTopType, which uses \tygar with the initial cover $\abset_{\top} = \{\tau, \bot\}$.

  \item \mQryType, which uses \tygar with the initial cover $\abset_Q$.

  \item \mQryTypeBounded$[N]$, which is like \mQryType,
  but the size of the abstract cover is \emph{bounded}:
  once the cover reaches size $N$, it stops further refinement and reverts to \mNoGar-style enumeration.
\end{enumerate}

\mypara{Results}
\autoref{fig:results} reports total synthesis time for four out of the five variants.
\mFixType did not complete any benchmark within \timeout seconds:
it spent all this time creating the TTN, and is thus is omitted from tables and graphs.
\autoref{fig:searchgraph} plots the number of successfully completed benchmarks
against time taken for the remaining four variants
(higher and weighted to the left is better).
As you can see, \mNoGar is quite fast on easy problems,
but then it plateaus, and can only solve \nogarSolnCount out of \benchmarkCount queries.
On the other hand, \mTopType and \mQryType are slower,
and only manage to solve \tygarZSolnCount and \tygarQSolnCount queries, respectively.
After several refinement iterations, the ATNs grow too large,
and these two variants spend a lot of time in the SMT solver, as shown in
columns \textit{st-Q} and \textit{st-0} in \autoref{fig:results}.
Other than \mFixType, no other variant spent any meaningful amount of time
building the ATN.

%

\mypara{Bounded Refinement}
\input{initialcover}
We observe that \mNoGar and \mQryType have complimentary strengths and weaknesses:
although \mNoGar is usually faster,
\mQryType was able to find some solutions that \mNoGar could not
(for example, query 33: \T{appBoth}).
%
%
We conclude that refinement is able to discover interesting abstractions,
but because it is forced to make a new distinction between types in every iteration,
after a while it is bound to start making irrelevant distinctions,
and the ATN grows too large for the solver to efficiently navigate.
To combine  the strengths of the two approaches,
we consider \mQryTypeBounded,
which first uses refinement, and then switches to enumeration
once the ATN reaches a certain bound on its number of places.
To determine the optimal bound, we run the experiment
with bounds 5, 10, 15, and 20.

\autoref{fig:boundgraph} plots the results.
As you can see, for easy queries, a bound of 5 performs the best:
this correspond to our intuition that
when the solution is easily reachable,
it is faster to simply enumerate more candidates than spend time on refinement.
%
However, as benchmarks get harder,
having more places at ones disposal renders searches faster:
the bounds of 10 and 15 seem to offer a sweet spot.
Our best variant---\mQryTypeBounded$[10]$---solves \tygarQBSolnCount out of \benchmarkCount queries
with the median synthesis time of \firstSolutionTime seconds;
in the rest of this section we use \mQryTypeBounded$[10]$ as the default \tool configuration.

\mQryTypeBounded$[10]$ solves all queries that were solved by \mNoGar
plus six additional queries on which \mNoGar times out.
A closer look at these six queries indicates that they tend to be more complex.
For example, recall that \mNoGar times out on the query \T{appBoth},
while \mQryTypeBounded$[10]$ finds a solution of size four in two seconds.
Generally, our benchmark set is favorable for \mNoGar:
most \hoogle queries are easy, both because of programmers' expectations of what \hoogle can do
and also because we do not know the desired solution,
and hence consider any (relevantly) well-typed solution correct.
The benefits of refinement are more pronounced
on queries with solution size four and higher:
\mQryTypeBounded$[10]$ solves 6 out of 7, while \mNoGar solves only 2.


%


\subsection{Quality of Solutions}

\mypara{Setup}
To evaluate the quality of the solutions,
we ask \tool to return, for each query,
at most \emph{five} well-typed results within a timeout of \qualitytimeout seconds.
Complete results are available in \autoref{appendix:results}.
We then manually inspect the solutions
and for each one determine whether it is \emph{interesting},
\ie whether it is something a programmer might find useful,
based on our own experience as Haskell programmers%
\footnote{Unfortunately, we do not have ground truth solutions for most of our queries,
so we have to resort to subjective analysis.}.
\autoref{fig:results} reports for each query,
the \emph{number} of interesting solutions,
divided by the number of total solutions found within the timeout.
%
%
To evaluate the effects of relevant typing and demand analysis (\autoref{sec:examples:demand}),
we compare three variants of \tool:
\begin{inparaenum}[(1)]
  \item \tool with all features enabled, based on \mQryTypeBounded, labeled \qualityH.
  \item Our tool without the demand analyzer filter, labeled \qualityHD.
  \item Our tool with structural typing in place of the relevant typing, labeled \qualityHR
   %
   (in this variant, the SMT solver is free to choose any non-negative number of tokens
   to assign to each query argument).
\end{inparaenum}

%
%



\mypara{Analysis}
%
First of all, we observe that 
whenever an interesting solution was found by \qualityHD or \qualityHR,
it was also found by \qualityH,
indicating that our filters are not overly conservative.
%
We also observe that on easy queries---taking less than a second---%
demand analysis and relevant typing did little to help:
if an interesting solution were found, then all three variants would find
it and give it a high rank.
However, on medium and hard queries---taking longer than a second---%
the demand analyzer and relevant typing helped promote interesting solutions
higher in rank.
Overall, 66/179 solutions produces by \qualityH were interesting (37\%),
compared with 65/189 for \qualityHD (34\%)
and 26/199 for \qualityHR (13\%).
As you can see, relevant typing is essential to ensure that interesting solutions even get to the top five,
whereas demand analysis is more useful to reduce the total number of solutions the programmer has to sift through.
This is not surprising, since relevant typing mainly filters out \emph{short} programs
while demand analysis is left to deal with \emph{longer} ones.
In our experience, demand analysis was most useful
when queries involved types like \T{Either a b}, where one could
produce a value of type \T{a} from a value of type \T{b}
by constructing and destructing the \T{Either}.
%
One final observation is that in benchmarks 14, 18, 33, and 35,
\qualityHR found fewer results \emph{in total} that the other two versions;
we attribute this to the SMT solver 
struggling with determining the appropriate token multiplicities for the initial marking.

\mypara{Noteworthy solutions}
We presented three illustrative solutions generated by \tool
as examples throughout \autoref{sec:examples}:
\begin{itemize}
\item \T{a -> [Maybe a] -> a} corresponds to benchmark \firstJustPosition (fromFirstMaybes);
the solution from \autoref{sec:examples} is generated at rank \firstJustRank.
\item \T{(a -> a) -> a -> Int -> a} corresponds to benchmark \applyNTimesPosition (applyNTimes);
the solution from \autoref{sec:examples} is generated at rank \applyNTimesRank.
\item \T{Eq a => [(a,b)] -> a -> b} corresponds to benchmark \lookupPosition (lookup);
the solution from \autoref{sec:examples} is generated at rank \lookupRank.
\end{itemize}
\tool has also produced code snippets that surprised us:
for example, on the query \T{(a -> b, a) -> b},
the authors' intuition was to destruct the pair then apply the function.
Instead \tool produces \T{\\x -> uncurry (\$) x} or alternatively \T{\\x -> uncurry id x},
both of which, contrary to our intuition, are not only well-typed,
but also are functionally equivalent to our intended solution.
It was welcome to see a synthesis tool write more succinct code that its authors.

%% file: results.tex
\begin{figure}
  \resizebox{\textwidth}{!}{
  \input{results_data}
  }

\caption{\tool synthesis times and solution quality on \benchmarkCount benchmarks.
We report the total time to first solution,
time spend in the SMT solver,
and time spent type checking (including demand analysis).
`QB10', `Q', `0', `NO' correspond to four variants of the search algorithm:
\mQryTypeBounded$[10]$, \mQryType, \mTopType, and \mNoGar.
All times are in seconds.
Absence indicates no solution found within the timeout of \timeout seconds.
Last three columns report the number of interesting solutions
among the first five 
(or fewer, if fewer solutions were found within the timeout of \qualitytimeout seconds).
`H+', `H-D`, and `H-R` correspond, respectively,
to the default configuration of \tool, 
disabling the demand analyzer,
and using structural typing over relevant typing.
}\label{fig:results}

\end{figure}

%% file: results_data.tex
\begin{tabular}{rll|rrrr|rrrr|rrr|lll} \hline
\multirow{2}{*}{N} & \multirow{2}{*}{Name} & \multirow{2}{*}{Query} & 
\multicolumn{4}{c|}{Time: Total} & \multicolumn{4}{c|}{Time: SMT Solver} & \multicolumn{3}{c|}{Time: Type Checking} &
\multicolumn{3}{c}{\# Interesting / All}\\\cline{4-17}
& & & QB10 & Q & 0 & NO & QB10 & Q & 0 & NO & QB10 & 0 & NO & H+ & H-D & H-R \\ 
 \hline 
1 & firstRight & [Either a b] -\ensuremath{>} Either a b & 0.3 & 0.3 & 0.6 & 0.3 & 0.0 & 0.0 & 0.1 & 0.0 & 0.2 & 0.2 & 0.2 & 2/5 & 2/5 & 2/5 \\ 
2 & firstKey & [(a,b)] -\ensuremath{>} a & 3.9 & 21.2 & 58.4 &  & 2.4 & 17.2 & 52.2 &  & 0.8 & 0.2 &  & 0/2 & 0/4 & 0/3 \\ 
3 & flatten & [[[a]]] -\ensuremath{>} [a] & 1.7 & 5.5 & 1.1 & 0.5 & 0.9 & 2.5 & 0.3 & 0.1 & 0.3 & 0.2 & 0.4 & 5/5 & 5/5 & 0/5 \\ 
4 & repl-funcs & (a-\ensuremath{>}b)-\ensuremath{>}Int-\ensuremath{>}[a-\ensuremath{>}b] & 0.4 & 0.4 & 0.7 & 0.5 & 0.0 & 0.0 & 0.1 & 0.0 & 0.3 & 0.3 & 0.4 & 2/5 & 2/5 & 1/5 \\ 
5 & containsEdge & [Int] -\ensuremath{>} (Int,Int) -\ensuremath{>} Bool & 15.4 & 14.4 & 19.0 & 5.1 & 13.2 & 12.1 & 15.9 & 0.8 & 1.8 & 0.4 & 4.1 & 0/5 & 0/5 & 0/5 \\ 
6 & multiApp & (a -\ensuremath{>} b -\ensuremath{>} c) -\ensuremath{>} (a -\ensuremath{>} b) -\ensuremath{>} a -\ensuremath{>} c & 1.2 & 2.4 & 1.2 & 0.5 & 0.4 & 0.9 & 0.5 & 0.2 & 0.3 & 0.2 & 0.2 & 1/5 & 1/5 & 1/5 \\ 
7 & appendN & Int -\ensuremath{>} [a] -\ensuremath{>} [a] & 0.3 & 0.3 & 0.3 & 0.3 & 0.0 & 0.0 & 0.0 & 0.0 & 0.2 & 0.3 & 0.2 & 2/5 & 2/5 & 0/5 \\ 
8 & pipe & [(a -\ensuremath{>} a)] -\ensuremath{>} (a -\ensuremath{>} a) & 0.7 & 0.6 & 2.1 & 0.7 & 0.1 & 0.1 & 0.6 & 0.1 & 0.2 & 0.7 & 0.6 & 1/5 & 1/5 & 0/5 \\ 
9 & intToBS & Int64 -\ensuremath{>} ByteString & 0.6 & 0.6 & 1.6 & 0.3 & 0.1 & 0.1 & 0.5 & 0.0 & 0.3 & 0.3 & 0.2 & 3/5 & 3/5 & 0/5 \\ 
10 & cartProduct & [a] -\ensuremath{>} [b] -\ensuremath{>} [[(a,b)]] & 1.5 & 8.8 & 1.3 & 1.3 & 0.6 & 5.5 & 0.4 & 0.5 & 0.3 & 0.2 & 0.6 & 0/5 & 0/5 & 0/5 \\ 
11 & applyNtimes & (a-\ensuremath{>}a) -\ensuremath{>} a -\ensuremath{>} Int -\ensuremath{>} a & 6.4 & 23.5 & 0.6 & 1.0 & 4.9 & 19.8 & 0.2 & 0.3 & 1.2 & 0.3 & 0.6 & 0/5 & 0/5 & 0/5 \\ 
12 & firstMatch & [a] -\ensuremath{>} (a -\ensuremath{>} Bool) -\ensuremath{>} a & 1.5 & 1.4 & 2.4 & 0.5 & 0.7 & 0.6 & 1.3 & 0.2 & 0.2 & 0.2 & 0.3 & 5/5 & 5/5 & 5/5 \\ 
13 & mbElem & Eq a =\ensuremath{>} a -\ensuremath{>} [a] -\ensuremath{>} Maybe a & 46.8 &  & 5.6 &  & 45.5 &  & 4.0 &  & 0.8 & 0.3 &  & 0/3 & 0/3 & 0/5 \\ 
14 & mapEither & (a -\ensuremath{>} Either b c) -\ensuremath{>} [a] -\ensuremath{>} ([b], [c]) & 2.6 & 43.7 & 55.4 & 3.5 & 1.7 & 37.6 & 49.8 & 0.5 & 0.3 & 0.2 & 1.7 & 1/4 & 1/5 & 1/1 \\ 
15 & hoogle01 & (a -\ensuremath{>} b) -\ensuremath{>} [a] -\ensuremath{>} b & 0.5 & 0.5 & 1.1 & 0.3 & 0.1 & 0.1 & 0.3 & 0.0 & 0.3 & 0.3 & 0.2 & 2/5 & 2/5 & 2/5 \\ 
16 & zipWithResult & (a-\ensuremath{>}b)-\ensuremath{>}[a]-\ensuremath{>}[(a,b)] & 11.1 &  &  &  & 9.2 &  &  &  & 0.7 &  &  & 1/2 & 1/2 & 0/5 \\ 
17 & splitStr & String -\ensuremath{>} Char -\ensuremath{>} [String] & 0.7 & 0.7 & 1.0 & 0.4 & 0.2 & 0.1 & 0.3 & 0.1 & 0.3 & 0.3 & 0.2 & 0/5 & 0/5 & 0/5 \\ 
18 & lookup & Eq a =\ensuremath{>} [(a,b)] -\ensuremath{>} a -\ensuremath{>} b & 0.7 & 0.7 & 0.7 & 0.8 & 0.2 & 0.2 & 0.2 & 0.3 & 0.3 & 0.3 & 0.3 & 1/5 & 1/3 & 1/4 \\ 
19 & fromFirstMaybes & a -\ensuremath{>} [Maybe a] -\ensuremath{>} a & 1.4 & 3.0 & 3.4 & 0.7 & 0.3 & 0.9 & 1.2 & 0.1 & 0.7 & 0.8 & 0.5 & 2/5 & 2/5 & 0/5 \\ 
20 & map & (a-\ensuremath{>}b)-\ensuremath{>}[a]-\ensuremath{>}[b] & 0.3 & 0.3 & 0.4 & 0.4 & 0.0 & 0.0 & 0.1 & 0.0 & 0.2 & 0.2 & 0.3 & 5/5 & 5/5 & 0/5 \\ 
21 & maybe & Maybe a -\ensuremath{>} a -\ensuremath{>} Maybe a & 0.3 & 0.4 & 0.4 & 0.6 & 0.1 & 0.0 & 0.1 & 0.1 & 0.2 & 0.2 & 0.5 & 2/5 & 1/5 & 0/5 \\ 
22 & rights & [Either a b] -\ensuremath{>} Either a [b] & 1.5 & 31.9 & 11.9 & 0.8 & 0.6 & 20.4 & 5.7 & 0.1 & 0.4 & 0.3 & 0.6 & 1/2 & 1/2 & 1/5 \\ 
23 & mbAppFirst & b -\ensuremath{>} (a -\ensuremath{>} b) -\ensuremath{>} [a] -\ensuremath{>} b & 2.0 & 1.3 & 2.0 & 0.4 & 1.2 & 0.4 & 0.9 & 0.1 & 0.3 & 0.3 & 0.3 & 1/3 & 1/5 & 0/5 \\ 
24 & mergeEither & Either a (Either a b) -\ensuremath{>} Either a b & 2.8 &  &  & 1.0 & 1.7 &  &  & 0.1 & 0.6 &  & 0.7 & 0/3 & 0/3 & 0/5 \\ 
25 & test & Bool -\ensuremath{>} a -\ensuremath{>} Maybe a & 1.4 & 8.8 & 26.4 & 0.7 & 0.7 & 7.1 & 24.3 & 0.3 & 0.2 & 0.3 & 0.3 & 2/5 & 2/5 & 0/5 \\ 
26 & multiAppPair & (a -\ensuremath{>} b, a -\ensuremath{>} c) -\ensuremath{>} a -\ensuremath{>} (b, c) & 2.0 &  &  & 1.5 & 1.2 &  &  & 0.3 & 0.5 &  & 1.0 & 1/2 & 1/4 & 0/5 \\ 
27 & splitAtFirst & a -\ensuremath{>} [a] -\ensuremath{>} ([a], [a]) & 0.6 & 0.6 & 2.3 & 0.4 & 0.1 & 0.1 & 1.1 & 0.1 & 0.3 & 0.3 & 0.2 & 2/5 & 2/5 & 0/5 \\ 
28 & 2partApp & (a-\ensuremath{>}b)-\ensuremath{>}(b-\ensuremath{>}c)-\ensuremath{>}[a]-\ensuremath{>}[c] & 2.3 & 2.2 & 22.9 & 1.5 & 1.2 & 1.2 & 18.7 & 0.5 & 0.2 & 0.3 & 0.3 & 1/5 & 1/5 & 0/5 \\ 
29 & areEq & Eq a =\ensuremath{>} a -\ensuremath{>} a -\ensuremath{>} Maybe a & 44.9 &  &  &  & 40.3 &  &  &  & 3.8 &  &  & 0/2 & 0/5 & 0/5 \\ 
30 & eitherTriple & Either a b -\ensuremath{>} Either a b -\ensuremath{>} Either a b & 5.3 &  &  & 3.2 & 1.9 &  &  & 0.1 & 2.8 &  & 2.9 & 0/5 & 0/5 & 0/5 \\ 
31 & mapMaybes & (a -\ensuremath{>} Maybe b) -\ensuremath{>} [a] -\ensuremath{>} Maybe b & 0.5 & 0.5 & 1.1 & 0.3 & 0.1 & 0.1 & 0.3 & 0.0 & 0.3 & 0.2 & 0.2 & 2/5 & 2/5 & 2/5 \\ 
32 & head-rest & [a] -\ensuremath{>} (a, [a]) & 1.4 & 51.1 & 1.0 & 0.8 & 0.7 & 40.6 & 0.3 & 0.1 & 0.2 & 0.3 & 0.6 & 3/5 & 3/5 & 2/5 \\ 
33 & appBoth & (a -\ensuremath{>} b) -\ensuremath{>} (a -\ensuremath{>} c) -\ensuremath{>} a -\ensuremath{>} (b, c) & 2.1 & 2.8 & 51.1 &  & 1.3 & 1.5 & 44.3 &  & 0.3 & 0.3 &  & 1/5 & 1/5 & 1/1 \\ 
34 & applyPair & (a -\ensuremath{>} b, a) -\ensuremath{>} b & 1.2 & 1.1 & 3.6 & 0.6 & 0.4 & 0.4 & 1.6 & 0.1 & 0.2 & 0.3 & 0.4 & 2/3 & 2/5 & 1/5 \\ 
35 & resolveEither & Either a b -\ensuremath{>} (a-\ensuremath{>}b) -\ensuremath{>} b & 1.0 & 1.3 & 1.5 & 0.5 & 0.4 & 0.5 & 0.6 & 0.2 & 0.2 & 0.2 & 0.2 & 1/5 & 1/2 & 1/5 \\ 
36 & head-tail & [a] -\ensuremath{>} (a,a) & 2.2 &  &  & 20.2 & 1.5 &  &  & 0.4 & 0.3 &  & 18.8 & 0/5 & 0/5 & 0/5 \\ 
37 & indexesOf & ([(a,Int)] -\ensuremath{>} [(a,Int)]) -\ensuremath{>} [a] -\ensuremath{>} [Int] -\ensuremath{>} [Int] &  &  &  &  &  &  &  &  &  &  &  &  &  &  \\ 
38 & app3 & (a -\ensuremath{>} b -\ensuremath{>} c -\ensuremath{>} d) -\ensuremath{>} a -\ensuremath{>} c -\ensuremath{>} b -\ensuremath{>} d & 0.3 & 0.3 & 0.3 & 0.3 & 0.0 & 0.0 & 0.0 & 0.0 & 0.2 & 0.3 & 0.2 & 1/5 & 1/5 & 1/5 \\ 
39 & both & (a -\ensuremath{>} b) -\ensuremath{>} (a, a) -\ensuremath{>} (b, b) & 1.1 &  &  & 1.3 & 0.5 &  &  & 0.2 & 0.3 &  & 1.0 & 1/1 & 1/1 & 0/5 \\ 
40 & takeNdropM & Int -\ensuremath{>} Int -\ensuremath{>} [a] -\ensuremath{>} ([a], [a]) & 0.4 & 0.4 & 1.3 & 0.4 & 0.0 & 0.0 & 0.4 & 0.0 & 0.3 & 0.3 & 0.3 & 5/5 & 5/5 & 0/5 \\ 
41 & firstMaybe & [Maybe a] -\ensuremath{>} a & 1.2 & 1.6 & 1.4 & 0.7 & 0.5 & 0.6 & 0.4 & 0.1 & 0.2 & 0.2 & 0.5 & 4/5 & 4/5 & 2/5 \\ 
42 & mbToEither & Maybe a -\ensuremath{>} b -\ensuremath{>} Either a b & 47.4 &  &  &  & 21.7 &  &  &  & 24.2 &  &  & 0/2 & 0/5 & 0/5 \\ 
43 & pred-match & [a] -\ensuremath{>} (a -\ensuremath{>} Bool) -\ensuremath{>} Int & 1.1 & 1.1 & 3.6 & 0.4 & 0.4 & 0.4 & 2.0 & 0.1 & 0.3 & 0.3 & 0.2 & 3/5 & 3/5 & 3/5 \\ 
44 & singleList & Int -\ensuremath{>} [Int] & 0.3 & 0.3 & 0.4 & 0.3 & 0.0 & 0.0 & 0.0 & 0.0 & 0.2 & 0.2 & 0.3 & 1/5 & 1/5 & 0/5 \\ 
\hline 
 \end{tabular}

%% file: initialcover.tex
\begin{figure}[t!]
    \begin{minipage}{0.49\textwidth}
        \centering
    \includegraphics[width=1\columnwidth]{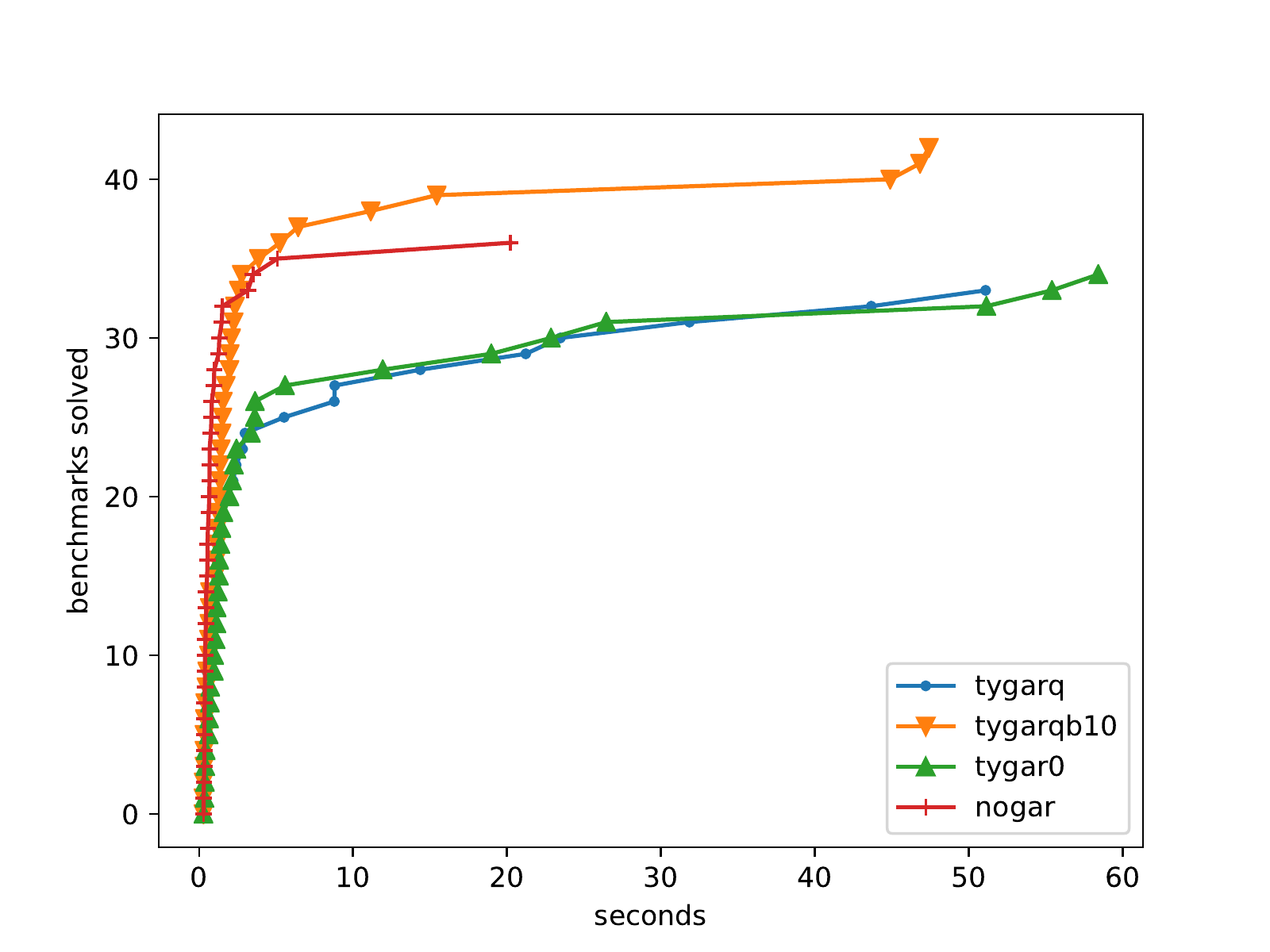}
    \caption{Queries solved over time for our initial variants
    and the best refinement bound.
    }\label{fig:searchgraph}
    \end{minipage}
    \begin{minipage}{0.49\textwidth}
        \centering
    \includegraphics[width=\columnwidth]{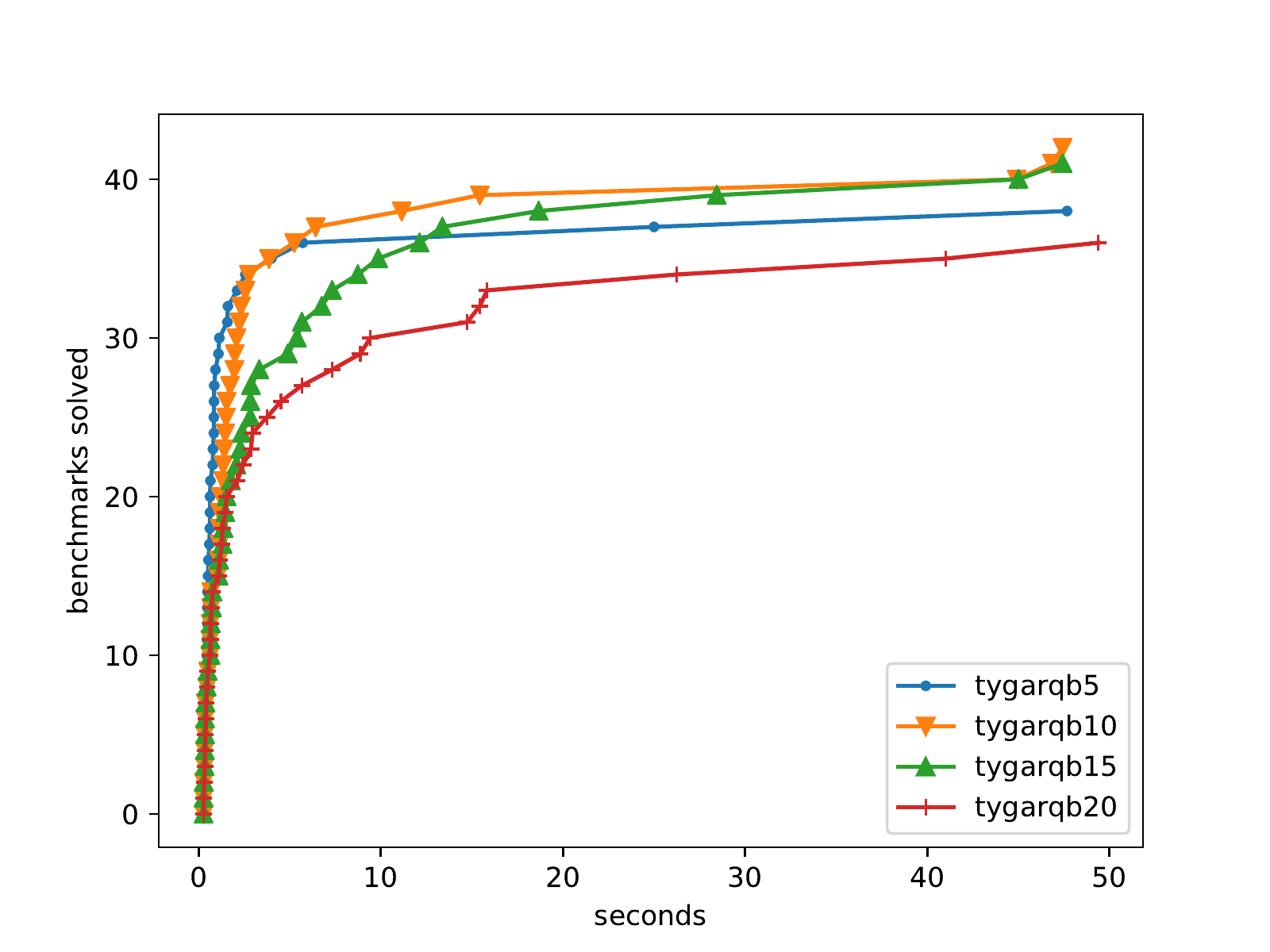}
    \caption{Queries solved over time for varying refinement bounds.
    The variant's number indicates the refinement bound on the abstract cover.
    }\label{fig:boundgraph}
    \end{minipage}
\end{figure}

%% file: related.tex
\section{Related Work}\label{sec:related}

Finally, we situate our work with other research into ways of 
synthesizing code that meets a given specification. 
For brevity, we restrict ourselves to the (considerable) 
literature that focuses on using \emph{types} as specifications, 
and omit discussing methods that use \eg input-output examples 
or tests \cite{Gulwani11,Katayama12,Lee18,OseraZ15}, 
logical specifications \cite{Srivastava10,CodeHint2014} 
or program sketches \cite{Solar-Lezama:2008}.

\mypara{API Search} 
Modern IDEs support various forms of code-completion,
based on at the very least common prefixes of names 
(\eg completing \T{In} into \T{Integer} or \T{fo} 
into \T{foldl'}) and so on. 
Many tools use type information to only return completions 
that are well-typed at the point of completion.
This approach is generalized by search based tools 
like \hoogle \cite{Hoogle} that search for type 
isomorphisms \cite{DicosmoJFP92} to find functions 
that ``match'' a given type signature (query).
The above can be viewed as returning single-component results, 
as opposed to our goal 
of searching for terms that \emph{combine} components 
in order to satisfy a given type query.

\mypara{Search using Statistical Models} 
Several groups have looked into using statistical methods to 
improve search-based code-completion. 
One approach is to analyze large code bases to 
precompute statistical models that can be used 
to predict the \emph{most likely} sequences of 
method calls at a given point or that yield values 
of a given (first order) type \cite{Raychev14}.
It is possible to generalize the above to train 
probabilistic models (grammars) that \emph{generate} 
the most likely programs that must contain certain 
properties like method names, types, or 
keywords~\cite{Bayou17}.
We conjecture that while the above methods are 
very useful for effectively searching for commonly 
occurring code snippets,
they are less useful in functional languages,
where higher-order components offer high degree of compositionality
and lead to less code repetition.
%

\mypara{Type Inhabitation} 
The work most directly related to ours are methods based 
on finding terms that \emph{inhabit} a (query) 
type \cite{urzyczyn97}.
One approach is to use the correspondence between types 
and logics, to reduce the inhabitation question to that 
of validity of a logical formula (encoding the type).
A classic example is \djinn \cite{djinn} which implements 
a decision procedure for intuitionistic propositional 
calculus \cite{Dyckhoff98} to synthesize terms that 
have a given type.
Recent work by Rehof \etal extends the notion of inhabitation 
to support object oriented frameworks whose components behaviors 
can be specified via intersection types \cite{rehof16}.
However, both these approaches lack a \emph{relevancy} 
requirement of its snippets, and hence return undesirable 
results.
For example, when queried with a type \T{a -> [a]}, 
\djinn would yield a function that always returns the empty list. 
One way to avoid undesirable results 
is to use dependent or refinement types to capture the semantics of the desired terms more precisely.
%
\synquid~\cite{PolikarpovaKS16} and \myth~\cite{FrankleOWZ16}
use different flavors of refinement types to synthesize recursive functions,
while \agda \cite{Norell08} 
makes heavy use of proof search to enable type- or hole-driven 
development. 
However, unlike \tool, methods based on classical proof 
search do not scale up to large component libraries. 

\mypara{Scalable Proof Search} 
One way to scale search is explored by \cite{Perelman12} 
which uses a very restricted form of inhabitation queries 
to synthesize local ``auto-completion'' terms corresponding 
to method names, parameters, field lookups and so on, but 
over massive component libraries (\eg the .NET framework).
In contrast, the \insynth system \cite{GveroKKP13} addresses 
the problem of scalability by extending proof search with 
a notion of \emph{succinctness} that collapses types 
into equivalence classes, thereby abstracting the space 
over which proof search must be performed.
Further, \insynth uses \emph{weights} derived from 
empirical analysis of library usage to bias the search 
to more likely results.
However, \insynth is limited to simple types \ie does 
not support parametric polymorphism which is the focus 
of our work.

\mypara{Graph Reachability} 
Our approach is directly inspired by methods that reduce 
the synthesis problem to some form of \emph{reachability}.
\prospector~\cite{Mandelin05} is an early exemplar 
where the components are \emph{unary} functions that take 
a single input. Consequently, the component library can be 
represented as a directed graph of edges between input and 
output types, and synthesis is reduced to finding a 
path from the query's input type to its output type.
\sypet~\cite{FengM0DR17}, which forms the basis of our work,
is a generalization of \prospector to account for general 
first-order functions which can take multiple inputs, thereby 
generalizing synthesis to reachability on Petri nets.
The key contribution of our work is the notion of \tygar 
that generalizes \sypet's approach to polymorphic and 
higher-order components.

\mypara{Counterexample-Guided Abstraction Refinement} 
While the notion of counterexample-guided abstraction 
refinement (CEGAR) is classical at this point 
\cite{Kurshan10}, 
there are two lines of work in particular closely related to ours.
First, \cite{Ganty07,Kloos13} describe an iterative 
abstraction-refinement process for verifying Petri nets, 
using SMT~\cite{Esparza14}. 
However, in their setting, the refinement loop is used 
to perform unbounded verification of the (infinite-state) 
Petri net.
In contrast, \tool performs a bounded search on each 
Petri net, but uses \tygar to refine the net itself 
with new type instantiations that eliminate the 
construction of ill-typed terms.
Second, \blaze~\cite{WangDS18} describes a CEGAR approach 
for synthesizing programs from input-output examples, by 
iteratively refining \emph{finite tree-automata}
whose states correspond to values in a predicate-abstraction domain.
Programs that do not satisfy 
the input-output tests are used to iteratively refine the domain 
until a suitable correct program is found.
Our approach differs in that we aim to synthesize terms 
of a given \emph{type}.
Consequently, although our refinement mechanism is inspired by \blaze,
we develop a novel abstract domain---%
a finite sub-lattice of the type subsumption lattice---%
and show how to use proofs of \emph{untypeability} to refine this domain.
Moreover, we show how CEGAR can be combined with Petri nets
(as opposed to tree automata)
in order to enforce relevancy.

\mypara{Types and Abstract Interpretation}
The connection between types and abstract interpretation (AI)
was first introduced in~\cite{Cousot1997}.
The goal of their work, however, was to cast \emph{existing} type systems in the framework of AI,
while we use this framework to systematically construct \emph{new} type systems 
that further abstract an existing one.
More recently, \cite{GarciaCT16} used the AI framework to formalize \emph{gradual typing}.
Like that work, we use AI to derive an abstract type system for our language,
but otherwise the goals of the two techniques are very different.
Moreover, as we hint in \autoref{sec:algo:refine},
our abstract domain is subtly but crucially different from traditional gradual typing,
because our refinement algorithm relies on non-linear terms 
(\ie types with repeated variables).

%% file: conclusions.tex
\section{Conclusions}\label{sec:concl}

We have presented \tygar, a new algorithm for synthesizing 
terms over polymorphic components that inhabit a given 
\emph{query} type. 
The key challenge here is the \emph{infinite} space of 
monomorphic instances arising from the polymorphic signatures. 
We introduced a new notion of \emph{abstract typing} that allows 
us to use ideas from the framework of abstract interpretation 
to compute a finite overapproximation of this search space.
We then show how spurious terms that are well-typed 
in the abstract domain but ill-typed in reality, yield 
\emph{proofs of untypeability} which can then iteratively 
refine the abstract search space until a well-typed 
solution is found.

We have implemented \tygar in \tool, and our evaluation 
on a suite of \benchmarkCount queries demonstrates the benefits 
of counterexample-driven refinement.
In particular, we show how lazily refining a coarse abstract 
domain using proofs of untypeability allows us to synthesize 
correct results \emph{faster} than a naive approach that 
eagerly instantiates all the types from the query followed 
by a brute-force enumeration.
Our experiments further demonstrate that the gains from 
iterative refinement over enumeration are 
even more pronounced on harder queries over more complex 
types.
Our support for polymorphism also allows \tool to work 
with higher-order and type-class constrained components,
which, thanks to parametricity, allows for more precise 
queries than simple monomorphic types.

In future work it would be valuable to investigate 
ways to improve the quality of the results, \eg by 
prioritizing components that are more popular, or 
less partial, or by extending our method to use 
other forms of specifications such as examples 
or refinement types.

%% file: appendix.tex
\section{Proofs}\label{appendix:props}

\subsection{Type Transformers}

\begin{lemma}[Monotonicity of Substitution]\label{ap:sub-monotone}
If $\sigma' \sub \sigma$ then $\sapp{\sigma'}{B} \sub \sapp{\sigma}{B}$.
\end{lemma}

\begin{lemma}[Monotonicity of Unification]\label{ap:mgu-monotone}
If $B_1 \sub B_2$ then $\mgu{B,B_1} \sub \mgu{B,B_2}$
\end{lemma}

\begin{lemma}[Monotonicity of Type Transformers]\label{ap:trans-monotone}
For any component $c$ and any types $\many{B^1_i}$ and $\many{B^2_i}$,
such that $\many{B^1_i \sub B^2_i}$,
we have $\sem{c}(\many{B^1_i}) \sub \sem{c}(\many{B^2_i})$.
\end{lemma}

\begin{proof}
    Let $\fresh{\ty{c}} = \many{B_i} \to B$.
    By the definition of type transformers, we have
    \begin{align*}
        \sem{c}(\many{B^1_i}) &= \sapp{\sigma_1}{B} & \sigma_1 = \mgu{\many{B_i, B^1_i}} \\
        \sem{c}(\many{B^2_i}) &= \sapp{\sigma_2}{B} & \sigma_2 = \mgu{\many{B_i, B^2_i}}
    \end{align*}
    Since $\many{B^1_i \sub B^2_i}$,
    we have $\sigma_1 \sub \sigma_2$ by \autoref{ap:mgu-monotone},
    and hence $\sapp{\sigma_1}{B} \sub \sapp{\sigma_2}{B}$ by \autoref{ap:sub-monotone}.
\end{proof}

\begin{lemma}[Soundness of Type Transformers]\label{ap:trans-sound}
If $\fresh{\ty{c}} = \many{B_i} \to B$ and
$\many{\sapp{\sigma}{B_i} \sub B'_i}$
then $\sapp{\sigma}{B} \sub \sem{c}(\many{B'_i})$.
\end{lemma}

\begin{proof}
    By \autoref{ap:trans-monotone}, since $\many{\sapp{\sigma}{B_i} \sub B'_i}$,
    we have $\sem{c}(\many{\sapp{\sigma}{B_i}}) \sub \sem{c}(\many{B'_i})$.
    But $\sem{c}(\many{\sapp{\sigma}{B_i}}) = \sapp{\sigma}{B}$,
    because $\mgu{B, \sapp{\sigma}{B}}\equiv \sigma$.
    Hence $\sapp{\sigma}{B} \sub \sem{c}(\many{B'_i})$ as desired.
\end{proof}

\subsection{Algorithmic Typing}

\begin{figure}
\textbf{Typing}\quad$\boxed{\jtyping{\Gamma}{\nex}{t}}$
\small
\begin{gather*}
\inference[\textsc{T-Var}]
{\Gamma(x) = b}
{\jtyping{\Gamma}{x}{b} }
\\
\inference[\textsc{T-App}]
{ \ty{c} = \many{\forall\tau}.T &
  \sapp{\sigma}{T} = \many{b_i}\to b &
  \many{\jtyping{\Gamma}{e_i}{b_i}}}
{\jtyping{\Gamma}{\eapp{c}{\many{e_i}}}{b}}
\\
\inference[\textsc{T-Fun}]
{ \jtyping{\Gamma, x:b}{\nex}{t} }
{\jtyping{\Gamma}{\elam{x}{\nex}}{b \to t}}
\end{gather*}
\caption{Uncurried version of declarative typing.}\label{fig:typing-uncurry}
\end{figure}

In the following, we use a version the declarative type system in \autoref{fig:typing-uncurry},
which is defined over uncurried applications,
and hence more closely matches algorithmic typing.
It is straightforward to show that for terms in $\eta$-long form,
this type system is equivalent to the one in \autoref{fig:lang}.

\begin{lemma}[Soundness of inference]\label{ap:aux-left}
  If $\jtinfer{\Gamma}{e}{B}$ and $b \sub B$, then $\jtyping{\Gamma}{e}{b}$.
\end{lemma}

\begin{proof}
    By induction on the derivation of $\jtinfer{\Gamma}{e}{B}$.
    \begin{itemize}
       \item Case \textsc{I-Var}: Given the conclusion $\jtinfer{\Gamma}{x}{b'}$,
       we get to assume $\Gamma(x) = b'$.
       Consequently, applying \textsc{T-Var}, we get $\jtyping{\Gamma}{x}{b'}$.
       But since $b'$ is ground and $b \sub b'$, we have $b' = b$.

       \item Case \textsc{I-App}: Given the conclusion \jtinfer{\Gamma}{\eapp{c}{\many{e_i}}}{\sem{c}(\many{B_i})},
       we get to assume $\many{\jtinfer{\Gamma}{e_i}{B_i}}$.
       Assume $\fresh{\ty{c}} = \many{B'_i}\to B'$.
       By definition of $\sem{c}$ we have $\sem{c}(\many{B_i}) = \sapp{\sigma_1}{B'}$ where $\sigma_1 = \mgu{\many{B'_i,B_i}}$.
       By the assumption of the theorem $b \sub B = \sem{c}(\many{B_i})$,
       and hence there exists $\sigma_2$ such that
       $b = \sapp{\sigma_2}{(\sem{c}(\many{B_i}))} = \sapp{(\sigma_2\circ\sigma_1)}{B'}$.

       Let $\rho$ be some substitution such that $\sapp{(\rho\circ\sigma_2\circ\sigma_1)}{B'_i}$,
       is ground for all for all $B'_i$;
       we will denote these ground types $b_i$.
       Then each $b_i = \sapp{(\rho\circ\sigma_2)}{(\sapp{\sigma_1}{B'_i})} = \sapp{(\rho\circ\sigma_2)}{(\sapp{\sigma_1}{B_i})}$
       (since $\sigma_1$ is a unifier),
       and hence $b_i \sub B_i$.
       Then by the IH, we can show for all $e_i$:
       \jtyping{\Gamma}{e_i}{b_i} (1).

       Finally, let us use \textsc{T-App} to construct the derivation of
       $\jtyping{\Gamma}{\eapp{c}{\many{e_i}}}{b}$
       The middle premise follows with $\sigma = \rho\circ\sigma_2\circ\sigma_1$
       (note that $\sapp{\sigma}{B'} = \sapp{\rho}{b} = b$,
       since $b$ is already ground).
       We get the last premise by (1).
    \end{itemize}
   \end{proof}

\begin{lemma}[Completeness of Inference]\label{ap:aux-right}
    If $\jtyping{\Gamma}{e}{b}$, then $\exists B . \jtinfer{\Gamma}{e}{B}$ and $b \sub B$.
\end{lemma}

\begin{proof}
    By induction on the derivation of $\jtyping{\Gamma}{e}{b}$.
    \begin{itemize}
    \item Case \textsc{T-Var}:
    Given the conclusion $\jtyping{\Gamma}{x}{b}$, we get to assume $\Gamma(x) = b$.
    Using \textsc{I-Var} we build the derivation $\jtinfer{\Gamma}{x}{b}$.
    Trivially $b \sub b$ by reflexivity of $\sub$.


    \item Case \textsc{T-App}:
    Given the conclusion \jtyping{\Gamma}{\eapp{c}{\many{e_i}}}{b},
    we get to assume for each $i$: \jtyping{\Gamma}{e_i}{b_i},
    where $\many{b_i}\to b = \sapp{\sigma}{T}$ and $T = \many{B'_i}\to B' = \fresh{\ty{c}}$.
    By IH we obtain $\many{B_i}$,
    such that \jtinfer{\Gamma}{e_i}{B_i} and $b_i \sub B_i$, for each $i$.
    Using \textsc{I-Var} we build the derivation $\jtinfer{\Gamma}{\eapp{c}{\many{e_i}}}{\sem{c}(\many{B_i})}$.
    It remains to show that $b \sub \sem{c}(\many{B_i})$.

    Since $\many{b_i \sub B_i}$ and $b_i = \sapp{\sigma}{B'_i}$,
    we get $\many{\sapp{\sigma}{B'_i} \sub B_i}$.
    Hence we can apply \autoref{ap:trans-sound} to conclude
    $\sapp{\sigma}{B'} \sub \sem{c}(\many{B_i})$,
    but $\sapp{\sigma}{B'} = b$, so we obtain the desired conclusion.

    \item Case \textsc{T-Fun}: impossible since $b$ is a base type.
    \end{itemize}
\end{proof}

\begin{lemma}[Soundness of Checking]\label{ap:algo-sound}
If \jtcheck{\Gamma}{\nex}{t}, then \jtyping{\Gamma}{\nex}{t}.
\end{lemma}
\begin{proof}
By induction on the derivation of \jtcheck{\Gamma}{E}{t}.
\begin{itemize}
   \item Case \textsc{C-Base}: Given the conclusion \jtcheck{\Gamma}{e}{b},
   we get to assume \jtinfer{\Gamma}{e}{B} and $b \sub B$.
   By \autoref{ap:aux-left}, we directly obtain \jtyping{\Gamma}{e}{b}.

   \item Case \textsc{C-Fun}: Given the conclusion \jtcheck{\Gamma}{\elam{x}{\nex}}{b \to t},
   we get to assume \jtcheck{\Gamma, x:b}{\nex}{t}.
   By IH, we get \jtyping{\Gamma, x:b}{\nex}{t}.
   Hence we use \textsc{T-Fun} to obtain \jtyping{\Gamma}{\elam{x}{e}}{b \to t}.
\end{itemize}
\end{proof}

\begin{lemma}[Completeness of Checking]\label{ap:algo-complete}
If \jtyping{\Gamma}{\nex}{t} then \jtcheck{\Gamma}{\nex}{t}.
\end{lemma}
\begin{proof}
By induction on the derivation of $\jtyping{\Gamma}{E}{t}$.
\begin{itemize}
\item Case \textsc{T-Var}, \textsc{T-App}:
Given $\jtyping{\Gamma}{e}{b}$, we use \autoref{ap:aux-right}
to derive $\jtinfer{\Gamma}{e}{B}$ and $b \sub B$.
Hence we apply \textsc{C-Base} to obtain $\jtcheck{\Gamma}{e}{b}$.

\item Case \textsc{T-Fun}:
Given the conclusion \jtyping{\Gamma}{\elam{x}{\nex}}{b \to t}
we get to assume \jtyping{\Gamma, x:b}{\nex}{t}.
By IH $\jtcheck{\Gamma, x:b}{E}{t}$.
We apply \textsc{C-Fun} to obtain \jtcheck{\Gamma}{\elam{x}{e}}{b \to t}.
\end{itemize}
\end{proof}

\begin{theorem}[Type Checking is Sound and Complete]\label{ap:algo-sound-complete}
\jtyping{\cdot}{\nex}{t} iff \jtcheck{\cdot}{\nex}{t}.
\end{theorem}
\begin{proof}
By \autoref{ap:algo-sound} and \autoref{ap:algo-complete}.
\end{proof}

\subsection{Abstract Typing}

\begin{lemma}[Refinement]\label{ap:refine}
If $B' \sub B$ and $\abset' \refines \abset$, then $\alpha_{\abset'}(B') \sub \alpha_{\abset}(B)$.
\end{lemma}
\begin{proof}
Let $A = \alpha_{\abset}(B)$, $A' = \alpha_{\abset'}(B')$.
Because $\abset \subseteq \abset'$, we have $A \in \abset'$.
Since $B' \sub B \sub A$, we have $A' \sub A$ by definition of $\alpha_{\abset'}$.
\end{proof}

\begin{lemma}[Inference Refinement]\label{ap:refine-infer}
If $\abset' \refines \abset$
and \jainfer{\Gamma}{e}{B'}{\abset'},
then \jainfer{\Gamma}{e}{B}{\abset} and $B' \sub B$.
\end{lemma}
\begin{proof}
  By induction on the derivation of \jainfer{\Gamma}{e}{B'}{\abset'}.
    \begin{itemize}
        \item Case \textsc{I-Var}: Follows by \autoref{ap:refine}.

        \item Case \textsc{I-App}: Given the conclusion \jainfer{\Gamma}{\eapp{c}{\many{e_i}}}{\alpha_{\abset'}\left(\sem{c}(\many{B'_i})\right)}{\abset'},
        we get to assume \jainfer{\Gamma}{e_i}{B'_i}{\abset'} for all $i$.
        By IH we get \jainfer{\Gamma}{e_i}{B_i}{\abset} and $B'_i \sub B_i$.
        We use \textsc{I-App} to obtain \jainfer{\Gamma}{\eapp{c}{\many{e_i}}}{\alpha_{\abset}\left(\sem{c}(\many{B_i})\right)}{\abset}.
        Furthermore, by \autoref{ap:trans-monotone}, we get $\sem{c}(\many{B'_i}) \sub \sem{c}(\many{B_i})$,
        and hence by \autoref{ap:refine}, $\alpha_{\abset'}\left(\sem{c}(\many{B'_i})\right) \sub \alpha_{\abset}\left(\sem{c}(\many{B_i})\right)$.
    \end{itemize}
\end{proof}

\begin{theorem}[Typing Preservation]\label{ap:refine-check}
If $\abset' \refines \abset$
and $\jacheck{\Gamma}{\nex}{t}{\abset'}$
then $\jacheck{\Gamma}{\nex}{t}{\abset}$.
\end{theorem}

\begin{proof}
    By induction on the derivation of $\jacheck{\Gamma}{\nex}{t}{\abset'}$.

    \begin{itemize}
        \item Case \textsc{C-Base}:
        Given the conclusion \jacheck{\Gamma}{e}{b}{\abset'},
        we get to assume \jainfer{\Gamma}{e}{B'}{\abset'} and $b \sub B'$.
        Then by \autoref{ap:refine-infer} we get \jainfer{\Gamma}{e}{B}{\abset} and $B' \sub B$, hence $b \sub B$.
        We use \textsc{C-Base} to conclude \jacheck{\Gamma}{e}{b}{\abset}.

        \item Case \textsc{C-Fun}:
        Given the conclusion \jacheck{\Gamma}{\elam{x}{\nex}}{b \to t}{\abset'}
        we get to assume $\jacheck{\Gamma, x:b}{\nex}{t}{\abset'}$.
        By IH we get \jacheck{\Gamma, x:b}{\nex}{t}{\abset}.
        We use \textsc{C-Fun} to colclude \jacheck{\Gamma}{\elam{x}{\nex}}{b \to t}{\abset}.
    \end{itemize}
\end{proof}



\subsection{Synthesis}

\begin{theorem}
The synthesis problem in \corelang is undecidable.
\end{theorem}
\begin{proof}
By reduction from the Post Correspondence Problem (PCP).
Let $[(a_1, b_1), \ldots, (a_n, b_n)]$ be an instance of the PCP,
where each $a_i$, $b_i$ are bit strings.
We translate this instance into a synthesis problem $(\Lambda, T)$ as follows.

The set of \emph{type constructors} in $\Lambda$ is:
two nullary constructors \T{Start} and \T{Goal},
two unary constructors \T{T} and \T{F},
and a single binary constructor \T{P}.
Given a bit string $bs$ and a type $T$,
the type $\wrap{T}{bs}$ is defined as follows:
\begin{align*}
\wrap{T}{[]} &= T \\
\wrap{T}{0:bs} &= \wrap{\T{F}\ T}{bs}\\
\wrap{T}{1:bs} &= \wrap{\T{T}\ T}{bs}
\end{align*}
Now we can define the set of components in $\Lambda$:
\begin{itemize}
\item for each pair $(a_i, b_i)$, there is a component $s_i$
of type $\T{Start} \to \T{P}\ \wrap{\T{Start}}{a_i}\ \wrap{\T{Start}}{b_i}$;
\item for each pair $(a_i, b_i)$, there is a component $n_i$
of type $\forall \alpha \beta . \T{P}\ \alpha\ \beta \to \T{P}\ \wrap{\alpha}{a_i}\ \wrap{\beta}{b_i}$;
\item a component $f$ of type $\forall \alpha . \T{P}\ \alpha\ \alpha \to \T{Goal}$.
\end{itemize}
The query type $T$ is $\T{Start} \to \T{Goal}$.
The sequence of \T{P} instances in the solution to this synthesis problem
corresponds to the solution to the PCP.
\end{proof}

For example, a PCP instance $[(1, 101), (10, 00), (011, 11)]$
gives rise to the following components:
\begin{lstlisting}
s1 :: Start -> P (T Start) (T (F (T Start)))
s2 :: Start -> P (F (T Start)) (F (F Start))
s3 :: Start -> P (T (T (F Start))) (T (T Start))
n1 :: P _a _b -> P (T _a) (T (F (T _b)))
n2 :: P _a _b -> P (F (T _a)) (F (F _b))
n3 :: P _a _b -> P (T (T (F _a))) (T (T _b))
f :: P _a _a -> Goal
\end{lstlisting}

The solution for this instance is \T{\\x -> f (n3 (n2 (n3 (s1 x))))}.

\subsection{Proof of Untypeability}

For a $P\colon \mathbf{e} \to \basebots$ and a term $e^*$, define
\begin{description}
\item[$\inv_1$:] For any $e \in \subterms{e^*}$, if \jtinfer{\Gamma}{e}{B}, then $B \sub P[e]$;
\item[$\inv_2$:] For any $e = \eapp{c}{\many{e_j}}\in \subterms{e^*}$,
$\sem{c}(\many{P[e_j]}) \sub P[e]$;
\item[$\inv_3$:] $P[e^*] = \bot$.
\end{description}

\begin{lemma}[Precision of Inference]\label{ap:invp-infer}
If $\inv_1 \wedge \inv_2$, $\rng{P}\subseteq \abset$, $e\in \subterms{e^*}$, and \jainfer{\Gamma}{e}{B}{\abset},
then $B \sub P[e]$.
\end{lemma}
\begin{proof}
    By induction on the derivation of \jainfer{\Gamma}{e}{B}{\abset}.
    \begin{itemize}
        \item Case \textsc{I-Var}:
        Given the conclusion \jainfer{\Gamma}{x}{\alpha_{\abset}(b)}{\abset},
        we get to assume $\Gamma(x) = b$, and hence \jtinfer{\Gamma}{x}{b}.
        By $\inv_1$, we have $b \sub P[x]$.
        \item Case \textsc{I-App}:
        Given the conclusion \jainfer{\Gamma}{\eapp{c}{\many{e_i}}}{\alpha_{\abset}(\sem{c}(\many{B_i}))}{\abset},
        we get to assume \jainfer{\Gamma}{e_i}{B_i}{\abset} for each $i$.
        By IH, we have $B_i \sub P[e_i]$ for each $i$.
        Then by \autoref{ap:trans-monotone} we have $\sem{c}(\many{B_i}) \sub \sem{c}(\many{P[e_i]})$,
        and by $\inv_2$ we have $\sem{c}(\many{P[e_i]})\sub P[e]$;
        hence $\sem{c}(\many{B_i}) \sub P[e]$.
        Applying $\alpha_{\abset}$ to both sides, by \autoref{ap:refine} we have
        $\alpha_{\abset}(\sem{c}(\many{B_i})) \sub \alpha_{\abset}(P[e])$.
        But $\alpha_{\abset}(P[e]) = P[e]$ because $\rng{P}\subseteq \abset$,
        hence we get $\alpha_{\abset}(\sem{c}(\many{B_i})) \sub P[e]$ as required.

    \end{itemize}
\end{proof}

\begin{lemma}[Untypeability]\label{ap:invp}
Let $\nex = \elam{\many{x_i}}{e}$, $t = \many{b_i}\to b$, $e^* = \eapp{r}{e}$, where $\ty{r} = b\to b$.
If $\inv_1 \wedge \inv_2 \wedge \inv_3$ and $\rng{P}\subseteq \abset$,
then \jancheck{\cdot}{\nex}{t}{\abset}.
\end{lemma}
\begin{proof}
Assume the contrary: that we can derive \jacheck{\cdot}{\nex}{t}{\abset}.
Then we must be able to derive \jacheck{\Gamma}{e}{b}{\abset}, where $\Gamma = \many{x_i\colon b_i}$.
By \textsc{C-Base}, if we derive \jainfer{\Gamma}{e}{B}{\abset}, it must be that $b \sub B$
(note that inference result always exists and is unique).

By \autoref{ap:invp-infer}, $\jainfer{\Gamma}{e^*}{B^*}{\abset} \sub P[e^*]$,
but by $\inv_3$ we have $P[e^*] = \bot$, hence $B^* = \bot$.
By \textsc{I-App} and considering $\ty{r}$, we have $B^* = \sapp{\sigma}{b}$, where $\sigma = \mgu{B, b}$.
The only $\sigma$ such that $\sapp{\sigma}{b} = \bot$ is $\sigma_{\bot}$,
hence we get $\mgu{B, b} = \sigma_{\bot}$.
But then it cannot be that $b \sub B$, a contradiction.
\end{proof}

\newpage
\section{Evaluation Results}\label{appendix:results}

\input{quality_qb_results}

%% file: quality_qb_results.tex
\textbf{appBoth:\ \T{(a -> b) -> (a -> c) -> a -> (b, c)}}\\[.5em]
\framebox{\parbox{\textwidth}{
 Demand Analysis\\[.25em]
\T{(,) (arg2 arg0) (arg1 arg0)}\\
\T{(,) ((\$) arg2 arg0) (arg1 arg0)}\\
\T{(,) (arg2 (fromJust Nothing)) (arg1 arg0)}\\
\T{(,) (arg2 (head [])) (arg1 arg0)}\\
\T{(,) (arg2 (last [])) (arg1 arg0)}}}\\[.5em]
\framebox{\parbox{\textwidth}{
 No Demand Analysis\\[.25em]
\T{(,) (arg2 arg0) (arg1 arg0)}\\
\T{(,) ((\$) arg2 arg0) (arg1 arg0)}\\
\T{(,) (arg2 (fromJust Nothing)) (arg1 arg0)}\\
\T{(,) (arg2 (head [])) (arg1 arg0)}\\
\T{(,) (arg2 (last [])) (arg1 arg0)}}}\\[.5em]
\framebox{\parbox{\textwidth}{
 No Relevancy\\[.25em]
\T{(,) (arg2 arg0) (arg1 arg0)}}}\\[.5em]
\textbf{countFilter:\ \T{(a -> Bool) -> [a] -> Int}}\\[.5em]
\framebox{\parbox{\textwidth}{
 Demand Analysis\\[.25em]
\T{length (map arg1 arg0)}\\
\T{length (dropWhile arg1 arg0)}\\
\T{length (filter arg1 arg0)}\\
\T{length (takeWhile arg1 arg0)}\\
\T{length (repeat (arg1 (head arg0)))}}}\\[.5em]
\framebox{\parbox{\textwidth}{
 No Demand Analysis\\[.25em]
\T{length (map arg1 arg0)}\\
\T{length (dropWhile arg1 arg0)}\\
\T{length (filter arg1 arg0)}\\
\T{length (takeWhile arg1 arg0)}\\
\T{length (repeat (arg1 (head arg0)))}}}\\[.5em]
\framebox{\parbox{\textwidth}{
 No Relevancy\\[.25em]
\T{length (map arg1 arg0)}\\
\T{length (dropWhile arg1 arg0)}\\
\T{length (filter arg1 arg0)}\\
\T{length (takeWhile arg1 arg0)}\\
\T{length (repeat arg1)}}}\\[.5em]
\textbf{mbElem:\ \T{Eq a => a -> [a] -> Maybe a}}\\[.5em]
\framebox{\parbox{\textwidth}{
 Demand Analysis\\[.25em]
\T{lookup arg1 (zip arg0 [])}\\
\T{lookup arg1 (zip [] arg0)}\\
\T{lookup arg1 (zip arg0 arg0)}}}\\[.5em]
\framebox{\parbox{\textwidth}{
 No Demand Analysis\\[.25em]
\T{lookup arg1 (zip arg0 [])}\\
\T{lookup arg1 (zip [] arg0)}\\
\T{lookup arg1 (zip arg0 arg0)}}}\\[.5em]
\framebox{\parbox{\textwidth}{
 No Relevancy\\[.25em]
\T{Just arg1}\\
\T{listToMaybe arg0}\\
\T{Nothing}\\
\T{listToMaybe (cycle arg0)}\\
\T{listToMaybe (init arg0)}}}\\[.5em]
\textbf{hoogle01:\ \T{(a -> b) -> [a] -> b}}\\[.5em]
\framebox{\parbox{\textwidth}{
 Demand Analysis\\[.25em]
\T{arg1 (head arg0)}\\
\T{arg1 (last arg0)}\\
\T{(\$) arg1 (head arg0)}\\
\T{(\$) arg1 (last arg0)}\\
\T{arg1 (head (cycle arg0))}}}\\[.5em]
\framebox{\parbox{\textwidth}{
 No Demand Analysis\\[.25em]
\T{arg1 (head arg0)}\\
\T{arg1 (last arg0)}\\
\T{(\$) arg1 (head arg0)}\\
\T{(\$) arg1 (last arg0)}\\
\T{arg1 (head (cycle arg0))}}}\\[.5em]
\framebox{\parbox{\textwidth}{
 No Relevancy\\[.25em]
\T{arg1 (head arg0)}\\
\T{arg1 (last arg0)}\\
\T{arg1 (fromJust (listToMaybe arg0))}\\
\T{arg1 (fromJust Nothing)}\\
\T{(\$) arg1 (head arg0)}}}\\[.5em]
\textbf{fromFirstMaybes:\ \T{a -> [Maybe a] -> a}}\\[.5em]
\framebox{\parbox{\textwidth}{
 Demand Analysis\\[.25em]
\T{fromMaybe arg1 (head arg0)}\\
\T{fromMaybe arg1 (last arg0)}\\
\T{maybe arg1 id (head arg0)}\\
\T{maybe arg1 id (last arg0)}\\
\T{bool arg1 arg1 (null arg0)}}}\\[.5em]
\framebox{\parbox{\textwidth}{
 No Demand Analysis\\[.25em]
\T{fromRight arg1 (Left arg0)}\\
\T{fromLeft arg1 (Right arg0)}\\
\T{fromMaybe arg1 (head arg0)}\\
\T{fromMaybe arg1 (last arg0)}\\
\T{fromLeft arg1 (Right (catMaybes arg0))}}}\\[.5em]
\framebox{\parbox{\textwidth}{
 No Relevancy\\[.25em]
\T{fromLeft arg1 (Right arg0)}\\
\T{fromRight arg1 (Left arg0)}\\
\T{fromLeft arg1 (Right arg1)}\\
\T{fromRight arg1 (Right arg1)}\\
\T{fromLeft arg1 (Left arg1)}}}\\[.5em]
\textbf{mbAppFirst:\ \T{b -> (a -> b) -> [a] -> b}}\\[.5em]
\framebox{\parbox{\textwidth}{
 Demand Analysis\\[.25em]
\T{maybe arg2 arg1 (listToMaybe arg0)}\\
\T{fromMaybe arg2 (listToMaybe (map arg1 arg0))}\\
\T{fromLeft arg2 (Right (arg1 (head arg0)))}}}\\[.5em]
\framebox{\parbox{\textwidth}{
 No Demand Analysis\\[.25em]
\T{maybe arg2 arg1 (listToMaybe arg0)}\\
\T{fromLeft (arg1 (head arg0)) (Right arg2)}\\
\T{fromRight (arg1 (head arg0)) (Right arg2)}\\
\T{fromLeft (arg1 (head arg0)) (Left arg2)}\\
\T{fromRight (arg1 (head arg0)) (Left arg2)}}}\\[.5em]
\framebox{\parbox{\textwidth}{
 No Relevancy\\[.25em]
\T{arg1 (head arg0)}\\
\T{arg1 (last arg0)}\\
\T{fromLeft arg2 (Right arg1)}\\
\T{fromRight arg2 (Left arg1)}\\
\T{(\$) arg1 (head arg0)}}}\\[.5em]
\textbf{appendN:\ \T{Int -> [a] -> [a]}}\\[.5em]
\framebox{\parbox{\textwidth}{
 Demand Analysis\\[.25em]
\T{drop arg1 arg0}\\
\T{take arg1 arg0}\\
\T{drop arg1 (cycle arg0)}\\
\T{take arg1 (cycle arg0)}\\
\T{drop arg1 (init arg0)}}}\\[.5em]
\framebox{\parbox{\textwidth}{
 No Demand Analysis\\[.25em]
\T{drop arg1 arg0}\\
\T{take arg1 arg0}\\
\T{drop arg1 (cycle arg0)}\\
\T{take arg1 (cycle arg0)}\\
\T{drop arg1 (init arg0)}}}\\[.5em]
\framebox{\parbox{\textwidth}{
 No Relevancy\\[.25em]
\T{cycle arg0}\\
\T{init arg0}\\
\T{reverse arg0}\\
\T{tail arg0}\\
\T{drop arg1 arg0}}}\\[.5em]
\textbf{firstMaybe:\ \T{[Maybe a] -> a}}\\[.5em]
\framebox{\parbox{\textwidth}{
 Demand Analysis\\[.25em]
\T{fromJust (head arg0)}\\
\T{fromJust (last arg0)}\\
\T{head (catMaybes arg0)}\\
\T{last (catMaybes arg0)}\\
\T{fromJust (head (cycle arg0))}}}\\[.5em]
\framebox{\parbox{\textwidth}{
 No Demand Analysis\\[.25em]
\T{fromJust (head arg0)}\\
\T{fromJust (last arg0)}\\
\T{head (catMaybes arg0)}\\
\T{last (catMaybes arg0)}\\
\T{fromJust (head (cycle arg0))}}}\\[.5em]
\framebox{\parbox{\textwidth}{
 No Relevancy\\[.25em]
\T{head []}\\
\T{last []}\\
\T{fromJust (head arg0)}\\
\T{fromJust (last arg0)}\\
\T{fromJust Nothing}}}\\[.5em]
\textbf{mapEither:\ \T{(a -> Either b c) -> [a] -> ([b], [c])}}\\[.5em]
\framebox{\parbox{\textwidth}{
 Demand Analysis\\[.25em]
\T{partitionEithers (map arg1 arg0)}\\
\T{partitionEithers (repeat (arg1 (head arg0)))}\\
\T{partitionEithers (repeat (arg1 (last arg0)))}\\
\T{curry (last []) arg1 arg0}}}\\[.5em]
\framebox{\parbox{\textwidth}{
 No Demand Analysis\\[.25em]
\T{partitionEithers (map arg1 arg0)}\\
\T{partitionEithers (repeat (arg1 (head arg0)))}\\
\T{partitionEithers (repeat (arg1 (last arg0)))}\\
\T{curry (head []) arg1 arg0}\\
\T{curry (last []) arg1 arg0}}}\\[.5em]
\framebox{\parbox{\textwidth}{
 No Relevancy\\[.25em]
\T{partitionEithers (map arg1 arg0)}}}\\[.5em]
\textbf{head-rest:\ \T{[a] -> (a, [a])}}\\[.5em]
\framebox{\parbox{\textwidth}{
 Demand Analysis\\[.25em]
\T{fromJust (uncons arg0)}\\
\T{(,) (last arg0) arg0}\\
\T{(,) (head arg0) arg0}\\
\T{(,) (last arg0) []}\\
\T{(,) (head arg0) []}}}\\[.5em]
\framebox{\parbox{\textwidth}{
 No Demand Analysis\\[.25em]
\T{fromJust (uncons arg0)}\\
\T{(,) (last arg0) arg0}\\
\T{(,) (head arg0) arg0}\\
\T{(,) (last arg0) []}\\
\T{(,) (head arg0) []}}}\\[.5em]
\framebox{\parbox{\textwidth}{
 No Relevancy\\[.25em]
\T{head []}\\
\T{last []}\\
\T{fromJust (uncons arg0)}\\
\T{fromJust Nothing}\\
\T{(,) (head arg0) arg0}}}\\[.5em]
\textbf{applyPair:\ \T{(a -> b, a) -> b}}\\[.5em]
\framebox{\parbox{\textwidth}{
 Demand Analysis\\[.25em]
\T{uncurry (\$) arg0}\\
\T{uncurry id arg0}\\
\T{(\$) (fst arg0) (snd arg0)}}}\\[.5em]
\framebox{\parbox{\textwidth}{
 No Demand Analysis\\[.25em]
\T{uncurry (\$) arg0}\\
\T{uncurry id arg0}\\
\T{(\$) (fst arg0) (snd arg0)}\\
\T{uncurry (head []) arg0}\\
\T{(\$) (head []) arg0}}}\\[.5em]
\framebox{\parbox{\textwidth}{
 No Relevancy\\[.25em]
\T{uncurry id arg0}\\
\T{uncurry (\$) arg0}\\
\T{fromJust Nothing}\\
\T{head []}\\
\T{last []}}}\\[.5em]
\textbf{maybe:\ \T{Maybe a -> a -> Maybe a}}\\[.5em]
\framebox{\parbox{\textwidth}{
 Demand Analysis\\[.25em]
\T{Just (fromMaybe arg0 arg1)}\\
\T{Just (maybe arg0 id arg1)}\\
\T{listToMaybe (repeat (fromMaybe arg0 arg1))}\\
\T{curry (last []) arg1 arg0}\\
\T{curry (last []) arg0 arg1}}}\\[.5em]
\framebox{\parbox{\textwidth}{
 No Demand Analysis\\[.25em]
\T{Just (fromMaybe arg0 arg1)}\\
\T{fromRight arg1 (Left arg0)}\\
\T{fromLeft arg1 (Right arg0)}\\
\T{fromRight (Just arg0) (Left arg1)}\\
\T{fromLeft (Just arg0) (Left arg1)}}}\\[.5em]
\framebox{\parbox{\textwidth}{
 No Relevancy\\[.25em]
\T{Just arg0}\\
\T{Nothing}\\
\T{fromRight arg1 (Left arg0)}\\
\T{fromLeft arg1 (Right arg0)}\\
\T{(\$) id arg1}}}\\[.5em]
\textbf{multiAppPair:\ \T{(a -> b, a -> c) -> a -> (b, c)}}\\[.5em]
\framebox{\parbox{\textwidth}{
 Demand Analysis\\[.25em]
\T{(,) ((\$) (fst arg1) arg0) ((\$) (snd arg1) arg0)}\\
\T{curry (last []) arg0 arg1}}}\\[.5em]
\framebox{\parbox{\textwidth}{
 No Demand Analysis\\[.25em]
\T{(,) ((\$) (fst arg1) arg0) ((\$) (snd arg1) arg0)}\\
\T{curry (fromJust Nothing) arg0 arg1}\\
\T{curry (head []) arg0 arg1}\\
\T{curry (last []) arg0 arg1}}}\\[.5em]
\framebox{\parbox{\textwidth}{
 No Relevancy\\[.25em]
\T{fromJust Nothing}\\
\T{head []}\\
\T{last []}}}\\[.5em]
\textbf{applyNtimes:\ \T{(a -> a) -> a -> Int -> a}}\\[.5em]
\framebox{\parbox{\textwidth}{
 Demand Analysis\\[.25em]
\T{arg2 ((!!) (repeat arg1) arg0)}\\
\T{arg2 (head (replicate arg0 arg1))}\\
\T{arg2 (last (replicate arg0 arg1))}\\
\T{head (replicate arg0 (arg2 arg1))}\\
\T{last (replicate arg0 (arg2 arg1))}}}\\[.5em]
\framebox{\parbox{\textwidth}{
 No Demand Analysis\\[.25em]
\T{arg2 (fromRight arg1 (Left arg0))}\\
\T{arg2 (fromLeft arg1 (Right arg0))}\\
\T{fromRight (arg2 arg1) (Left arg0)}\\
\T{fromLeft (arg2 arg1) (Right arg0)}\\
\T{arg2 ((!!) (repeat arg1) arg0)}}}\\[.5em]
\framebox{\parbox{\textwidth}{
 No Relevancy\\[.25em]
\T{arg2 arg1}\\
\T{arg2 (arg2 arg1)}\\
\T{(\$) arg2 arg1}\\
\T{fromMaybe (arg2 arg1) (Just (arg2 arg1))}\\
\T{fromMaybe (arg2 arg1) Nothing}}}\\[.5em]
\textbf{splitAtFirst:\ \T{a -> [a] -> ([a], [a])}}\\[.5em]
\framebox{\parbox{\textwidth}{
 Demand Analysis\\[.25em]
\T{(,) arg0 (repeat arg1)}\\
\T{swap ((,) arg0 (repeat arg1))}\\
\T{splitAt (length arg0) (repeat arg1)}\\
\T{(,) (repeat arg1) (cycle arg0)}\\
\T{(,) (repeat arg1) (init arg0)}}}\\[.5em]
\framebox{\parbox{\textwidth}{
 No Demand Analysis\\[.25em]
\T{(,) arg0 (repeat arg1)}\\
\T{swap ((,) arg0 (repeat arg1))}\\
\T{splitAt (length arg0) (repeat arg1)}\\
\T{(,) (repeat arg1) (cycle arg0)}\\
\T{(,) (repeat arg1) (init arg0)}}}\\[.5em]
\framebox{\parbox{\textwidth}{
 No Relevancy\\[.25em]
\T{(,) arg0 arg0}\\
\T{(,) (cycle arg0) (cycle arg0)}\\
\T{(,) (init arg0) (init arg0)}\\
\T{(,) (reverse arg0) (reverse arg0)}\\
\T{(,) (tail arg0) (tail arg0)}}}\\[.5em]
\textbf{eitherTriple:\ \T{Either a b -> Either a b -> Either a b}}\\[.5em]
\framebox{\parbox{\textwidth}{
 Demand Analysis\\[.25em]
\T{bool arg1 arg1 (isRight arg0)}\\
\T{bool arg1 arg1 (isLeft arg0)}\\
\T{fromMaybe arg1 (listToMaybe (repeat arg0))}\\
\T{bool arg1 arg0 (last [])}\\
\T{bool arg1 arg0 (and [])}}}\\[.5em]
\framebox{\parbox{\textwidth}{
 No Demand Analysis\\[.25em]
\T{fromLeft arg1 (Left arg0)}\\
\T{fromRight arg1 (Left arg0)}\\
\T{bool arg1 arg0 False}\\
\T{bool arg1 arg0 True}\\
\T{bool arg1 arg0 otherwise}}}\\[.5em]
\framebox{\parbox{\textwidth}{
 No Relevancy\\[.25em]
\T{fromMaybe arg1 Nothing}\\
\T{bool arg1 arg1 False}\\
\T{bool arg1 arg1 True}\\
\T{bool arg1 arg1 otherwise}\\
\T{fromLeft arg1 (Right arg1)}}}\\[.5em]
\textbf{multiApp:\ \T{(a -> b -> c) -> (a -> b) -> a -> c}}\\[.5em]
\framebox{\parbox{\textwidth}{
 Demand Analysis\\[.25em]
\T{arg2 arg0 (arg1 arg0)}\\
\T{arg2 (fromJust Nothing) (arg1 arg0)}\\
\T{arg2 (head []) (arg1 arg0)}\\
\T{arg2 (last []) (arg1 arg0)}\\
\T{arg2 arg0 (arg1 (fromJust Nothing))}}}\\[.5em]
\framebox{\parbox{\textwidth}{
 No Demand Analysis\\[.25em]
\T{arg2 arg0 (arg1 arg0)}\\
\T{arg2 arg0 ((\$) arg1 arg0)}\\
\T{arg2 (fromJust Nothing) (arg1 arg0)}\\
\T{arg2 (head []) (arg1 arg0)}\\
\T{arg2 (last []) (arg1 arg0)}}}\\[.5em]
\framebox{\parbox{\textwidth}{
 No Relevancy\\[.25em]
\T{arg2 arg0 (arg1 arg0)}\\
\T{arg2 arg0 ((\$) arg1 arg0)}\\
\T{arg2 (snd ((,) (arg1 arg0) arg0)) (fst ((,) (arg1 arg0) arg0))}\\
\T{uncurry arg2 ((,) arg0 (arg1 arg0))}\\
\T{arg2 (snd ((,) arg0 arg0)) (arg1 (fst ((,) arg0 arg0)))}}}\\[.5em]
\textbf{mergeEither:\ \T{Either a (Either a b) -> Either a b}}\\[.5em]
\framebox{\parbox{\textwidth}{
 Demand Analysis\\[.25em]
\T{fromRight (fromJust Nothing) arg0}\\
\T{fromRight (head []) arg0}\\
\T{fromRight (last []) arg0}}}\\[.5em]
\framebox{\parbox{\textwidth}{
 No Demand Analysis\\[.25em]
\T{fromRight (fromJust Nothing) arg0}\\
\T{fromRight (head []) arg0}\\
\T{fromRight (last []) arg0}}}\\[.5em]
\framebox{\parbox{\textwidth}{
 No Relevancy\\[.25em]
\T{fromJust Nothing}\\
\T{head []}\\
\T{last []}\\
\T{fromRight (fromJust Nothing) arg0}\\
\T{fromRight (head []) arg0}}}\\[.5em]
\textbf{firstRight:\ \T{[Either a b] -> Either a b}}\\[.5em]
\framebox{\parbox{\textwidth}{
 Demand Analysis\\[.25em]
\T{head arg0}\\
\T{last arg0}\\
\T{head (cycle arg0)}\\
\T{last (cycle arg0)}\\
\T{head (init arg0)}}}\\[.5em]
\framebox{\parbox{\textwidth}{
 No Demand Analysis\\[.25em]
\T{head arg0}\\
\T{last arg0}\\
\T{head (cycle arg0)}\\
\T{last (cycle arg0)}\\
\T{head (init arg0)}}}\\[.5em]
\framebox{\parbox{\textwidth}{
 No Relevancy\\[.25em]
\T{head arg0}\\
\T{last arg0}\\
\T{head (cycle arg0)}\\
\T{last (cycle arg0)}\\
\T{head (init arg0)}}}\\[.5em]
\textbf{flatten:\ \T{[[[a]]] -> [a]}}\\[.5em]
\framebox{\parbox{\textwidth}{
 Demand Analysis\\[.25em]
\T{head (head arg0)}\\
\T{last (head arg0)}\\
\T{concat (head arg0)}\\
\T{head (last arg0)}\\
\T{last (last arg0)}}}\\[.5em]
\framebox{\parbox{\textwidth}{
 No Demand Analysis\\[.25em]
\T{head (head arg0)}\\
\T{last (head arg0)}\\
\T{concat (head arg0)}\\
\T{head (last arg0)}\\
\T{last (last arg0)}}}\\[.5em]
\framebox{\parbox{\textwidth}{
 No Relevancy\\[.25em]
\T{[]}\\
\T{lefts []}\\
\T{rights []}\\
\T{catMaybes []}\\
\T{cycle []}}}\\[.5em]
\textbf{pipe:\ \T{[(a -> a)] -> (a -> a)}}\\[.5em]
\framebox{\parbox{\textwidth}{
 Demand Analysis\\[.25em]
\T{foldr (\$) arg0 arg1}\\
\T{(\$) (head arg1) arg0}\\
\T{(\$) (last arg1) arg0}\\
\T{foldr id arg0 arg1}\\
\T{foldr id arg0 (cycle arg1)}}}\\[.5em]
\framebox{\parbox{\textwidth}{
 No Demand Analysis\\[.25em]
\T{foldr (\$) arg0 arg1}\\
\T{fromLeft arg0 (Right arg1)}\\
\T{(\$) (head arg1) arg0}\\
\T{(\$) (last arg1) arg0}\\
\T{foldr id arg0 arg1}}}\\[.5em]
\framebox{\parbox{\textwidth}{
 No Relevancy\\[.25em]
\T{fromLeft arg0 (Right arg0)}\\
\T{fromRight arg0 (Right arg0)}\\
\T{fromLeft arg0 (Left arg0)}\\
\T{fromRight arg0 (Left arg0)}\\
\T{fromMaybe arg0 (Just arg0)}}}\\[.5em]
\textbf{firstKey:\ \T{[(a,b)] -> a}}\\[.5em]
\framebox{\parbox{\textwidth}{
 Demand Analysis\\[.25em]
\T{(!!) [] (length arg0)}\\
\T{(\$) (last []) arg0}}}\\[.5em]
\framebox{\parbox{\textwidth}{
 No Demand Analysis\\[.25em]
\T{(!!) [] (length arg0)}\\
\T{(\$) (last []) arg0}\\
\T{(\$) (head []) arg0}\\
\T{(\$) (fromJust Nothing) arg0}}}\\[.5em]
\framebox{\parbox{\textwidth}{
 No Relevancy\\[.25em]
\T{fromJust Nothing}\\
\T{head []}\\
\T{last []}}}\\[.5em]
\textbf{splitStr:\ \T{String -> Char -> [String]}}\\[.5em]
\framebox{\parbox{\textwidth}{
 Demand Analysis\\[.25em]
\T{repeat (showChar arg0 arg1)}\\
\T{repeat ((:) arg0 arg1)}\\
\T{(:) (repeat arg0) (repeat arg1)}\\
\T{repeat ((++) arg1 (repeat arg0))}\\
\T{repeat (showString arg1 (repeat arg0))}}}\\[.5em]
\framebox{\parbox{\textwidth}{
 No Demand Analysis\\[.25em]
\T{repeat (showChar arg0 arg1)}\\
\T{repeat ((:) arg0 arg1)}\\
\T{(:) (repeat arg0) (repeat arg1)}\\
\T{repeat ((++) arg1 (repeat arg0))}\\
\T{repeat (showString arg1 (repeat arg0))}}}\\[.5em]
\framebox{\parbox{\textwidth}{
 No Relevancy\\[.25em]
\T{repeat arg1}\\
\T{[]}\\
\T{repeat (cycle arg1)}\\
\T{repeat (init arg1)}\\
\T{repeat (reverse arg1)}}}\\[.5em]
\textbf{areEq:\ \T{Eq a => a -> a -> Maybe a}}\\[.5em]
\framebox{\parbox{\textwidth}{
 Demand Analysis\\[.25em]
\T{fromMaybe (Just arg1) (lookup arg0 [])}\\
\T{Just (fromMaybe arg1 (lookup arg0 []))}}}\\[.5em]
\framebox{\parbox{\textwidth}{
 No Demand Analysis\\[.25em]
\T{fromMaybe (Just arg1) (lookup arg0 [])}\\
\T{Just (fromMaybe arg1 (lookup arg0 []))}\\
\T{fromLeft (Just arg1) (Right ((,) arg0))}\\
\T{fromRight (Just arg1) (Left ((,) arg0))}\\
\T{fromLeft (Just arg1) (Right ((,) arg0 ))}}}\\[.5em]
\framebox{\parbox{\textwidth}{
 No Relevancy\\[.25em]
\T{Just arg1}\\
\T{Nothing}\\
\T{listToMaybe (repeat arg1)}\\
\T{lookup arg1 []}\\
\T{fromJust Nothing}}}\\[.5em]
\textbf{lookup:\ \T{Eq a => [(a,b)] -> a -> b}}\\[.5em]
\framebox{\parbox{\textwidth}{
 Demand Analysis\\[.25em]
\T{fromJust (lookup arg0 arg1)}\\
\T{head (maybeToList (lookup arg0 arg1))}\\
\T{last (maybeToList (lookup arg0 arg1))}\\
\T{fromJust (lookup arg0 (cycle arg1))}\\
\T{fromJust (lookup arg0 (init arg1))}}}\\[.5em]
\framebox{\parbox{\textwidth}{
 No Demand Analysis\\[.25em]
\T{fromJust (lookup arg0 arg1)}\\
\T{head (maybeToList (lookup arg0 arg1))}\\
\T{last (maybeToList (lookup arg0 arg1))}}}\\[.5em]
\framebox{\parbox{\textwidth}{
 No Relevancy\\[.25em]
\T{fromJust (lookup arg0 arg1)}\\
\T{fromJust Nothing}\\
\T{head []}\\
\T{last []}}}\\[.5em]
\textbf{map:\ \T{(a -> b) -> [a] -> [b]}}\\[.5em]
\framebox{\parbox{\textwidth}{
 Demand Analysis\\[.25em]
\T{map arg1 arg0}\\
\T{repeat (arg1 (last arg0))}\\
\T{map arg1 (cycle arg0)}\\
\T{map arg1 (init arg0)}\\
\T{map arg1 (reverse arg0)}}}\\[.5em]
\framebox{\parbox{\textwidth}{
 No Demand Analysis\\[.25em]
\T{map arg1 arg0}\\
\T{repeat (arg1 (last arg0))}\\
\T{map arg1 (cycle arg0)}\\
\T{map arg1 (init arg0)}\\
\T{map arg1 (reverse arg0)}}}\\[.5em]
\framebox{\parbox{\textwidth}{
 No Relevancy\\[.25em]
\T{map arg1 arg0}\\
\T{lefts []}\\
\T{rights []}\\
\T{catMaybes []}\\
\T{concat []}}}\\[.5em]
\textbf{resolveEither:\ \T{Either a b -> (a->b) -> b}}\\[.5em]
\framebox{\parbox{\textwidth}{
 Demand Analysis\\[.25em]
\T{either arg0 id arg1}\\
\T{arg0 (head (lefts (repeat arg1)))}\\
\T{arg0 (last (lefts (repeat arg1)))}\\
\T{either arg0 (head []) arg1}\\
\T{either arg0 (last []) arg1}}}\\[.5em]
\framebox{\parbox{\textwidth}{
 No Demand Analysis\\[.25em]
\T{either arg0 id arg1}\\
\T{either arg0 (fromJust Nothing) arg1}}}\\[.5em]
\framebox{\parbox{\textwidth}{
 No Relevancy\\[.25em]
\T{either arg0 id arg1}\\
\T{arg0 (fromJust Nothing)}\\
\T{arg0 (head [])}\\
\T{arg0 (last [])}\\
\T{fromRight (arg0 (fromJust Nothing)) arg1}}}\\[.5em]
\textbf{firstMatch:\ \T{[a] -> (a -> Bool) -> a}}\\[.5em]
\framebox{\parbox{\textwidth}{
 Demand Analysis\\[.25em]
\T{last (dropWhile arg0 arg1)}\\
\T{head (dropWhile arg0 arg1)}\\
\T{last (filter arg0 arg1)}\\
\T{head (filter arg0 arg1)}\\
\T{last (takeWhile arg0 arg1)}}}\\[.5em]
\framebox{\parbox{\textwidth}{
 No Demand Analysis\\[.25em]
\T{last (dropWhile arg0 arg1)}\\
\T{head (dropWhile arg0 arg1)}\\
\T{last (filter arg0 arg1)}\\
\T{head (filter arg0 arg1)}\\
\T{last (takeWhile arg0 arg1)}}}\\[.5em]
\framebox{\parbox{\textwidth}{
 No Relevancy\\[.25em]
\T{head (dropWhile arg0 arg1)}\\
\T{last (dropWhile arg0 arg1)}\\
\T{head (filter arg0 arg1)}\\
\T{last (filter arg0 arg1)}\\
\T{head (takeWhile arg0 arg1)}}}\\[.5em]
\textbf{test:\ \T{Bool -> a -> Maybe a}}\\[.5em]
\framebox{\parbox{\textwidth}{
 Demand Analysis\\[.25em]
\T{bool (Just arg0) Nothing arg1}\\
\T{bool Nothing (Just arg0) arg1}\\
\T{bool (Just arg0) (Just arg0) arg1}\\
\T{Just (bool arg0 arg0 arg1)}\\
\T{curry (last []) arg0 arg1}}}\\[.5em]
\framebox{\parbox{\textwidth}{
 No Demand Analysis\\[.25em]
\T{bool (Just arg0) Nothing arg1}\\
\T{bool Nothing (Just arg0) arg1}\\
\T{Just (fromLeft arg0 (Right arg1))}\\
\T{Just (fromRight arg0 (Left arg1))}\\
\T{bool (Just arg0) (Just arg0) arg1}}}\\[.5em]
\framebox{\parbox{\textwidth}{
 No Relevancy\\[.25em]
\T{Just arg0}\\
\T{Nothing}\\
\T{Just (bool arg0 arg0 arg1)}\\
\T{listToMaybe (repeat arg0)}\\
\T{listToMaybe []}}}\\[.5em]
\textbf{intToBS:\ \T{Int64 -> ByteString}}\\[.5em]
\framebox{\parbox{\textwidth}{
 Demand Analysis\\[.25em]
\T{drop arg0 empty}\\
\T{take arg0 empty}\\
\T{toLazyByteString (int64Dec arg0)}\\
\T{toLazyByteString (int64HexFixed arg0)}\\
\T{toLazyByteString (int64LE arg0)}}}\\[.5em]
\framebox{\parbox{\textwidth}{
 No Demand Analysis\\[.25em]
\T{drop arg0 empty}\\
\T{take arg0 empty}\\
\T{toLazyByteString (int64Dec arg0)}\\
\T{toLazyByteString (int64HexFixed arg0)}\\
\T{toLazyByteString (int64LE arg0)}}}\\[.5em]
\framebox{\parbox{\textwidth}{
 No Relevancy\\[.25em]
\T{empty}\\
\T{concat []}\\
\T{fromChunks []}\\
\T{pack []}\\
\T{drop arg0 empty}}}\\[.5em]
\textbf{repl-funcs:\ \T{(a -> b) -> Int -> [a -> b]}}\\[.5em]
\framebox{\parbox{\textwidth}{
 Demand Analysis\\[.25em]
\T{replicate arg0 arg1}\\
\T{cycle (replicate arg0 arg1)}\\
\T{init (replicate arg0 arg1)}\\
\T{reverse (replicate arg0 arg1)}\\
\T{tail (replicate arg0 arg1)}}}\\[.5em]
\framebox{\parbox{\textwidth}{
 No Demand Analysis\\[.25em]
\T{replicate arg0 arg1}\\
\T{cycle (replicate arg0 arg1)}\\
\T{init (replicate arg0 arg1)}\\
\T{reverse (replicate arg0 arg1)}\\
\T{tail (replicate arg0 arg1)}}}\\[.5em]
\framebox{\parbox{\textwidth}{
 No Relevancy\\[.25em]
\T{repeat arg1}\\
\T{replicate arg0 arg1}\\
\T{lefts []}\\
\T{rights []}\\
\T{catMaybes []}}}\\[.5em]
\textbf{mapMaybes:\ \T{(a -> Maybe b) -> [a] -> Maybe b}}\\[.5em]
\framebox{\parbox{\textwidth}{
 Demand Analysis\\[.25em]
\T{arg1 (head arg0)}\\
\T{arg1 (last arg0)}\\
\T{(\$) arg1 (head arg0)}\\
\T{(\$) arg1 (last arg0)}\\
\T{arg1 (fromJust (listToMaybe arg0))}}}\\[.5em]
\framebox{\parbox{\textwidth}{
 No Demand Analysis\\[.25em]
\T{arg1 (head arg0)}\\
\T{arg1 (last arg0)}\\
\T{(\$) arg1 (head arg0)}\\
\T{(\$) arg1 (last arg0)}\\
\T{arg1 (fromJust (listToMaybe arg0))}}}\\[.5em]
\framebox{\parbox{\textwidth}{
 No Relevancy\\[.25em]
\T{arg1 (head arg0)}\\
\T{arg1 (last arg0)}\\
\T{arg1 ((!!) arg0 (length arg0))}\\
\T{arg1 (head [])}\\
\T{arg1 (last [])}}}\\[.5em]
\textbf{takeNdropM:\ \T{Int -> Int -> [a] -> ([a], [a])}}\\[.5em]
\framebox{\parbox{\textwidth}{
 Demand Analysis\\[.25em]
\T{splitAt arg2 (drop arg1 arg0)}\\
\T{splitAt arg2 (take arg1 arg0)}\\
\T{splitAt arg2 (take arg1 (cycle arg0))}\\
\T{splitAt arg2 (drop arg1 (cycle arg0))}\\
\T{splitAt arg2 (take arg1 (init arg0))}}}\\[.5em]
\framebox{\parbox{\textwidth}{
 No Demand Analysis\\[.25em]
\T{splitAt arg2 (drop arg1 arg0)}\\
\T{splitAt arg2 (take arg1 arg0)}\\
\T{splitAt arg2 (take arg1 (cycle arg0))}\\
\T{splitAt arg2 (drop arg1 (cycle arg0))}\\
\T{splitAt arg2 (take arg1 (init arg0))}}}\\[.5em]
\framebox{\parbox{\textwidth}{
 No Relevancy\\[.25em]
\T{splitAt arg2 arg0}\\
\T{(,) arg0 arg0}\\
\T{(,) (cycle arg0) (cycle arg0)}\\
\T{(,) (init arg0) (init arg0)}\\
\T{(,) (reverse arg0) (reverse arg0)}}}\\[.5em]
\textbf{cartProduct:\ \T{[a] -> [b] -> [[(a,b)]]}}\\[.5em]
\framebox{\parbox{\textwidth}{
 Demand Analysis\\[.25em]
\T{repeat (zip arg1 arg0)}\\
\T{repeat (cycle (zip arg1 arg0))}\\
\T{repeat (init (zip arg1 arg0))}\\
\T{repeat (reverse (zip arg1 arg0))}\\
\T{repeat (tail (zip arg1 arg0))}}}\\[.5em]
\framebox{\parbox{\textwidth}{
 No Demand Analysis\\[.25em]
\T{repeat (zip arg1 arg0)}\\
\T{repeat (cycle (zip arg1 arg0))}\\
\T{repeat (init (zip arg1 arg0))}\\
\T{repeat (reverse (zip arg1 arg0))}\\
\T{repeat (tail (zip arg1 arg0))}}}\\[.5em]
\framebox{\parbox{\textwidth}{
 No Relevancy\\[.25em]
\T{[]}\\
\T{repeat (zip arg1 arg0)}\\
\T{fromJust Nothing}\\
\T{maybeToList Nothing}\\
\T{lefts []}}}\\[.5em]
\textbf{hoogle02:\ \T{b -> (a -> b) -> [a] -> b}}\\[.5em]
\framebox{\parbox{\textwidth}{
 Demand Analysis\\[.25em]
\T{maybe arg2 arg1 (listToMaybe arg0)}\\
\T{fromMaybe arg2 (listToMaybe (map arg1 arg0))}}}\\[.5em]
\framebox{\parbox{\textwidth}{
 No Demand Analysis\\[.25em]
\T{maybe arg2 arg1 (listToMaybe arg0)}\\
\T{fromLeft (arg1 (head arg0)) (Right arg2)}\\
\T{fromRight (arg1 (head arg0)) (Right arg2)}\\
\T{fromLeft (arg1 (head arg0)) (Left arg2)}\\
\T{fromRight (arg1 (head arg0)) (Left arg2)}}}\\[.5em]
\framebox{\parbox{\textwidth}{
 No Relevancy\\[.25em]
\T{arg1 (head arg0)}\\
\T{arg1 (last arg0)}\\
\T{fromLeft arg2 (Right arg1)}\\
\T{fromRight arg2 (Left arg1)}\\
\T{(\$) arg1 (head arg0)}}}\\[.5em]
\textbf{containsEdge:\ \T{[Int] -> (Int,Int) -> Bool}}\\[.5em]
\framebox{\parbox{\textwidth}{
 Demand Analysis\\[.25em]
\T{null (repeat ((,) arg0 arg1))}\\
\T{null (replicate (length arg1) arg0)}\\
\T{null (replicate (head arg1) arg0)}\\
\T{null (replicate (last arg1) arg0)}\\
\T{null (repeat ((,) arg1 arg0))}}}\\[.5em]
\framebox{\parbox{\textwidth}{
 No Demand Analysis\\[.25em]
\T{null (fromLeft arg1 (Right arg0))}\\
\T{null (fromRight arg1 (Left arg0))}\\
\T{isLeft (Right ((,) arg0 arg1))}\\
\T{isRight (Right ((,) arg0 arg1))}\\
\T{isLeft (Left ((,) arg0 arg1))}}}\\[.5em]
\framebox{\parbox{\textwidth}{
 No Relevancy\\[.25em]
\T{False}\\
\T{True}\\
\T{otherwise}\\
\T{null arg1}\\
\T{isJust Nothing}}}\\[.5em]
\textbf{app3:\ \T{(a -> b -> c -> d) -> a -> c -> b -> d}}\\[.5em]
\framebox{\parbox{\textwidth}{
 Demand Analysis\\[.25em]
\T{arg3 arg2 arg0 arg1}\\
\T{fromMaybe (arg3 arg2 arg0 arg1) Nothing}\\
\T{arg3 (fst ((,) arg2 arg0)) (snd ((,) arg2 arg0)) arg1}\\
\T{arg3 (snd ((,) arg0 arg2)) (fst ((,) arg0 arg2)) arg1}\\
\T{arg3 arg2 arg0 (fromMaybe arg1 Nothing)}}}\\[.5em]
\framebox{\parbox{\textwidth}{
 No Demand Analysis\\[.25em]
\T{arg3 arg2 arg0 arg1}\\
\T{fromMaybe (arg3 arg2 arg0 arg1) Nothing}\\
\T{arg3 (fst ((,) arg2 arg0)) (snd ((,) arg2 arg0)) arg1}\\
\T{arg3 (snd ((,) arg0 arg2)) (fst ((,) arg0 arg2)) arg1}\\
\T{arg3 arg2 arg0 (fromMaybe arg1 Nothing)}}}\\[.5em]
\framebox{\parbox{\textwidth}{
 No Relevancy\\[.25em]
\T{arg3 arg2 arg0 arg1}\\
\T{fromMaybe (arg3 arg2 arg0 arg1) Nothing}\\
\T{fromLeft (arg3 arg2 arg0 arg1) (Right arg2)}\\
\T{fromRight (arg3 arg2 arg0 arg1) (Left arg2)}\\
\T{fromLeft (arg3 arg2 arg0 arg1) (Right arg1)}}}\\[.5em]
\textbf{indexesOf:\ \T{([(a,Int)] -> [(a,Int)]) -> [a] -> [Int] -> [Int]}}\\[.5em]
\framebox{\parbox{\textwidth}{
 No Demand Analysis\\[.25em]
\T{fromLeft arg0 (Right ((,) arg1 arg2))}\\
\T{fromRight arg0 (Left ((,) arg1 arg2))}\\
\T{fromLeft arg0 (Right ((,) arg2 arg1))}\\
\T{fromRight arg0 (Left ((,) arg2 arg1))}}}\\[.5em]
\framebox{\parbox{\textwidth}{
 No Relevancy\\[.25em]
\T{lefts []}\\
\T{rights []}\\
\T{catMaybes []}\\
\T{concat []}\\
\T{cycle []}}}\\[.5em]
\textbf{both:\ \T{(a -> b) -> (a, a) -> (b, b)}}\\[.5em]
\framebox{\parbox{\textwidth}{
 Demand Analysis\\[.25em]
\T{(,) (arg1 (snd arg0)) (arg1 (fst arg0))}}}\\[.5em]
\framebox{\parbox{\textwidth}{
 No Demand Analysis\\[.25em]
\T{(,) (arg1 (snd arg0)) (arg1 (fst arg0))}}}\\[.5em]
\framebox{\parbox{\textwidth}{
 No Relevancy\\[.25em]
\T{head []}\\
\T{last []}\\
\T{head (maybeToList Nothing)}\\
\T{last (maybeToList Nothing)}\\
\T{head (lefts [])}}}\\[.5em]
\textbf{zipWithResult:\ \T{(a -> b) -> [a] -> [(a, b)]}}\\[.5em]
\framebox{\parbox{\textwidth}{
 Demand Analysis\\[.25em]
\T{zip arg0 (map arg1 [])}\\
\T{zip arg0 (map arg1 arg0)}}}\\[.5em]
\framebox{\parbox{\textwidth}{
 No Demand Analysis\\[.25em]
\T{zip arg0 (map arg1 [])}\\
\T{zip arg0 (map arg1 arg0)}}}\\[.5em]
\framebox{\parbox{\textwidth}{
 No Relevancy\\[.25em]
\T{lefts []}\\
\T{rights []}\\
\T{catMaybes []}\\
\T{concat []}\\
\T{cycle []}}}\\[.5em]
\textbf{rights:\ \T{[Either a b] -> Either a [b]}}\\[.5em]
\framebox{\parbox{\textwidth}{
 Demand Analysis\\[.25em]
\T{Right (rights arg0)}\\
\T{(!!) [] (length arg0)}}}\\[.5em]
\framebox{\parbox{\textwidth}{
 No Demand Analysis\\[.25em]
\T{Right (rights arg0)}\\
\T{(!!) [] (length arg0)}}}\\[.5em]
\framebox{\parbox{\textwidth}{
 No Relevancy\\[.25em]
\T{head []}\\
\T{last []}\\
\T{Right []}\\
\T{fromJust Nothing}\\
\T{Right (rights arg0)}}}\\[.5em]
\textbf{mbToEither:\ \T{Maybe a -> b -> Either a b}}\\[.5em]
\framebox{\parbox{\textwidth}{
 Demand Analysis\\[.25em]
\T{curry (last []) arg1 arg0}\\
\T{curry (last []) arg0 arg1}}}\\[.5em]
\framebox{\parbox{\textwidth}{
 No Demand Analysis\\[.25em]
\T{fromRight (Right arg0) (Left arg1)}\\
\T{fromLeft (Right arg0) (Right arg1)}\\
\T{Right (fromRight arg0 (Left arg1))}\\
\T{Right (fromLeft arg0 (Right arg1))}\\
\T{curry (fromJust Nothing) arg1 arg0}}}\\[.5em]
\framebox{\parbox{\textwidth}{
 No Relevancy\\[.25em]
\T{Right arg0}\\
\T{fromJust Nothing}\\
\T{Left (fromJust arg1)}\\
\T{head []}\\
\T{last []}}}\\[.5em]
\textbf{singleList:\ \T{Int -> [Int]}}\\[.5em]
\framebox{\parbox{\textwidth}{
 Demand Analysis\\[.25em]
\T{repeat arg0}\\
\T{cycle (repeat arg0)}\\
\T{init (repeat arg0)}\\
\T{reverse (repeat arg0)}\\
\T{tail (repeat arg0)}}}\\[.5em]
\framebox{\parbox{\textwidth}{
 No Demand Analysis\\[.25em]
\T{repeat arg0}\\
\T{cycle (repeat arg0)}\\
\T{init (repeat arg0)}\\
\T{reverse (repeat arg0)}\\
\T{tail (repeat arg0)}}}\\[.5em]
\framebox{\parbox{\textwidth}{
 No Relevancy\\[.25em]
\T{replicate arg0 arg0}\\
\T{repeat arg0}\\
\T{[]}\\
\T{(:) arg0 []}\\
\T{iterate id arg0}}}\\[.5em]
\textbf{head-tail:\ \T{[a] -> (a,a)}}\\[.5em]
\framebox{\parbox{\textwidth}{
 Demand Analysis\\[.25em]
\T{(,) (last arg0) (last arg0)}\\
\T{(,) (head arg0) (head arg0)}\\
\T{last (zip [] arg0)}\\
\T{head (zip [] arg0)}\\
\T{(\$) (last []) arg0}}}\\[.5em]
\framebox{\parbox{\textwidth}{
 No Demand Analysis\\[.25em]
\T{(,) (last arg0) (last arg0)}\\
\T{(,) (head arg0) (head arg0)}\\
\T{last (zip [] arg0)}\\
\T{head (zip [] arg0)}\\
\T{(\$) (last []) arg0}}}\\[.5em]
\framebox{\parbox{\textwidth}{
 No Relevancy\\[.25em]
\T{fromJust Nothing}\\
\T{head []}\\
\T{last []}\\
\T{head (zip arg0 arg0)}\\
\T{last (zip arg0 arg0)}}}\\[.5em]
\textbf{2partApp:\ \T{(a -> b) -> (b -> c) -> [a] -> [c]}}\\[.5em]
\framebox{\parbox{\textwidth}{
 Demand Analysis\\[.25em]
\T{map arg1 (map arg2 arg0)}\\
\T{repeat (arg1 (arg2 (last arg0)))}\\
\T{repeat (arg1 (arg2 (head arg0)))}\\
\T{iterate id (arg1 (arg2 (last arg0)))}\\
\T{iterate' id (arg1 (arg2 (last arg0)))}}}\\[.5em]
\framebox{\parbox{\textwidth}{
 No Demand Analysis\\[.25em]
\T{map arg1 (map arg2 arg0)}\\
\T{repeat (arg1 (arg2 (last arg0)))}\\
\T{repeat (arg1 (arg2 (head arg0)))}\\
\T{repeat (arg1 ((\$) arg2 (last arg0)))}\\
\T{repeat (arg1 ((\$) arg2 (head arg0)))}}}\\[.5em]
\framebox{\parbox{\textwidth}{
 No Relevancy\\[.25em]
\T{fromJust (fromJust Nothing)}\\
\T{head (fromJust Nothing)}\\
\T{last (fromJust Nothing)}}}\\[.5em]